\documentclass[11pt,reqno]{amsart}

\usepackage[utf8]{inputenc}
\usepackage{amsfonts,latexsym,amssymb,amsmath,amsthm,amsrefs,etex,slashed,cancel}
\usepackage{mathtools}
\usepackage{todonotes}
\usepackage{hyperref}
\usepackage{a4wide}
\usepackage{enumerate}
\usepackage{mathrsfs}
\let\originallhook\lhook
\usepackage{MnSymbol}
\let\lhook\originallhook
\usepackage{nicefrac}

\usepackage{nicematrix}
\usepackage{tikz}

\renewcommand{\Re}{\operatorname{Re}}
\renewcommand{\Im}{\operatorname{Im}}

\DeclareMathOperator{\supp}{supp}
\newcommand{\der}{\mathrm{d}}
\newcommand{\rmi}{\mathrm{i}}

\newcommand{\tr}{\mathrm{tr}}
\newcommand{\Tr}{\mathrm{tr}}

\newcommand{\M}{\mathbb{M}}
\DeclarePairedDelimiterX{\norm}[1]{\lVert}{\rVert}{#1}

\newcommand{\R}{\mathbb{R}}
\newcommand{\C}{\mathbb{C}}
\newcommand{\N}{\mathbb{N}}

\newcommand{\comp}{\mathrm{c}}
\newcommand{\compp}{\mathrm{comp}}

\newcommand{\loc}{\mathrm{loc}}

\newcommand{\DN}{Q}
\newcommand{\CLP}{K}
\newcommand{\tm}{\Delta_T}
\newcommand{\Nuc}{\mathscr{I}_1}
\newcommand{\HS}{\mathscr{I}_2}

\newcommand{\diag}{\mathrm{diag}}

\newcommand{\gammaDir}[1]{\gamma_D^{#1}}
\newcommand{\gammaNeu}[1]{\gamma_N^{#1}}
\newcommand{\gammaCau}[1]{\gamma^{#1}}

\numberwithin{equation}{section}

\newtheorem{theorem}{Theorem}[section]
\newtheorem{definition}[theorem]{Definition}

\newtheorem{lemma}[theorem]{Lemma}

\newtheorem{corollary}[theorem]{Corollary}
\newtheorem{proposition}[theorem]{Proposition}

\newtheorem{rem}[theorem]{Remark}

\usetikzlibrary{calc}

\def\tikzeps   {8}        
\def\tikzR     {3}        
\def\tikzrmin  {.6}       
\def\tikzdelta {.15}      

\def\gSect    {gray!30}  
\def\gOverlap {gray!60}  
\def\arcstyle {thin}     

\newcommand*\fillwedge[3]{%
  \path[fill=#1] (0,0) -- (#2:\tikzR) arc (#2:#3:\tikzR) -- cycle;
}

\title[Relative trace formulas]{Relative trace formulas for obstacle scattering with Neumann and transmission boundary conditions}

\author[A. Hofmann]{Arne Hofmann}
\address{Institute for Analysis, Leibniz University Hannover, Welfengarten 1, D-30167 Hannover} \email{arne.hofmann@math.uni-hannover.de}

\author[A. Strohmaier]{Alexander Strohmaier}
\address{Institute for Analysis, Leibniz University Hannover, Welfengarten 1, D-30167 Hannover} \email{a.strohmaier@math.uni-hannover.de}

\begin{document}

\begin{abstract}
 We consider the case of scattering by several obstacles in $\mathbb{R}^d$ for $d \geq 2$.
 We establish a relative trace formula for Neumann and transmission boundary conditions analogous to the one obtained in \cite{HSW} for Dirichlet boundary conditions.

 In the case of $f(x) = x^{1/2}$ the trace has the interpretation of the Casimir energy of the obstacle configuration.
 In the one-dimensional case, we recover a rigorous version of the Lifshitz formula for the Casimir energy of parallel plates with frequency-independent electric permittivity and magnetic permeability.
 We thereby strengthen the mathematical foundations of the Casimir effect and demonstrate the flexibility of the rigorous approach established in \cite{Fang2021AMA} and \cite{HSW}.
\end{abstract}

\maketitle
\setcounter{tocdepth}{1}
\tableofcontents

\section{Introduction}

\subsection{Overview.}
Trace formulae are an important tool to study spectra of self-adjoint differential operators on manifolds.
On compact manifolds one typically studies traces of functions of the Laplace operator $f(-\Delta)$, which are trace-class if the function $f$ is decaying sufficiently fast near $\infty$. In the noncompact setting the spectrum of the Laplace operator is no longer discrete and such functions of the Laplacian are not trace-class. Scattering theory in this case builds on forming operator differences. An instructive abstract example is the Birman--Krein trace formula. 
Let $A$ and $A_0$ be self-adjoint operators, bounded below, and assume that for some $k\in \N$ and some $z\in \mathbb C\setminus \mathbb R$ one has
\begin{align}
    \label{eq:trace-class-hypothesis-birman-krein}
    (A-z)^{-k}-(A_0-z)^{-k}\in \Nuc(L^2(\mathbb R^d)),
\end{align}
where $\Nuc$ denotes the trace class. 
Under this hypothesis one can define a spectral shift function $\xi(\lambda)$ for the pair $(A,A_0)$, and for suitable functions $f$ the Birman--Krein trace formula gives
\[
    \tr\bigl(f(A)-f(A_0)\bigr)
    =
    \int_{\mathbb R} f'(\lambda)\,\xi(\lambda)\,d\lambda .
\]
In the scattering setting, the same spectral shift function is related to the determinant of the scattering matrix by the Birman--Krein formula
\[
    \det S(\lambda)=e^{-2\pi \rmi\, \xi(\lambda)}.
\]
The admissible class of functions $f$ depends on the resolvent power appearing in the trace-class hypothesis (\ref{eq:trace-class-hypothesis-birman-krein}), but always requires decay at infinity.

For applications in mathematical physics, in particular in the context of quantum field theories and Casimir energies, one is also interested in traces associated to positive powers of the Laplace operator, $(-\Delta)^\alpha, \; \alpha>0$, the case $\alpha=\frac12$ being of particular interest in the context of the Casimir effect. The Birman-Krein formula does not apply in this situation.

In \cite{HSW}, a trace formula was proved for a \emph{relative operator} associated with a configuration of obstacles satisfying Dirichlet boundary conditions. 
The relative construction subtracts not only the free operator, but also the contributions of the individual obstacles. 
These additional cancellations make it possible to treat functions that are not necessarily bounded, including functions relevant to Casimir energy.
The trace is computed in terms of a function $\Xi$ which is a close relative of the spectral shift function and the scattering phase (and of the relative determinant studied in \cite{carron1999determinant}),
and can be expressed as a Fredholm determinant of an operator built from the single layer operators of the obstacles.
As a special case, applying the trace formula to the function $f(x) = x^{1/2}$ yields the Casimir energy of the configuration.\\

Trace-formula representations have a long history at the
formal level in the physics literature on Casimir energies of obstacle
configurations. An early functional-integral treatment, in which the vacuum energy is expressed as a logarithmic determinant of a boundary operator, can for example be found already in \cite{bordagRobaschikWieczorek1985}.
Around 2006--2009 this was given a scattering-theoretic
interpretation by Kenneth and Klich \cite{kenneth06,kennethKlich2008} and
by Emig, Graham, Jaffe, Kardar and Rahi \cite{emig2007,emig2008casimir,%
rahiEmig2009}, who wrote the Casimir energy of a configuration of compact
objects as an integral over a determinant built from the $T$-matrices of the individual
objects and the free propagator between them. A survey of numerical methods exploiting representations
of this kind can be found in \cite{johnson2011numerical}. The function
$\Xi$ in our trace formula plays the role of this physics-side determinant,
with two differences: we work with boundary layer operators (which are
well adapted to Lipschitz boundaries) and all convergence statements are proved rigorously in trace-class norm
rather than at the level of formal manipulations of functional integrals.\\

In the present paper we extend these theorems by treating obstacles with Neumann or transmission boundary conditions.
In particular the transmission case requires new machinery. The function $\Xi$ is defined as a
Fredholm determinant on a $\lambda$-dependent subspace bundle
rather than on a fixed Sobolev space, and
the boundary integral system involves an operator
built from Calderon projectors at two distinct spectral parameters, intertwined
by the transmission matrix. The treatment of Lipschitz boundaries
in arbitrary dimension, and in particular the analysis at $\lambda=0$ in
dimension two, requires Hahn-meromorphic Fredholm theory rather than ordinary holomorphic
Fredholm theory. A further feature is that in the Neumann
case the relevant layer operators are singular at $\lambda=0$.

As a special case, the formula for the transmission case allows us to compute the Casimir energy for materials with a dielectric constant which is independent of the spectral parameter.
Our methods are general enough to handle obstacles with Lipschitz boundaries, and therefore apply to the polyhedral obstacles that occur in practical NEMS (nanoelectromechanical systems) devices where Casimir forces are relevant.
As a special case we recover the Lifshitz formula for dielectric half spaces that was given in \cite{miltonCasimirEnergyDispersion2010}.

We emphasize that our formula is proved completely rigorously, and does not rely on undefined path integrals.
We also give a more general formula for bounded dielectric obstacles.

\subsection{Setting and main results.}
We fix once and for all the geometric setting for the rest of the article.
Let $\Omega_1, \ldots, \Omega_M$ be a finite collection of open, bounded, connected Lipschitz domains in $\mathbb{R}^d$, $d\geq 2$, with $\operatorname{dist}(\Omega_i, \Omega_j) \geq \delta > 0$ for all $i\neq j$.
We write $\Omega := \bigsqcup_{i=1}^M \Omega_i$ for the full obstacle arrangement.
We also call $\Omega^- := \Omega$ the \emph{interior domain}.
The \emph{exterior domain} is $\Omega^+ := \mathbb{R}^d \setminus \overline{\Omega}$.
We have the same distinction of interior and exterior for each individual obstacle, so that $\Omega_i^+ := \mathbb{R}^d \setminus \overline{\Omega}_i$.

One can then consider the Laplace operator $\Delta_D$ obtained by imposing Dirichlet boundary conditions on $\partial \Omega$, as an unbounded self-adjoint operator in the Hilbert space $L^2(\R^d)$. Similarly, one constructs the operators $\Delta_{D,i}$ by imposing Dirichlet boundary conditions on $\partial \Omega_i$ only.
Finally we denote by $\Delta_0$ the Laplace operator on $\R^d$ without boundary conditions.
In this paper, $\Delta=\sum_{i=1}^d \partial_i^2$ (so $-\Delta$ is the \emph{positive} Laplacian).

The focus of our analysis is on linear operators of the form
$$
 R_{s,D} = (-\Delta_{D}+m^2)^s - (-\Delta_{0}+m^2)^s - \sum_{i=1}^M \left[ (-\Delta_{D,i}+m^2)^s - (-\Delta_{0}+m^2)^s\right].
$$
A priori, these are densely defined operators with domain containing $C_0^\infty(\mathbb{R}^d\setminus \partial \Omega)$ for any $s \in \mathbb{C}$.

For smooth domains it was shown in \cite{HSW} that for any $s \in \C$ with $\Re(s)>0$ and any $m > 0$, $R_{s, D}$ is bounded and the operator closure of $R_{s,D}$ is a trace-class operator on $L^2(\R^d)$.

In the case $s=-1$ we obtain the \emph{relative resolvent} at $\lambda = \rmi m$, i.e. $R_{-1, D} = R_{\mathrm{rel}, D}(
\rmi m)$ where
\begin{align*}
R_{\mathrm{rel}, D}(\lambda) = (-\Delta_{D}+m^2)^{-1} - (-\Delta_{0}+m^2)^{-1} - \sum_{i=1}^M \left[ (-\Delta_{D,i}+m^2)^{-1} - (-\Delta_{0}+m^2)^{-1}\right].
\end{align*}

The trace of $R_{s,D}$ can be computed as
$$
\tr(R_{s,D}) = \frac{2 s}{\pi} \sin (\pi s) \int_m^{\infty} \lambda\left(\lambda^2-m^2\right)^{s-1} \Xi_D(\rmi \lambda) \mathrm{d} \lambda.
$$
Here $\Xi_D(i\lambda)$ is a continuous function on $[0,\infty)$, which decays exponentially as $\lambda \to \infty$.
It is defined as a Fredholm determinant of an operator obtained from the single layer $S_\lambda$ operator on $\partial \Omega$.
This allows boundary element methods to be applied to compute operator traces. A particular example is the Casimir energy of the configuration $\Omega= \Omega_1 \cup \Omega_2$, thought of as two compact objects placed in space. One can define this Casimir energy as $\frac{1}{2} \tr R_{\frac{1}{2},D}$ and it was shown in \cite{Fang2021AMA} for the case of smooth domains that the force computed from this energy is indeed always the same as the force computed from the renormalised stress energy tensor. This therefore allows to compute Casimir forces via boundary integral operators.

The aim of this article is to establish such formulas for boundary conditions other than Dirichlet boundary conditions.

Our main results are stated in terms of the boundary layer operators, to be defined in Section \ref{sec:boundary-layer-ops}. We denote by $S_\lambda$ the single layer operator, $D_\lambda$ and $D_\lambda'$ are the double layer operator and its adjoint. The hypersingular operator will be denoted by $N_\lambda$.

Our first theorem states the main result of \cite{HSW} for Lipschitz domains.

\begin{theorem} \label{thm:dirichlet-trace-formula}
    Assume $\Re(s)>0$. Then
 the operator $R_{s,D}$ extends continuously to a trace-class operator in $L^2(\R^d)$ with trace
\begin{align*}
\tr(R_{s,D}) = \frac{2 s}{\pi} \sin (\pi s) \int_m^{\infty} \lambda\left(\lambda^2-m^2\right)^{s-1} \Xi_D(\rmi \lambda) \mathrm{d} \lambda.
\end{align*}
where 
$$
  \Xi_D(\zeta) = \log \det \left( S_\zeta (S_{\textrm{diag},\zeta})^{-1} \right), \qquad \operatorname{Im}(\zeta) > 0.
$$
\end{theorem}
In the above the operator $S_{\textrm{diag},\lambda}$ is the diagonal part of the single layer operator 
$S_\lambda: H^{-\frac{1}{2}}(\partial \Omega) \to H^{\frac{1}{2}}(\partial \Omega)$ 
on $\partial \Omega$ with respect to the decomposition $H^s(\partial \Omega) = \bigoplus_{i=1}^M H^s(\partial \Omega_i)$,  for $-1 \leq s \leq 1$.

We will not go through the proof of Theorem \ref{thm:dirichlet-trace-formula} in detail.
The proof of the corresponding theorem for Neumann boundary conditions (Theorem \ref{thm:neumann-trace-formula}) will establish all the analytical properties required to see that \cite{HSW} can be adapted to transfer the result from smooth to Lipschitz boundaries.

We now formulate the main results of this article.
Replacing Dirichlet boundary conditions with Neumann boundary conditions we obtain the Neumann analogue of the relative operator
$$
R_{s,N} = (-\Delta_{N}+m^2)^s - (-\Delta_0 + m^2)^{s} - \sum_{i=1}^M\left[(-\Delta_{N,i}+m^2)^s - (-\Delta_{0}+m^2)^s\right].
$$
\begin{theorem} \label{thm:neumann-trace-formula}
Assume $\Re(s)>0$. Then
 the operator $R_{s,N}$ extends continuously to a trace-class operator in $L^2(\R^d)$ with trace
 $$
 \mathrm{tr}(R_{s,N}) = \frac{2s}{\pi} \sin(\pi s) \int_m^\infty \lambda\left(\lambda^2-m^2\right)^{s-1} \Xi_N(\rmi \lambda) \mathrm{d} \lambda.
$$ 
where 
$$
  \Xi_N(\zeta) = \log \det \left( N_\zeta (N_{\textrm{diag},\zeta})^{-1} \right), \qquad \operatorname{Im}(\zeta) > 0.
$$
\end{theorem}
Unlike in the case of Dirichlet boundary conditions there is the extra complication that $(N_{\textrm{diag},\lambda})^{-1}$ is not regular at $\lambda=0$.
However, it turns out that the singularities of $N_{\mathrm{diag}, \lambda}^{-1}$ are in the kernel of $N_\lambda$, and therefore $N_\lambda N_{\mathrm{diag}, \lambda}^{-1}$ is regular at $\lambda = 0$.
This resembles the parallel situation for the Maxwell problem treated in \cite{strohmaierRelativeTraceFormula2021}.

As shown for the Dirichlet case in \cite{HSW}, the theorem follows from a holomorphic functional calculus argument and the computation of the trace of the relative resolvent.

\begin{theorem}
The relative Neumann resolvent 
\begin{align*}
R_{\mathrm{rel}, N}(\lambda) = (-\Delta_{N}-\lambda^2)^{-1} - (-\Delta_{0}-\lambda^2)^{-1} - \sum_{i=1}^M\left[ (-\Delta_{N,i}-\lambda^2)^{-1} - (-\Delta_{0}-\lambda^2)^{-1}\right].
\end{align*} is trace class and
\begin{align}
\tr R_{\mathrm{rel}, N}(\lambda) = -\frac{1}{2\lambda}\Xi_N'(\lambda).
\end{align}
\end{theorem}

We also give a formula for transmission boundary conditions, which model penetrable obstacles, where the incident wave is partially reflected and partially transmitted through the boundary.
The transmission boundary conditions ensure conservation of energy despite these jumps.
The transmission problem corresponds, in acoustics, to waves propagating in an ambient fluid containing bubbles of some fluid or gas.
In electromagnetics it corresponds to electromagnetic waves propagating at an interface between media of differing dielectric constants.

The definition of the relative resolvent is more complicated in the transmission case, because the presence or absence of an obstacle $\Omega_i$ changes the behaviour of waves not only at the boundary $\partial \Omega_i$ but in the whole domain $\Omega_i$.
In Section \ref{sec:transmission-resolvent} we will introduce a self-adjoint operator $\Delta_{T, i}$ on $L^2(\mathbb{R}^d)$ which acts like $\kappa_+^2\Delta$ on the exterior domain $\Omega_i^+$ but acts like $\kappa_-^2 \Delta$ on $\Omega_i$, where $\kappa_{\pm} > 0$ are fixed parameters encoding the speed of wave propagation in the exterior domain $\Omega^+$ and the interior domain $\Omega^-$.
In the absence of any obstacle, it reduces to the operator $\kappa_+^2 \Delta_0$.

Therefore we consider the relative resolvent
\begin{align}
  \label{eq:relative-resolvent-transmission}
R_{\mathrm{rel}, T}(\lambda) = (-\Delta_{T}-\lambda^2)^{-1} - (-\kappa_+^2\Delta_0 - \lambda^2)^{-1} - \sum_{i=1}^M 
\left[(-\Delta_{T,i}-\lambda^2)^{-1} -  (-\kappa_+^2\Delta_{0}-\lambda^2)^{-1}\right].
\end{align}
and the relative operators
\begin{align}
R_{s, T} = (-\Delta_{T}+m^2)^s - (-\kappa_+^2\Delta_0 + m^2)^{s} - \sum_{i=1}^M\left[(-\Delta_{T,i}+m^2)^s - (-\kappa_+^2\Delta_{0}+m^2)^s\right].
\end{align}
\begin{theorem}
The relative resolvent for transmission boundary conditions, (\ref{eq:relative-resolvent-transmission}) is trace class and satisfies
\begin{align}
\tr R_{\mathrm{rel}, T}(\lambda) = -\frac{1}{2\lambda}\Xi_T'(\lambda),
\end{align}
where $\Xi_T(\zeta) = \log \det \left( (H_\zeta (H_{\mathrm{diag},\zeta})^{-1})|_{\mathcal{B}_{\zeta/\kappa_+}^-} \right)$, 
and the operator $H_\zeta$ is given by
\begin{align}
 H_\zeta = P_{\zeta/\kappa_+}^+\M - \M P_{\zeta/\kappa_-}^-, \qquad \operatorname{Im}(\zeta) > 0.
\end{align}
The $P_\zeta^\pm$ are the Calderon projectors, $\M$
encodes the transmission boundary conditions, and $\mathcal{B}_{\zeta/\kappa_+}^- = \operatorname{ran} P_{\zeta/\kappa_+}^-$.
\end{theorem}

Establishing this result requires more care than the Neumann case.
For one thing, the definition of the function $\Xi_T$ now requires taking a Fredholm determinant in a family of subspaces $\mathcal{B}_{\zeta/\kappa_+}^-$.
We will see that the bundle can be holomorphically trivialised, and therefore the usual definition of the Fredholm determinant can be applied.

The same functional calculus computation as in the Dirichlet and Neumann cases yields the result:
\begin{theorem}
    \label{thm:transmission-trace-formula}
If $\operatorname{Re}(s) > 0$, the operator $R_{s, T}$ extends continuously to a trace-class operator on $L^2(\mathbb{R}^d)$ with trace
\begin{align*}
\mathrm{tr}(R_{s,T}) = \frac{2s}{\pi} \sin(\pi s) \int_m^\infty \lambda\left(\lambda^2-m^2\right)^{s-1} \Xi_T(\rmi \lambda) \mathrm{d} \lambda.
\end{align*}
\end{theorem}

\subsection{Notation and conventions.}
Given a domain $\Omega \subset \mathbb{R}^d$ we will take $\Omega_+$ to be the exterior
and $\Omega_-$ to be the interior domain.
This is the convention in \cite{mcleanStronglyEllipticSystems2000}.
If $\Omega$ has connected components $\Omega = \Omega_1 \sqcup \ldots \sqcup \Omega_n$, then we will decompose function spaces ``in the bulk'' $L^2(\mathbb{R}^d) \simeq L^2(\Omega^+) \oplus L^2(\Omega_1) \oplus \ldots \oplus L^2(\Omega_n)$, and on the boundary $L^2(\partial \Omega) \simeq L^2(\partial \Omega_1) \oplus \ldots \oplus L^2(\partial \Omega_n)$.
We will write $p_i: L^2(\mathbb{R}^d) \to L^2(\Omega_i)$ for the projections in the bulk and $q_i: L^2(\partial \Omega) \to L^2(\partial \Omega_i)$ for the projections on the boundary.
Moreover we have $p_- := 1 - p_+ = \sum_{i=1}^n p_i$.
We will abuse notation and use the same symbols, $p_i, q_i$ for decompositions of Sobolev spaces $H^s(\partial \Omega)$, $H^s_{(\mathrm{loc})}(\mathbb{R}^d)$ etc.

The use of Green identities (integration by parts) often makes it more convenient to use a bilinear integration pairing rather than a sesquilinear inner product.
The transpose of an operator $A$ with respect to the
bilinear integration pairing will be denoted $A^t$.
Thus if we consider the free resolvent $R(\lambda) = (-\Delta - \lambda^2)^{-1}$ on $L^2(\mathbb{R}^d)$, then
$R(\lambda)^t = R(\lambda)$, whereas $R(\lambda)^* = R(\overline{\lambda})$.
We also have $D_\lambda' = D_\lambda^t$; in the case of the double layer operator $D_\lambda$, the prime is conventional notation.

In estimates of the form $f(x) \leq C g(x)$, the letter $C$ denotes a generic constant which may vary from line to line. We will also write $f \lesssim g$ to denote that $f(x) \leq C g(x)$ for some constant $C > 0$ which does not depend on $x$.

We use integral signs to also denote bilinear pairings like $H^s \times H^{-s} \to \mathbb{C}$ when there is no danger of confusion.

We denote by $\mathcal{L}(X, Y)$ the space of bounded linear operators from a Banach space $X$ to a Banach space $Y$, and write $\mathcal{L}(X) = \mathcal{L}(X, X)$.

We write $\mathbb{C}^+ := \{z \in \mathbb{C} \mid \operatorname{Im}(z) > 0\}$ for the open upper half plane.

\section{Boundary layer operators.}
\label{sec:boundary-layer-ops}
Our analysis is based on the theory of Sobolev spaces on Lipschitz domains. 
For surveys of the theory the reader is referred to \cites{costabel1988, mcleanStronglyEllipticSystems2000,kirsch}.
One notable restriction relative to the smooth case is that the Sobolev spaces $H^s(\partial\Omega)$ are only intrinsically defined for $s \in [-1, 1]$,
and therefore we will only consider this range of Sobolev spaces on the boundary.

\subsection{Sobolev spaces and boundary traces.}
Throughout this section, let $\Omega$ be a Lipschitz domain in $\R^d$, $d \geq 2$.
That is, $\partial \Omega$ is, up to rotation, locally the graph of a Lipschitz function.
By Rademacher's theorem there exists an almost everywhere defined exterior normal vector field $\nu_x$.
We fix this vector field once and for all.

\begin{definition}
The Sobolev space of order $s \in \mathbb{R}$ on $\Omega$ is
\begin{align*}
H^s(\Omega) = \{ u = U|_{\Omega} \mid U \in H^s(\mathbb{R}^d)\},
\end{align*}
and inherits the Hilbert space structure from $H^s(\mathbb{R}^d)$.
\end{definition}

We want to define the restriction $u \mapsto u|_{\partial \Omega}$ for $u$ in some Sobolev space.
This requires defining Sobolev spaces on $\partial\Omega$.
The construction of the spaces $H^s(\partial\Omega)$ for $s \in [-1, 1]$ is given in \cite{mcleanStronglyEllipticSystems2000}, Chapter 3, Section ``Lipschitz Domains''.
Here we only cite their salient properties, which are directly analogous to the case of smooth domains.

\begin{lemma}[Theorem 3.37 \cite{mcleanStronglyEllipticSystems2000}]
The restriction maps
\begin{align*}
\gammaDir{\pm}: C^0(\overline{\Omega}^\pm) \to C^0(\partial \Omega), \qquad \gammaDir{\pm}(u) = u|_{\partial\Omega}
\end{align*}
extend uniquely to continuous linear maps
\begin{align*}
\gammaDir{\pm}: H^{s+\frac{1}{2}}_\loc(\Omega^{\pm}) \to H^s(\partial \Omega),\qquad s \in (0,1)
\end{align*}
with continuous right inverses
\begin{align*}
\eta_\pm: H^s(\partial \Omega) \to H^{s+\frac{1}{2}}(\Omega^\pm).
\end{align*}
\end{lemma}

The importance of the Dirichlet trace is seen in the fact that the Rellich embedding theorem is inherited by the boundary Sobolev spaces.

\begin{lemma}
Let $\Omega$ be a bounded Lipschitz domain.
\begin{enumerate}
\item Let $s, t \in (-1, 1)$ and $s < t$. Then $H^t(\partial \Omega) \subset H^s(\partial \Omega)$ and the inclusion is compact.
\end{enumerate}
\end{lemma}
\begin{proof}
The embedding $H^t(\partial \Omega) \subset H^s(\partial \Omega)$ is the composition $\gammaDir{+} \circ j \circ \eta_+$, where $j: H^{t+1/2}(K) \to H^{s+1/2}(K)$ is the compact embedding of the Sobolev spaces on a compact set $K$ containing $\Omega$.
\end{proof}

There is also a Neumann trace operator, which extends $u \mapsto \partial_\nu u|_{\partial \Omega}$.
It is defined weakly and requires the following definition.

\begin{definition}
The Banach space $H^{s + 1/2}_\Delta(\Omega)$, $s\in(0,1)$ is defined as
\begin{align*}
H^{s + 1/2}_\Delta(\Omega) = \{ u \in H^{s+1/2}(\Omega) \mid \Delta u \in L^2(\Omega)\}.
\end{align*}
with the norm
\begin{align*}
\norm{u}_{H^{s + 1/2}_\Delta(\Omega)}^2 = \norm{u}_{H^{s+1/2}(\Omega)}^2 + \norm{\Delta u}_{L^2(\Omega)}^2.
\end{align*}
\end{definition}

If $f \in H^{3/2-s}(\Omega)$, then $\gammaDir{\pm} f \in H^{-s+1}(\partial \Omega)$, and the Neumann trace $\gammaNeu{\pm} u \in H^{s -1 }(\partial\Omega)$ is defined by a formal integration by parts formula.
In order to get the sign conventions straight we state them explicitly.
For the interior trace we have:
\begin{align*}
 \int_{\partial \Omega} \gammaDir{-}(f) \gammaNeu{-}(u) \mathrm{d}\sigma = \int_\Omega f \Delta u \mathrm{d} x + \int_\Omega \nabla f \cdot \nabla u \mathrm{d}x.
\end{align*}

For the exterior Neumann trace $\gammaNeu{+} u$ we again use the vector field $\nu$ on the boundary, \emph{not} $-\nu$ (which would be the outward normal for $\Omega^+$).
Thus
\begin{align*}
-\int_{\partial \Omega} \gammaDir{+}(f) \gammaNeu{+}(u) \mathrm{d}\sigma = \int_{\Omega^+} f \Delta u \mathrm{d} x + \int_{\Omega^+} \nabla f \cdot \nabla u \mathrm{d}x
\end{align*}
for any $f \in H^{3/2-s}(\Omega^+)$ and any $u \in H^{s+1/2}_\Delta(\Omega^+)$.

\begin{lemma}
    \label{lem:continuity-neumann-trace}
The Neumann trace is a continuous linear map $\gammaNeu{\pm}: H^{s+1/2}_{\Delta, \loc}(\Omega) \to H^{s-1}(\partial \Omega)$ for $s \in (0,1)$.
If $u \in \mathscr{C}^\infty(\overline{\Omega}^\pm)$, then $\gammaNeu{\pm} u = \partial_\nu u|_{\partial\Omega}$.
We reiterate that $\nu$ is the normal vector field pointing into $\Omega^+$.
\end{lemma}
\begin{proof}
The case $s=1/2$ is Lemma 4.3 in \cite{mcleanStronglyEllipticSystems2000}.
The same proof goes through, mutatis mutandis, for $s\in (0,1)$.
\end{proof}

Our most important tool is integration by parts, in the form of the first, second and third Green identities.
We state the first two identities here, whereas the third Green identity will be given after the definition of the boundary layer potentials.

\begin{lemma}
\begin{enumerate}
\item \emph{First Green identity.} Suppose that $u \in H^1(\Omega)$ and $v \in H^1_\Delta(\Omega)$. Then
\begin{align*}
\int_\Omega u \Delta v = - \int_\Omega \nabla u \cdot \nabla v + \int_{\partial \Omega} (\gammaDir{}u) (\gammaNeu{}v).
\end{align*}
\item \emph{Second Green identity.} Suppose that $u, v \in H^1_\Delta(\Omega)$. Then
\begin{align*}
\int_\Omega \left[u \Delta v - v \Delta u\right] = \int_{\partial \Omega} \left[(\gammaDir{}u) (\gammaNeu{}v) - (\gammaDir{}v) (\gammaNeu{}u)\right].
\end{align*}
\end{enumerate}
\end{lemma}

\subsection{Free resolvent}
The Green's function $G_\lambda$ for the free Helmholtz equation $(-\Delta - \lambda^2)\phi=0$ in $\R^d$ is the distributional integral kernel of the resolvent $(-\Delta - \lambda^2)^{-1}$. This can be expressed explicitly as
\begin{align} \label{eqn:GreensFunction}
  G_{\lambda,0}(x,y) =  \frac{\rmi}{4} \left(  \frac{\lambda}{2 \pi |x-y|} \right)^{\frac{d-2}{2}} \mathrm{H}^{(1)}_{\frac{d-2}{2}}( \lambda |x-y|),
\end{align}
where $\mathrm{H}^{(1)}_{\alpha}$ denotes the Hankel function of the first kind of order $\alpha$.
In particular, in dimension three one has
$$
 G_{\lambda}(x,y) = \frac{1}{4 \pi} \frac{e^{\rmi \lambda |x-y|}}{|x-y|}.
$$
The above formula gives the resolvent kernel for $\Im(\lambda)>0$ but is defined on a larger subset of the complex plane.
This continuation of the resolvent defines a holomorphic family of continuous maps $H^{s}_\comp(\R^d) \to H^{s+2}_\loc(\R^d)$ on the logarithmic cover of $\C \setminus -\rmi [0,\infty)$.
Note that in case the dimension $d$ is even the above Hankel function fails to be analytic at zero. In that case
$$
 G_{\lambda} = \tilde G_{\lambda} + F_\lambda\, \lambda^{d-2} \log(\lambda),
$$
where $\tilde G_{\lambda}$ is an entire family of operators $H^s_{\compp}(\R^d) \to H^{s+2}_{\loc}(\R^d)$ and $F_\lambda$ is an entire family of operators with smooth integral kernel
$$
F_\lambda(x,y) = \frac{1}{2 \rmi} (2 \pi)^{-(d-1)}\int_{\mathbb{S}^{d-1}} e^{\rmi \lambda \theta \cdot (x-y)} d \theta,
$$
that is even in $\lambda$. Moreover we have that 
\begin{align*}
G_{-\lambda}(x,y)=\overline{G_{\lambda}(x,y)} \quad \text{for} \quad \lambda>0.
\end{align*}

\subsection{Boundary layer operators}
We now collect the standard properties of single and double layer potentials and their boundary traces on Lipschitz domains; the foundational references are \cite{costabel1988,MR769382}, and a textbook treatment is in \cite{mcleanStronglyEllipticSystems2000}.

For $f \in C^0(\partial \Omega)$, the single and double layer \emph{potentials} are defined as
\begin{align*}
 &\tilde S_\lambda f(x) = \int_{\partial \Omega} G_{\lambda}(x,y) f(y) d\sigma(y),\\
 &\tilde D_\lambda f(x) = \int_{\partial \Omega}  (\partial_{\nu,y} G_{\lambda}(x,y)) f(y) d\sigma(y),
\end{align*}
where $\partial_{\nu,y}$ denotes the outward normal derivative and $\sigma$ is the surface measure.

They extend to continuous maps
\begin{align*}
&\tilde S_\lambda: H^{-\frac{1}{2}}(\partial \Omega) \to H^{1}_\loc(\Omega_-) \oplus H^1_\loc(\Omega_+), \quad u \mapsto (-\Delta_0-\lambda^2)^{-1}(\gammaDir{})^* u\\
&\tilde D_\lambda: H^{\frac{1}{2}}(\partial \Omega) \to H^1_\loc(\Omega_-)\oplus H^1_\loc(\Omega_+),
\quad u \mapsto (-\Delta_0 - \lambda^2)^{-1} (\gammaNeu{})^* u.
\end{align*}

By taking boundary traces of the boundary layer potentials, we obtain the boundary layer operators.
The single layer operator is
$$S_\lambda: H^{-\frac{1}{2}}(\partial \Omega) \to H^{\frac{1}{2}}(\partial \Omega), u \mapsto \gammaDir{} \circ \tilde S_\lambda = \gammaDir{} \circ G_\lambda \circ \gammaDir{*},$$
and the double layer operator $D_\lambda$ is defined as
$$D_\lambda: H^{\frac{1}{2}}(\partial \Omega) \to H^{\frac{1}{2}}(\partial \Omega), u \mapsto 
\frac{1}{2}(\gammaDir{+} \circ \tilde D_\lambda + \gammaDir{-} \circ \tilde D_\lambda).$$
Its transpose is given by
$$D'_\lambda: H^{-\frac{1}{2}}(\partial \Omega) \to H^{-\frac{1}{2}}(\partial \Omega), u \mapsto  \frac{1}{2}(\gammaNeu{+} \tilde S_\lambda + \gammaNeu{-} \tilde S_\lambda),$$
which is well-defined as
$\tilde S_\lambda$ maps $H^{-\frac{1}{2}}(\partial \Omega)$ to $H^{s+\frac{1}{2}}_{\Delta,\loc}(\Omega_\pm)$.
Finally, the hypersingular operator $N_\lambda$ is defined as
$$
 N_\lambda: H^{\frac{1}{2}}(\partial \Omega) \to H^{-\frac{1}{2}}(\partial \Omega), \quad u \mapsto \gammaNeu{+} \tilde D_\lambda.
$$

One has the following ``jump relations'' for $u \in H^{-\frac{1}{2}}(\partial \Omega)$ and $v \in H^{\frac{1}{2}}(\partial \Omega)$:
\begin{align}
 &\gammaDir{\pm} \tilde S_\lambda u = S_\lambda u,\label{eq:jump-rel-1}\\
& \gammaDir{\pm} \tilde D_\lambda v = (\pm\frac{1}{2} + D_\lambda) v,\label{eq:jump-rel-2}\\
 &\gammaNeu{\pm} \tilde S_\lambda u =  (\mp\frac{1}{2} + D_\lambda') u,\label{eq:jump-rel-3}\\
 & \gammaNeu{\pm} \tilde D_\lambda v = N_\lambda v\label{eq:jump-rel-4}.
\end{align}

It is conventional to use a bracket notation for the jumps:
\begin{align*}
[\gamma_D u] = \gammaDir{+} u - \gammaDir{-} u, \quad
[\gamma_N u] = \gammaNeu{+} u - \gammaNeu{-} u.
\end{align*}
\begin{rem}
  In the definition of the boundary layer operators we follow the sign convention of \cite{HSW} and \cite{mcleanStronglyEllipticSystems2000}.
  In the numerical analysis literature one often finds different signs, such that the hypersingular operator in this literature corresponds to what we call $-N_\lambda$.
\end{rem}

We also record the following algebraic relations
\begin{align}
  S_\lambda N_\lambda &= -(\frac{1}{2} + D_\lambda) (\frac{1}{2} - D_\lambda),
  \label{eq:alg-rel-1}\\
  N_\lambda S_\lambda&= -(\frac{1}{2} + D_\lambda') (\frac{1}{2} - D_\lambda')
  \label{eq:alg-rel-2}.
\end{align}

We have the following result, sometimes called \emph{third Green identity}.
\begin{lemma}
    \label{lem:third-green}
Suppose that $u \in L^2(\mathbb{R}^d)$ is compactly supported, $u|_{\Omega^\pm} \in H^1(\Omega^\pm)$ and $f := (-\Delta - \lambda^2)u \in H^{-1}(\Omega^\pm)$.
Then
\begin{align*}
u = G_\lambda f + \widetilde{D}_\lambda [\gammaDir{} u] - \widetilde{S}_\lambda [\gammaNeu{} u].
\end{align*}
\end{lemma}

We will also need the interior and exterior Dirichlet-to-Neumann operators.
Conventionally, these are defined with respect to the \emph{outward-directed} normal vector field of the domain.
Concretely this means the following.
\begin{definition}
Let $u \in H^{\frac{1}{2}}(\partial \Omega)$ and let $\lambda \in \mathbb{C}^+$.
We define
\begin{enumerate}
\item $Q_\lambda^- u := \gammaNeu{-} v$ where $v \in H^1(\Omega)$ is the unique solution to $(-\Delta - \lambda^2) v = 0$, $\gammaDir{-} v = u$.
\item $Q_\lambda^+ u := - \gammaNeu{+} w$ where $w \in H^1(\Omega^+)$ is the unique solution to $(-\Delta - \lambda^2) w = 0$, $\gammaDir{+} w = u$.
\end{enumerate}
\end{definition}

With these conventions, one finds that the interior and exterior Dirichlet-to-Neumann operators are
\begin{align}
  \label{eq:DtN-rel}
   \DN_\lambda^\pm = S_\lambda^{-1}\left(\frac{1}{2} \mp D_\lambda\right)
\end{align}
and satisfy
\begin{align}
   \DN_\lambda^- + \DN_\lambda^+ = S_\lambda^{-1}.
\end{align}




\subsection{Hahn holomorphic functions.}
\label{subsec:hahn-holomorphic-ops}
The holomorphic dependence of the boundary layer operators on the spectral parameter $\lambda$ will play an important role, in particular in the use of analytic Fredholm theory.
In even dimensions analyticity at $\lambda=0$ fails because of $\log\lambda$ singularities. To take this into account we work instead with the weaker notion of \emph{Hahn holomorphy}
introduced in~\cite{MR3227433}. We will use the results on the class of these functions as a black box, but state the definition for the sake of completeness. The class of functions we consider are the $z \log z$ Hahn holomorphic functions and we will refer to these simply as Hahn-holomorphic.
A function defined near zero in a fixed sector of the logarithmic cover of $\C$ is Hahn holomorphic if it admits a normally convergent expansion of the form
\[
        A_\lambda \;=\; \sum_{(j,k) \in E} a_{j,k}\,\lambda^{j}(-\log\lambda)^{-k},
\]
with operator-valued coefficients $a_{j,k}$ and the index set $E$, depending on $A$, being a well-ordered subset of $\mathbb{Z} \times \mathbb{Z}$ equipped with the lexicographical ordering. 
In addition there must exist an $N \in \mathbb{N}_0$ such that all coefficients $a_{j,k}$ with $-k > N j$ vanish, and 
we must have $a_{j,k} = 0$ for $j < 0$ and for $j = 0$ and $k > 0$. 

Such functions form an integral domain and the quotient field 
is isomorphic to the field of Hahn-meromorphic functions. These Hahn meromorphic functions have similarly convergent expansions near zero without the requirement that the $a_{j,k}$ vanish when $j < 0$ or for $j=0, k > 0$. The theory resembles that of meromorphic functions in that bounded Hahn-meromorphic functions are Hahn-holomorphic and in that there exists 
a Hahn holomorphic version of the analytic Fredholm theorem. We refer to \cite{MR3227433}
for details where the Fredholm theorem is stated in Theorem 4.1. This allows us to use Fredholm theory near $\lambda=0$ in the even dimensional case in the same way as ordinary Fredholm theory in odd dimensions.

The following summarises the expansions of the boundary layer operators and their behaviour near $\lambda=0$ as straightforward consequences of the explicit description of the resolvent in terms of Hankel functions.

\begin{lemma}
Let $\mathbb{C}^+ := \{z \in \mathbb{C} \mid \operatorname{Im}(z) > 0\}$ be the open upper half plane and $\overline{\mathbb{C}^+} = \{z \in \mathbb{C} \mid \operatorname{Im}(z) \geq 0\}$ its closure.

Then
\begin{enumerate}
\item The functions
\begin{align*}
&S_\lambda: \mathbb{C}^+ \to \mathcal{L}(H^{-\frac{1}{2}}(\partial \Omega), H^{\frac{1}{2}}(\partial \Omega))\\
&D_\lambda: \mathbb{C}^+ \to \mathcal{L}(H^{\frac{1}{2}}(\partial \Omega), H^{\frac{1}{2}}(\partial \Omega))\\
&D'_\lambda: \mathbb{C}^+ \to \mathcal{L}(H^{-\frac{1}{2}}(\partial \Omega), H^{-\frac{1}{2}}(\partial \Omega))\\
&N_\lambda: \mathbb{C}^+ \to \mathcal{L}(H^{\frac{1}{2}}(\partial \Omega), H^{-\frac{1}{2}}(\partial \Omega))
\end{align*}
are holomorphic Banach space valued functions.
\item In odd dimensions $d\geq 3$, these operator valued functions are holomorphic in a neighbourhood of $\lambda = 0$.
In even dimension $d\geq 4$, they are Hahn holomorphic near $0$.
\item In dimension $d\geq 3$, the derivatives of all of these operator families are bounded in a neighbourhood of $\lambda = 0$.
\item In $d=2$, the operators $D_\lambda, D_\lambda', N_\lambda$ are Hahn holomorphic with Hahn holomorphic derivatives at $0$, whereas $S_\lambda = m_\lambda + r_\lambda \log \lambda$ with holomorphic $m_\lambda, r_\lambda$.
\end{enumerate}
\end{lemma}

The singularity structure of $S_\lambda$ in $d=2$ is investigated more thoroughly in Section \ref{sec:dimension-two}.

\subsection{Calderon projectors and Cauchy data.}
\begin{definition}
    Let $u_\pm \in H^1_{\Delta, \mathrm{loc}}(\overline{\Omega^\pm})$.
    Then the \emph{Cauchy data} of $u_\pm$ are defined to be the boundary traces
    $\gammaCau{\pm} u := \begin{pmatrix} \gammaDir{\pm} u\\ \gammaNeu{\pm}u\end{pmatrix} \in H^{1/2}(\partial \Omega) \oplus H^{-1/2}(\partial\Omega)$.

    We will write
    \begin{align*}
    \mathscr{C} := H^{1/2}(\partial \Omega) \oplus H^{-1/2}(\partial\Omega) 
    \end{align*}
    for the space of Cauchy data.
    It is a Banach space with the norm
    \begin{align*}
    \norm{\Phi}_{\mathscr{C}}^2 := \norm{\phi}_{H^{1/2}(\partial\Omega)}^2 + \norm{\psi}_{H^{-1/2}(\partial\Omega)}^2,\qquad \Phi = \begin{pmatrix} \phi \\ \psi\end{pmatrix} \in \mathscr{C}.
    \end{align*}
    We will generally use capital Greek letters for elements of $\mathscr{C}$.
\end{definition}

\begin{definition}
For $\operatorname{Im}(\lambda) > 0$ homogeneous solution spaces are defined as
\begin{align}
\mathcal{L}^\pm_\lambda := \{u \in H^1(\Omega^\pm) \mid (-\Delta - \lambda^2) u = 0\}.
\end{align}
\end{definition}

The following standard result is essentially a restatement of the second Green identity.

\begin{lemma}
The spaces $\mathcal{B}^\pm_\lambda := \gammaCau{\pm} \mathcal{L}^\pm_\lambda$
are isotropic subspaces for the symplectic form
\begin{align*}
\omega((\phi_1, \psi_1), (\phi_2, \psi_2)) = \int_{\partial\Omega} \left[\phi_1 \psi_2 - \phi_2 \psi_1\right].
\end{align*}
\end{lemma}




    \begin{definition}
        The Calderon projectors are the operators on $\mathscr{C}$ given by
        \begin{align*}
            P^\pm_\lambda
                =
                \begin{pmatrix} \frac{1}{2} \pm D_\lambda & \mp S_{\lambda}\\
                \pm N_{\lambda} & \frac{1}{2} \mp D_{\lambda}'
                \end{pmatrix}
                = \frac{1}{2} \mp A_\lambda.             
        \end{align*}
        The operator
        \begin{align*}
            A_\lambda = 
            \begin{pmatrix}
            -D_\lambda & S_\lambda\\ -N_\lambda & D_\lambda'
            \end{pmatrix}
        \end{align*}
        is called the \emph{multitrace operator}.
        \end{definition}

We will introduce the notation
\begin{align*}
    &\CLP_\lambda := \begin{pmatrix}
    \widetilde{D}_{\lambda} & -\widetilde{S}_{\lambda}
    \end{pmatrix}: \mathscr{C} \to H^1_{\text{loc}}(\mathbb{R}^d)\\
    &\CLP'_\lambda := \CLP_\lambda^t = \begin{pmatrix}
    \widetilde{S}_{\lambda}'\\ \widetilde{D}_{\lambda}'\end{pmatrix}:
    H^{-1}_{\mathrm{cpt}}(\mathbb{R}^d) \to \mathscr{C}
\end{align*}
This is short hand notation for
\begin{align*}
\CLP_\lambda\Phi &:=\widetilde D_\lambda\phi-\widetilde S_\lambda\psi\\
\CLP_\lambda' f &:= \begin{pmatrix} \widetilde{S}_{\lambda}' f \\ \widetilde{D}_{\lambda}' f \end{pmatrix}.
\end{align*}

\begin{lemma}
    \label{lem:calderon-projectors}
        The Calderon projectors satisfy the following identities:
        \begin{align}
            \label{eq:calderon-potential-trace-plus}
            &\gammaCau{+}
            \CLP_\lambda 
            = P_\lambda^+\\
            \label{eq:calderon-potential-trace-minus}
            &\gammaCau{-}
            \CLP_\lambda 
            = -P_\lambda^-\\
            &\CLP'_\lambda \CLP_\lambda 
            =
            \begin{pmatrix} \widetilde{S}_\lambda' \widetilde{D}_\lambda
            & - \widetilde{S}_\lambda' \widetilde{S}_\lambda\\
        \widetilde{D}_\lambda' \widetilde{D}_\lambda &
        - \widetilde{D}_\lambda' \widetilde{S}_\lambda\end{pmatrix}
        = -\frac{1}{2\lambda} \frac{d}{d\lambda} A_\lambda
        = -\frac{1}{2\lambda}\frac{d}{d\lambda} P_\lambda^-
        = \frac{1}{2\lambda}\frac{d}{d\lambda} P_\lambda^+.\label{eq:cald-potential-derivative}
            \end{align}
\end{lemma}  
\begin{proof}
   The identities (\ref{eq:calderon-potential-trace-plus}) and (\ref{eq:calderon-potential-trace-minus}) follow directly from the definition and the relations (\ref{eq:jump-rel-1} -- \ref{eq:jump-rel-4}).
    The identity (\ref{eq:cald-potential-derivative}) follows from boundary operator derivatives as for example stated in Lemma \ref{lem:bd-op-derivative}.
\end{proof}

The boundary operator $K_\lambda^t K_\lambda$ is formed by integrating over the whole space $\mathbb{R}^d$.
For the transmission problem, we also need to consider such integrals over the domains $\Omega^\pm$.

\begin{lemma}
Let $\chi_\pm = \chi_{\Omega_\pm}$ be the characteristic functions of the domains $\Omega_\pm$.
Then
\begin{align*}
&K_\lambda^t \chi_\pm K_\lambda
= -\frac{1}{2\lambda} \dot{A}_\lambda P^\pm_\lambda
\end{align*}
where $\dot{A}_\lambda = \frac{d}{d\lambda} A_\lambda$.
\end{lemma}

\begin{proof}
Let $\chi_\pm = \chi_{\Omega_\pm}$ be the characteristic functions of the domains $\Omega_\pm$.
Then
    \begin{align}
    \label{eq:calderon-vanishing-1}
    &\chi_+ K_\lambda P_\lambda^- = 0\\
    &\chi_- K_\lambda P_\lambda^+ = 0.
    \label{eq:calderon-vanishing-2}
    \end{align}
For, if $u_+ = \chi_+ K_\lambda P_\lambda^- \Phi$,
then $\gammaCau{+} u_+ = P_\lambda^+ P_\lambda^- \Phi = 0$.
Since $u_+$ solves the Helmholtz equation, it follows that $u_+ = 0$.

Now let $\Gamma^\pm_\lambda = K_\lambda^t \chi_\pm K_\lambda$.
Then by Lemma \ref{lem:calderon-projectors}, we have
\begin{align*}
  \Gamma^+_\lambda + \Gamma^-_\lambda = -\frac{1}{2\lambda} \dot{A}_\lambda.
\end{align*}
Moreover, by (\ref{eq:calderon-vanishing-1} -- \ref{eq:calderon-vanishing-2}), we have
$\Gamma^\pm_\lambda = \Gamma^\pm_\lambda P^\pm_\lambda$ and $\Gamma^\pm_\lambda P^\mp_\lambda = 0$.
Hence we have
\begin{align*}
    \Gamma^\pm_\lambda = -\frac{1}{2\lambda} \dot{A}_\lambda P^\pm_\lambda.
\end{align*}
\end{proof}

\begin{rem}
We call the operator $K_\lambda$ the \emph{Calderon potential operator}, because it plays a similar role to the single and double layer potential operators, namely of solving a given boundary value problem in terms of the boundary data.
Given boundary data $\Phi_{\pm} \in \mathcal{B}_{\pm}$, define $u_\pm := \pm K_\lambda \Phi_\pm$.
By (\ref{eq:calderon-potential-trace-plus}), (\ref{eq:calderon-potential-trace-minus}), $\gammaCau{\pm} u_\pm = P_\lambda^\pm \Phi_\pm = \Phi_\pm$.
\end{rem}

\begin{lemma}
The multitrace operator satisfies
\begin{align}
A_\lambda^2 = \frac14.
\end{align}
The Calderon projectors satisfy
\begin{align}
P_\lambda^+ P_\lambda^- = P_\lambda^- P_\lambda^+ = 0,
            \qquad (P_\lambda^\pm)^2 = P_\lambda^\pm.
\end{align}
Hence they are projections
\begin{align*}
P_{\lambda}^\pm: \mathscr{C} \to \mathcal{B}^\pm_\lambda
\end{align*}
and induce a decomposition of the Cauchy data space:
\begin{align}
\mathscr{C} = \mathcal{B}^+_\lambda \oplus \mathcal{B}^-_\lambda.
\end{align}
\end{lemma}
\begin{proof}
We have
\begin{align*}
\operatorname{ran} P_\lambda^{+} \subset \mathcal{B}_\lambda^{+}, \quad \operatorname{ran} P_\lambda^{-} \subset \mathcal{B}_\lambda^{-}
\end{align*}
by Equation (\ref{eq:calderon-potential-trace-plus}), (\ref{eq:calderon-potential-trace-minus}).

Conversely, let $\Phi  = (\phi, \psi) \in \mathcal{B}_\lambda^+$.
Then there exists $u_+ \in \mathcal{L}^+_\lambda$ with $\gammaCau{+} u_+ = \Phi$.
Define
\begin{align*}
u = \begin{cases}
u_+ & \text{in } \Omega^+\\
0 & \text{in } \Omega^-
\end{cases}
\end{align*}
Then the jump of $u$ is just $\gammaCau{+}u - \gammaCau{-} u = \Phi$.
The 3rd Green identity therefore gives
$\Phi = \gammaCau{+} K_\lambda \Phi = P_\lambda^+ \Phi$.

Similarly one shows $\operatorname{ran} P_\lambda^- \supset \mathcal{B}_\lambda^-$.

Suppose that \((f,g) \in \mathcal{B}^-_\lambda \cap \mathcal{B}^+_\lambda\).
Then there exist \(\psi_{\pm} \in \ker(\Delta_{\Omega^{\pm}} + \lambda^2)\) with
\((f,g) = (\gamma_0^{\pm}\psi_{\pm}, \gamma_1^{\pm}\psi_{\pm})\).
Then \(\psi = \psi_- + \psi_+ \in L^2(\mathbb{R}^d)\) has no Dirichlet or Neumann jumps at \(\partial\Omega\),
\([\gamma_0\psi] = 0 = [\gamma_1\psi]\).
By the third Green identity
(Lemma \ref{lem:third-green}),
\(\psi = 0\) and therefore \(f = 0 = g\).
\end{proof}


\section{Transmission problem.}

In the transmission problem the wavenumber is allowed to jump at the boundary between exterior and interior domains.
The boundary data are specified not in terms of the Dirichlet/Neumann traces themselves, but in terms of the jumps of the traces at the boundary.
\begin{definition}
Let $\kappa_+, \kappa_-, \nu_0, \nu_1 \in \mathbb{C}^{\times}$.
Let $\Omega \subset \mathbb{R}^d$ be a bounded Lipschitz domain.
We will say that $u = u_+ + u_- \in H^1(\Omega_-) \oplus H^1_{\mathrm{loc}}(\Omega^+)$ solves the transmission problem with transmission data $(\phi, \psi) \in H^{1/2}(\partial \Omega) \oplus H^{-1/2}(\partial \Omega)$ if
\begin{align}
    \label{eq:transmission-problem-def}
    \begin{cases}
&(-\Delta - (\lambda/\kappa_+)^2) u_+ = f_+/\kappa_+^2 \in H^{-1}(\Omega^+)\\
&(-\Delta - (\lambda/\kappa_-)^2) u_- = f_-/\kappa_-^2 \in H^{-1}(\Omega^-)\\
&\gamma_D^+ u - \nu_0\gamma_D^- u = \phi_0 \in H^{1/2}(\partial\Omega)\\
& \gamma_N^+u - \nu_1\gamma_N^- u = \psi_0 \in H^{-1/2}({\partial\Omega}).
\end{cases}
\end{align}
\end{definition}
\begin{rem}
By rescaling $u$ and $f$, we can always obtain an equivalent problem in which one of $\nu_0, \nu_1$ is equal to 1, and also such that one of $\kappa_+, \kappa_-$ is equal to 1.
\end{rem}
The homogeneous problem with $f_\pm = 0$ is of great importance.
We will reformulate this problem as a boundary integral equation using Calderon projectors.

Since the Calderon operators are projections onto the Cauchy data of solutions, solving (\ref{eq:transmission-problem-def}) with $f_\pm = 0$ is equivalent to solving
\begin{align}
\label{eq:transmission-bd-system-all}
P_{\lambda/\kappa_-}^+ \Phi_- = 0,\qquad P_{\lambda/\kappa_+}^- \Phi_+ = 0,\qquad \Phi_+ - \M \Phi_- = \Phi_0,
\end{align}
where $\Phi_0 = \begin{pmatrix} \phi_0 \\ \psi_0\end{pmatrix}$
and the matrix
\(
\M = \begin{pmatrix}
   \nu_0 & 0 \\ 0 & \nu_1
\end{pmatrix}
\)
encodes the transmission conditions.

These are six equations for four unknowns.
In \cite{costabelDirectBoundaryIntegral1985} it is shown that one can instead solve the following system of two equations for two unknowns:
\begin{align}
    \label{eq:transmission-bie}
H_\lambda \Phi = \Phi_0,
\end{align}
where the operator $H_\lambda$ is defined as
\begin{align}
\label{eq:defn-Hlambda-operator}
H_\lambda :=  P_{\lambda/\kappa_+}^+ \M - \M P_{\lambda/\kappa_-}^-= \M P_{\lambda/\kappa_-}^+ - P_{\lambda/\kappa_+}^- \M = - A_{\lambda/\kappa_+} \M - \M A_{\lambda/\kappa_-}.
\end{align}
\begin{lemma}
    We have the intertwining relation
\begin{align}
P_{\lambda / \kappa_+}^{\pm} H_\lambda = H_\lambda P_{\lambda / \kappa_-}^{\pm} \label{eq:calderon-intertwiner}
\end{align}
and therefore $H_\lambda$ maps $\mathcal{B}_{\lambda/\kappa_-}^\pm \to \mathcal{B}_{\lambda/\kappa_+}^\pm$.
\end{lemma}
\begin{proof}
    Follows immediately from \eqref{eq:defn-Hlambda-operator}.
\end{proof}

To make the relationship between the systems (\ref{eq:transmission-bd-system-all}) and (\ref{eq:transmission-bie}) more precise, we need to introduce the following assumptions.
\begin{definition}
We will say that \emph{Assumption A} is satisfied if the system
\begin{align}
P_{\lambda/\kappa_-}^+ \Phi_- = 0,\qquad P_{\lambda/\kappa_+}^- \Phi_+ = 0,\qquad \M \Phi_- = \Phi_+
\end{align}
only has the trivial solution
\begin{align*}
\Phi_- = \Phi_+ = 0.
\end{align*}

We will say that \emph{Assumption $\widetilde{A}$} is satisfied if the ``adjoint system''
\begin{align}
P_{\lambda/\kappa_+}^+ \widetilde{\Phi}_- = 0,\qquad P_{\lambda/\kappa_-}^- \widetilde{\Phi}_+ = 0,\qquad \widetilde{\Phi}_- = \M \widetilde{\Phi}_+
\end{align}
only has the trivial solution
\begin{align*}
\widetilde{\Phi}_- = \widetilde{\Phi}_+ = 0.
\end{align*}
\end{definition}

\begin{theorem}
\label{thm:transmission-problem-equivalence}
    \begin{itemize}
        \item[(i)] The system (\ref{eq:transmission-bd-system-all}) is equivalent to (\ref{eq:transmission-bie}) if and only if Assumption A holds.
Equivalence means that a pair $\Phi_-, \Phi_+ \in \mathscr{C}$ solves  (\ref{eq:transmission-bd-system-all}) if and only if $\Phi := \Phi_-$ solves (\ref{eq:transmission-bie}).
\item[(ii)] The operator $H_\lambda$ is a Fredholm operator of index zero which is bijective if and only if Assumption $A$ and $\widetilde{A}$ hold.
\item[(iii)] Assumptions $A$ and $\widetilde{A}$ are satisfied if
\begin{itemize}
\item[a)] $\kappa_+/\kappa_-, \nu_0/\nu_1 > 0$ and
\item[b)] $\operatorname{Im}(\lambda) > 0$ or $\lambda = 0$.
\end{itemize} 
\end{itemize}
\end{theorem}
\begin{proof}
These facts are all from \cite{costabelDirectBoundaryIntegral1985}.
Item (i) is Lemma 4.2 and the subsequent discussion.
Since this is purely algebraic, it goes through without change on Lipschitz domains.
Item (ii) is Corollary 4.5.
Item (iii) is Proposition 4.7. This is an integration by parts computation, and therefore also goes through without change on Lipschitz domains.

The Fredholm property of item (ii) relies on an a priori estimate which is satisfied for $\operatorname{Im}(\lambda) > 0$ by virtue of Proposition \ref{prop:cald-op-inv-bound}.
Moreover, by item (2) of Theorem \ref{thm:boundary-ops-fredholm}, $H_0 - H_i$ is compact, and therefore (ii) remains true at $\lambda = 0$.
The computation of the kernel is again purely algebraic and therefore also goes through without change for Lipschitz domains.
\end{proof}

\begin{rem}
The interested reader will find a catalog of more general conditions under which Assumptions A and $\widetilde{A}$ hold in Remark 4.8 of \cite{costabelDirectBoundaryIntegral1985}.

We will however be satisfied with the conditions given in Theorem \ref{thm:transmission-problem-equivalence} (iii), as these are sufficient for the applications we have in mind.
From now on, we will always assume $\kappa_\pm, \nu_0, \nu_1 > 0$.
\end{rem}

\subsection{Transmission resolvent.}
\label{sec:transmission-resolvent}
In the specific case $\nu_0 = 1$, $\nu_1 = (\kappa_-/\kappa_+)^2$, the transmission problem can be interpreted as defining the resolvent of an unbounded self-adjoint operator representing the \emph{kinetic energy} of the problem.

Define the quadratic form
\begin{align*}
    T(u,v) = -\kappa_-^2 \int_{\Omega^-} \overline{\nabla u} \cdot \nabla v - \kappa_+^2 \int_{\Omega^+} \overline{\nabla u} \cdot \nabla v
\end{align*}
with form domain $D(T) = H^1(\mathbb{R}^d)$.

The quadratic form $-T$ is continuous and coercive for the Hilbert space triple
$H^1(\mathbb{R}^d) \hookrightarrow L^2(\mathbb{R}^d) \hookrightarrow H^{-1}(\mathbb{R}^d)$.
Hence $T$ induces a self-adjoint operator on $L^2(\mathbb{R}^d)$ which we will denote by $\tm$.

\begin{lemma}
    \begin{enumerate}
\item The operator $\tm$ has the domain
\begin{align*}
D(\tm) = \{u \in H^1(\mathbb{R}^d) \mid \Delta u_\pm \in L^2(\Omega^\pm),\  \kappa_-^2\gammaNeu{-} u = \kappa_+^2 \gammaNeu{+}u\}.
\end{align*}
\item Its action on $u \in D(\tm)$ is given by
\begin{align*}
\tm u = \kappa_-^2 \chi_{\Omega^-} \Delta u_- + \kappa_+^2 \chi_{\Omega^+} \Delta u_+.
\end{align*}
\end{enumerate}
\end{lemma}
\begin{proof}
Let $(A, D(A))$ be the operator induced by the quadratic form $T$.
Its graph is
\begin{align*}
\operatorname{Gr}(A) = \{ (u, f) \in D(T) \times L^2(\mathbb{R}^d) \mid \int_{\mathbb{R}^d} \overline{f} v = T(u, v) \text{ for all } v \in D(T)\}.
\end{align*}
We will show that $\operatorname{Gr}(A) = \operatorname{Gr}(\tm)$
and therefore $A = \tm$, where $\tm$ is as in the statement of the Lemma.
Let $(u, f) \in \operatorname{Gr}(A)$.
Then for all $v \in H^1(\mathbb{R}^d)$ we have
\begin{align}
\label{eq:transmission-graph-eq-1}
\int_{\mathbb{R}^d} \overline{f} v = T(u,v) = -\kappa_-^2 \int_{\Omega^-} \overline{\nabla u} \cdot \nabla v dx - \kappa_+^2 \int_{\Omega^+} \overline{\nabla u} \cdot \nabla v dx.
\end{align}
Now if $v \in C_c^\infty(\mathbb{R}^d \setminus \partial \Omega)$, then we can integrate by parts without a boundary term:
\begin{align*}
\int_{\mathbb{R}^d} \overline{f} v = \kappa_-^2 \int_{\Omega^-} \overline{\Delta u} v + \kappa_+^2 \int_{\Omega^+} \overline{\Delta u} v.
\end{align*}
We conclude that $f = \kappa_-^2 \chi_{\Omega^-} \Delta u + \kappa_+^2 \chi_{\Omega^+} \Delta u$ almost everywhere.

In particular, $f_- = \kappa_-^2 \Delta u_-$ and $f_+ = \kappa_+^2 \Delta u_+$ (where we have denoted restriction to $\Omega^\pm$ by the subscript $\pm$).
Thus $\Delta u_\pm \in L^2(\Omega^\pm)$.

We insert $f$ back into (\ref{eq:transmission-graph-eq-1}) and obtain
\begin{align*}
\kappa_-^2 \int_{\Omega^-} \overline{\Delta u} v + \kappa_+^2 \int_{\Omega^+} \overline{\Delta u} v = -\kappa_-^2 \int_{\Omega^-} \overline{\nabla u} \cdot \nabla v - \kappa_+^2 \int_{\Omega^+} \overline{\nabla u} \cdot \nabla v.
\end{align*}
Rearranging gives
\begin{align*}
\kappa_-^2 \left[ \int_{\Omega^-} \overline{\Delta u} v + \int_{\Omega^-} \overline{\nabla u} \cdot \nabla v\right] + \kappa_+^2 \left[ \int_{\Omega^+} \overline{\Delta u} v + \int_{\Omega^+} \overline{\nabla u} \cdot \nabla v\right] = 0.
\end{align*}
This is the weak formulation of 
\begin{align*}
    \int_{\partial \Omega} \left[\kappa_-^2 \overline{\gammaNeu{-} u}\gammaDir{-} v - \kappa_+^2 \overline{\gammaNeu{+}u}\gammaDir{+} v\right]  = 0.
\end{align*}
Since $v\in H^1(\mathbb{R}^d)$ is arbitrary, this implies
$ \kappa_-^2 \gammaNeu{-} u = \kappa_+^2 \gammaNeu{+}u$.
Hence $(u, f) \in \operatorname{Gr}(\tm)$.

Now suppose that $(u, f) \in \operatorname{Gr}(\tm)$.
Then $f = \kappa_-^2 \chi_{\Omega^-} \Delta u + \kappa_+^2 \chi_{\Omega^+} \Delta u$, hence
\begin{align*}
\int_{\mathbb{R}^d} \overline{f} v = \kappa_-^2 \int_{\Omega^-} \overline{\Delta u} v + \kappa_+^2 \int_{\Omega^+} \overline{\Delta u} v.
\end{align*}
Now the weak formulation of $\kappa_-^2 \gammaNeu{-} u = \kappa_+^2 \gammaNeu{+}u$
tells us that this is precisely $T(u,v)$.
Hence $(u, f) \in \operatorname{Gr}(A)$.
\end{proof}

As $\Delta_T$ is self-adjoint, the transmission resolvent $(-\Delta_T - \lambda^2)^{-1}$ is holomorphic for $\operatorname{Im}(\lambda) > 0$.


\section{Invertibility of boundary layer operators.}

In order for the function $\Xi$ to be well defined, we need to establish the invertibility of our boundary operators.
In the case of Neumann boundary conditions, the hypersingular operator $N_\lambda$ is not invertible at $\lambda = 0$.
However, we will be able to show that $N_\lambda N_{\mathrm{diag},\lambda}^{-1}$ is nevertheless regular at $\lambda = 0$ because the singular part of $N_{\mathrm{diag},\lambda}^{-1}$ lies in the kernel of $N_0$.

We will rely on the analytic Fredholm theorem.
Take the holomorphic family $N_\lambda$ and suppose that $d\geq 3$ is odd.
Once we have established that $N_\lambda$ is Fredholm for every $\lambda \in \mathbb{C}^+$, the Fredholm theorem tells us that $N_\lambda^{-1}$ is meromorphic with finite rank poles.
It follows that $N_\lambda^{-1}$ is either holomorphic at 0 or has a pole (say of order $m$) and finite rank at 0.
Let $\lambda^{-m} \Pi$ be the leading order singular term.
Since $\lambda^{-m} N_\lambda \Pi$ must be holomorphic, it follows that $N_\lambda \Pi$ vanishes to order $m$.
In particular, $N_0 \Pi = 0$.

In even dimension, we will argue in the same way, making use of the Hahn holomorphic Fredholm theorem, Theorem 4.1 in \cite{MR3227433}.

To establish Fredholmness at every spectral parameter $\lambda$, we need the following compactness lemma.

\begin{lemma}
\label{lem:bd-op-derivative}
For $\operatorname{Im}(\lambda) > 0$, the operators
\begin{align}
  \dot{S}_\lambda := \frac{\der}{\der \lambda} S_\lambda &= 2 \lambda \tilde S_\lambda^t \tilde S_\lambda: H^{-\frac{1}{2}}(\partial \Omega) \to H^{\frac{1}{2}}(\partial \Omega) \\
  \dot{N}_\lambda := \frac{\der}{\der \lambda} N_\lambda &= 2 \lambda \tilde D_\lambda^t \tilde D_\lambda: H^{\frac{1}{2}}(\partial \Omega) \to H^{-\frac{1}{2}}(\partial \Omega)\\
  \dot{D}_\lambda := \frac{\der}{\der \lambda} D_\lambda &= 2 \lambda \tilde S_\lambda^t \tilde D_\lambda: H^{\frac{1}{2}}(\partial \Omega) \to H^{\frac{1}{2}}(\partial \Omega)
\end{align}
 are compact.
\end{lemma}
\begin{proof}
We will give the proof for $S_\lambda$, the proofs for
$N_\lambda$ and $D_\lambda$ are analogous.
By the jump relations, $S_\lambda$ can be defined by taking the exterior trace:
$S_\lambda = \gammaDir{+} G_\lambda (\gammaDir{+})^*$.
By the first resolvent identity,
\begin{align*}
\frac{d}{d\lambda} (-\Delta - \lambda^2)^{-1} = 2\lambda (-\Delta - \lambda^2)^{-2}
\end{align*}
and therefore
\begin{align*}
\frac{d}{d\lambda} \gammaDir{+} (-\Delta - \lambda^2)^{-1} (\gammaDir{+})^*
= 2\lambda \gammaDir{+} (-\Delta - \lambda^2)^{-1} (-\Delta - \lambda^2)^{-1}(\gammaDir{+})^* = 2\lambda \widetilde{S}_\lambda^t
\widetilde{S}_\lambda.
\end{align*}

Insert an arbitrary cutoff function $\psi \in C_c^\infty(W)$, where $W \subset \mathbb{R}^d$ is some bounded open set.
The composition
\begin{align*}
H^{-1/2}(\partial\Omega) \xrightarrow{\psi \widetilde{S}_\lambda} H^1(W)
\xhookrightarrow{i} H^{-1}(W) \xrightarrow{\widetilde{S}_\lambda^*} H^{1/2}(\partial\Omega)
\end{align*}
is compact since 
$i: H^1(W) \xhookrightarrow{} H^{-1}(W)$ is compact
and $\widetilde{S}_\lambda, \widetilde{S}_\lambda^*$ are bounded on the respective
Sobolev spaces.

Therefore it suffices to show that $\gammaDir{+} G_\lambda (1-\psi) G_\lambda (\gammaDir{+})^*$ is compact.

The kernel of $(1-\psi) G_\lambda (\gammaDir{+})^*$ is smooth and exponentially decaying.
By Lemma 12.3 in \cite{strohmaierRelativeTraceFormula2021}, it follows that $H^{-1/2}(\partial \Omega) \to H^1(\mathbb{R}^d)$ is Hilbert-Schmidt
and $\gammaDir{+} G_\lambda$ is bounded $H^1(\mathbb{R}^d) \to H^{1/2}(\partial \Omega)$, so the composition is compact.
\end{proof}

\begin{theorem}
    \label{thm:boundary-ops-fredholm}
    Let $\Omega\subset\mathbb{R}^d$ be a bounded Lipschitz domain.
    \begin{enumerate}
       \item At $\lambda = \rmi$, the single layer operator and the hypersingular operator
        are both symmetric and coercive in the following sense. There exists $c > 0$ such that
        \begin{align}
            &\langle \varphi, S_\rmi \varphi\rangle_{\partial\Omega} \geq c\norm{\varphi}_{H^{-1/2}(\partial\Omega)}^2 \label{eq:single-layer-coercive}\\
            - &\langle \varphi, N_\rmi \varphi \rangle_{\partial\Omega} \geq c \norm{\varphi}_{H^{1/2}(\partial\Omega)}^2. \label{eq:hypersingular-coercive}
        \end{align}
        \item For $\Im \lambda > 0$, the single layer, double layer and hypersingular operators are compact perturbations of the operators at $\lambda = \rmi$:
        \begin{align*}
            &S_\lambda - S_\rmi : H^{-1/2}(\partial\Omega) \to H^{1/2}(\partial\Omega)\qquad\text{is compact,}\\
            &D_\lambda - D_\rmi : H^{1/2}(\partial\Omega) \to H^{1/2}(\partial\Omega)\qquad\text{is compact,}\\
            &N_\lambda - N_\rmi : H^{1/2}(\partial\Omega) \to H^{-1/2}(\partial\Omega)\qquad\text{is compact.}
        \end{align*}
        \item The operators $S_\lambda, N_\lambda, (\frac12 \pm D_\lambda)$ are holomorphic families of Fredholm operators with index 0 for all $\lambda \in \mathbb{C}^+$.
    \end{enumerate}
 \end{theorem}
 \begin{proof}
    \begin{enumerate}
    \item These properties are given in Theorems 5.44, 5.47 \cite{kirsch} for $d=3$ and the proofs go through without change for any $d\geq 2$.
    
    \item We can integrate Lemma \ref{lem:bd-op-derivative} to obtain compactness.
    The operators $S_\lambda, D_\lambda, N_\lambda$ are holomorphic on $\mathbb{C}^+$.

    Hence the integral
    \begin{align*}
    S_\lambda - S_\rmi = \int_{\gamma} \frac{\der}{\der z} S_z dz
    \end{align*}
    converges in operator norm for any contour $\gamma$ connecting $\rmi$ to $\lambda$.
    Since the integrand is compact, it follows that $S_\lambda - S_\rmi$ is compact; the same argument applies to $D_\lambda$ and $N_\lambda$.
    
    \item Item (1) implies that $S_\rmi$ and $N_\rmi$ are invertible and therefore Fredholm operators of index zero.
    Equation (\ref{eq:alg-rel-1}) then implies the same for $\frac{1}{2} \pm D_\rmi$.
    Compactness of the differences then implies the same for $S_\lambda, \frac{1}{2} \pm D_\lambda, N_\lambda$.
    \end{enumerate}
 \end{proof}

 \begin{proposition}\label{prop:compute-bd-op-kernels}
 \begin{enumerate}
 \item The kernel of the single layer operator $S_\lambda$ is the space of Neumann boundary data of interior Dirichlet eigenfunctions for eigenvalue $\lambda^2$.
 \item If $\lambda^2$ is not an interior Dirichlet eigenvalue, $\ker\left(\frac{1}{2} - D_\lambda\right) = \{0\}$.
 \item If $\lambda^2$ is not an interior Dirichlet eigenvalue, the kernel of
  $\left(\frac{1}{2} + D_\lambda\right)$ is the space of Dirichlet boundary data of interior Neumann eigenfunctions.
 \end{enumerate}
 \end{proposition}
 \begin{proof}
    Given $\phi \in L_{\mathrm{loc}}^2(\mathbb{R}^d)$ we will write $\phi_\pm$
    for the restrictions to $\Omega^\pm$.
 \begin{enumerate}
 \item Suppose that $S_\lambda u = 0$ and define $\phi := \widetilde{S}_\lambda u$. 
 Then $\gamma_D^\pm \phi = S_\lambda u = 0$, so $\phi$ solves the Dirichlet problem
 in the interior and exterior domains.
 By unique solvability of the exterior Dirichlet problem (Theorem 9.10 \cite{mcleanStronglyEllipticSystems2000}), $\phi_+ = 0$
 and therefore $\gamma_N^+ \phi = 0$.
 Now the jump conditions give us $-u = [\gamma_N \phi] = -\gamma_N^- \phi$
 and therefore $u$ is the Neumann data of an interior Dirichlet eigenfunction.

Conversely suppose that $\phi$ is an interior Dirichlet eigenfunction and $u = \gamma_N^- \phi$.
Then by Theorem 7.5 of \cite{mcleanStronglyEllipticSystems2000}, $u$ must satisfy
$S_\lambda u = 0$.
 \item Suppose that $\left(\frac 1 2 - D_\lambda\right)u=0$ and put $\phi := \widetilde{D}_\lambda u$. Then $\gamma_D^- \phi = -\left(\frac 1 2 - D_\lambda\right)u = 0$
 and so $\phi$ is an interior Dirichlet eigenfunction.
 Hence if $\lambda^2$ is not an interior Dirichlet eigenvalue, $\phi_- = 0$ and
 therefore $\gamma_N^+\phi = \gamma_N^- \phi = 0$.
 Therefore $\phi$ satisfies the exterior Neumann problem; by unique exterior solvability (again Theorem 9.10 \cite{mcleanStronglyEllipticSystems2000}) $\phi = 0$ everywhere.
Then the jump relations imply $u = [\gamma_D \phi] = 0$.
 \item Suppose that $\left(\frac{1}{2} + D_\lambda\right)u = 0.$
 Put $\phi := \widetilde{D}_\lambda u$, then the jump relations $(\ref{eq:jump-rel-2})$
 give
 \begin{align*}
 &\gamma_D^- \phi = \left(-\frac{1}{2} + D_\lambda\right)u = \left(\frac{1}{2} + D_\lambda\right)u - u = -u,\\
 &\gamma_D^+ \phi = \left(\frac{1}{2} + D_\lambda\right) u = 0.
 \end{align*}
Thus $\phi_+$ solves the exterior Dirichlet problem
and therefore $\phi_+ = 0$.
Hence we have $\gamma_N^-\phi = \gamma_N^+\phi = 0$, so $\phi_-$ is a Neumann eigenfunction
and $u$ is the Dirichlet boundary data of an interior Neumann eigenfunction.

Conversely, suppose that $\phi$ is an interior Neumann eigenfunction and $u = \gamma_D \phi$.
Then by Theorem 7.7 of \cite{mcleanStronglyEllipticSystems2000}, $N_\lambda u = 0$ and
therefore $\left(\frac{1}{2} - D_\lambda\right)\left(\frac{1}{2} + D_\lambda\right)u=0$.
If $\lambda^2$ is not an interior Dirichlet eigenvalue, the kernel of $\frac{1}{2} - D_\lambda$ is trivial and therefore we must have $\left(\frac{1}{2} + D_\lambda\right)u=0$.
\end{enumerate}
\end{proof}

\begin{rem}
    By Corollary 8.3 of \cite{mcleanStronglyEllipticSystems2000}, $0$ is not an interior Dirichlet eigenvalue, and therefore $\ker S_0 = \{0\}$.
    The kernel of $\left(\frac{1}{2} + D_0\right)$ consists of functions that are constant on each component of $\partial \Omega$.
\end{rem}

\begin{proposition}
\label{prop:expand-d-to-n-inv}
Let $e_i: \partial\Omega \to \mathbb{R}$ be $e_i = \chi_{\partial\Omega_i}$ and define the finite rank operator $\Pi: H^{-1/2}(\partial\Omega) \to H^{1/2}(\partial\Omega)$
\begin{align*}
\Pi f = \sum_{i=1}^M |\Omega_i|^{-1} \langle f, e_i\rangle_{\partial\Omega} e_i.
\end{align*}
Then the operator $(Q_\lambda^-)^{-1}: H^{-1/2}(\partial\Omega) \to H^{1/2}(\partial \Omega)$ (inverse of interior Dirichlet-to-Neumann operator)
has an expansion
\begin{align}
    \label{eq:expand-d-to-n-inv}
    (Q_\lambda^-)^{-1} &= -\lambda^{-2} \Pi + \bigoplus_{i=1}^M M_{\lambda, i}\\
    M_{\lambda, i} f &= \sum_{\lambda_k \neq 0} \frac{1}{\lambda_{k,i}^2 - \lambda^2} \langle f, \gammaDir{-} \psi_{k,i}\rangle \gammaDir{-} \psi_{k,i}
\end{align}
which converges in $H^{1/2}(\partial\Omega)$ for all $f \in H^{-1/2}(\partial \Omega)$.

Here the $\psi_{k,i}$ are an orthonormal basis of $L^2(\Omega_i^-)$ given by interior Neumann eigenfunctions and $\lambda_{k,i}$ is the Neumann eigenvalue corresponding to $\psi_{k,i}$.
\end{proposition}
\begin{proof}
Since the interior Dirichlet problem on $\Omega$ decouples into the individual problems for the $\Omega_i$,
$Q_\lambda^-$ is a block-diagonal operator on $\bigoplus_{i=1}^N H^{1/2}(\partial\Omega_i)$.
Hence it suffices to establish the expansion for each $\Omega_i$ separately.
We will therefore assume from now on that $\Omega$ is connected.

    By Theorem 4.12 \cite{mcleanStronglyEllipticSystems2000}, there is an orthonormal basis $\{\psi_k\mid k=0,1,\ldots\}$ of $L^2(\Omega)$ consisting of interior Neumann eigenfunctions, i.e. the $\psi_k$ are in $H^1_\Delta(\Omega)$ and satisfy
    \begin{align*}
\begin{cases}
\Delta \psi_k = -\lambda_k^2 \psi_k\\
\gammaNeu{-} \psi_k = 0.
\end{cases}
\end{align*}

Moreover, the norm on $H^1(\Omega)$ is equivalent to the norm
\begin{align*}
\norm{u}^2 := \sum_{k=0}^\infty (1 + \lambda_k^2) |\langle \psi_k, u\rangle_{L^2(\Omega)}|^2
\end{align*}
and therefore if $E \in H^1(\Omega)$, the expansion $E = \sum_{k=0}^\infty \langle E, \psi_k\rangle \psi_k$ converges in $H^1(\Omega)$.
Now suppose that $E$ satisfies $(\Delta + \lambda^2) E = 0$.
Observe that
\begin{align*}
\langle \Delta E , \psi_k\rangle - \langle E, \Delta \psi_k\rangle = (\lambda_k^2 - \lambda^2) \langle E, \psi_k\rangle.
\end{align*}
Then by using the second Green identity we obtain
\begin{align*}
&\langle E, \psi_k\rangle = \frac{1}{\lambda_k^2 - \lambda^2} \left(\langle \Delta E , \psi_k\rangle - \langle E, \Delta \psi_k\rangle \right)\\
=\ &\frac{1}{\lambda_k^2 - \lambda^2} \left(\langle \gammaNeu{-} E, \gammaDir{-} \psi_k\rangle_{\partial\Omega} - \langle\gammaDir{-} E, \gammaNeu{-} \psi_k\rangle_{\partial\Omega}\right)\\
=\ &\frac{1}{\lambda_k^2 - \lambda^2} \langle \gammaNeu{-} E, \gammaDir{-} \psi_k\rangle.
\end{align*}
By the continuity of the Dirichlet trace $\gamma_D: H^1(\Omega) \to H^{1/2}(\partial \Omega)$, we obtain
\begin{align*}
(Q_\lambda^-)^{-1} \gammaNeu{-}E = \gammaDir{-} E = \sum_{k\geq 0} \frac{1}{\lambda_k^2 - \lambda^2} \langle \gammaNeu{-} E, \gammaDir{-} \psi_k\rangle \gamma_D \psi_k.
\end{align*}
The normalized Neumann eigenfunction for the first eigenvalue, $\lambda_0 = 0$ is the constant function $\psi_0 = |\Omega|^{-1/2}$,
which yields the expansion (\ref{eq:expand-d-to-n-inv}).
\end{proof}

\subsection{Invertibility of hypersingular operator.}
By Theorem \ref{thm:boundary-ops-fredholm} $N_\lambda$ has index 0 and is therefore invertible if and only if $\ker N_\lambda = \{0\}$.

 From the relations $(\ref{eq:alg-rel-1}), (\ref{eq:alg-rel-2}), (\ref{eq:DtN-rel})$ we obtain
 \begin{align*}
 N_\lambda = -Q_\lambda^-(\frac12 - D_\lambda) = -(\frac12 - D_\lambda') Q_{\lambda}^-.
 \end{align*}
Since $(\frac12 - D_\lambda)$ is invertible for $\lambda \in \mathfrak{D}_\epsilon \cup \{0\}$, so is its transpose $(\frac12 - D_\lambda')$.
Since $Q_\lambda^-$ is block diagonal,
\begin{align*}
N_\lambda N_{\mathrm{diag},\lambda}^{-1} = (\frac12 - D_\lambda') Q_\lambda^- (Q_{\mathrm{diag},\lambda}^-)^{-1} (\frac12 - D_{\mathrm{diag},\lambda}')^{-1}.
= (\frac12 - D_\lambda')(\frac12 - D_{\mathrm{diag},\lambda}')^{-1}.
\end{align*}
The behaviour near $\lambda =0$ of $N_\lambda N_{\mathrm{diag},\lambda}^{-1}$ is summarized in the following proposition.

\begin{proposition}
    \label{prop:invert-hypersingular}
The operator family $\lambda \mapsto N_\lambda N_{\mathrm{diag},\lambda}^{-1} \in \mathcal{L}(H^{-1/2}(\partial\Omega))$ is holomorphic in the upper half plane and has a continuous boundary value near $\lambda=0$ on the real line.
\end{proposition}
\begin{proof}
The operator $(\frac{1}{2} - D_{\lambda}')$ is a family of Fredholm operators of index $0$, holomorphic in odd dimensions and Hahn holomorphic in even dimensions.
Since the operator $(\frac{1}{2} - D_{\lambda, \mathrm{diag}}')$ is invertible for $\Im(\lambda)>0$ it follows that $(\frac{1}{2} - D_{\lambda, \mathrm{diag}}')^{-1}$ is a (Hahn) meromorphic family of operators with Laurent coefficients of finite rank at every pole.
By Prop. \ref{prop:compute-bd-op-kernels} zero is not a pole, and therefore 
$(\frac{1}{2} - D_{\lambda, \mathrm{diag}}')^{-1}$ is (Hahn) holomorphic.
Hence the same is true for $N_\lambda N_{\mathrm{diag},\lambda}^{-1} = (\frac12 - D_\lambda')(\frac12 - D_{\mathrm{diag},\lambda}')^{-1}.$

 \end{proof}

\section{Polynomial bounds on boundary layer operators.}
\label{sec:boundary-op-poly-growth}

In order to prove the trace formulas, we need to establish polynomial bounds in $\lambda$ on the growth of certain boundary operators.
First we need to establish quantitative versions of the unique solvability of the boundary problems we consider.

\begin{definition}
    \label{def:sector}
Let $0 < \epsilon < \pi/8$. The sector $\mathfrak{D}_\epsilon$ is defined as
\begin{align}
\mathfrak{D}_\epsilon := \{z \in \mathbb{C}\setminus 0 \mid \epsilon<\arg (z)<\pi-\epsilon\}.
\end{align}
\end{definition}

\begin{theorem}
\label{thm:helmholtz-lambda-estimate}
Let $\Omega$ be a Lipschitz domain.
Then for all $R > 0$ there exists a constant $C>0$ such that the following estimate holds for all $\lambda \in \mathfrak{D}_\epsilon\setminus B_R(0)$.

If $u$ solves $(-\Delta - \lambda^2) u = f$ in $\Omega$, $f \in H^{-1}(\Omega)$, then $u$ satisfies
\begin{align}
\label{eq:helmholtz-lambda-estimate}
\norm{u}_{H^1(\Omega)}^2 \leq C\left( \norm{u}_{H^1(\Omega)}\norm{f}_{H^{-1}(\Omega)} + \norm{\gamma_0 u}_{H^{1/2}(\partial \Omega)}\cdot\norm{\gamma_1 u}_{H^{-1/2}(\partial \Omega)}\right).
\end{align}
\end{theorem}
\begin{corollary}
  \label{cor:helmholtz-uniqueness-estimate}
This estimate together with continuity of the trace operators provides the estimates
\begin{align*}
&\norm{u}_{H^1(\Omega)} \leq C_1 \left( \norm{f}_{H^{-1}(\Omega)} + |\lambda|^2\norm{\gamma_0 u}_{H^{1/2}(\partial \Omega)}\right)\\
&\norm{u}_{H^1(\Omega)} \leq C_2\left( \norm{f}_{H^{-1}(\Omega)} + \norm{\gamma_1 u}_{H^{-1/2}(\partial \Omega)}\right)
\end{align*}
for some constants $C_1, C_2 > 0$ and for all $\lambda$ in the same region as above.
\end{corollary}
\begin{proof}[Proof of Theorem \ref{thm:helmholtz-lambda-estimate}]
The proof is a simple integration by parts argument, complicated only by the fact that we need to vary the computation very slightly for different sectors of the complex plane.

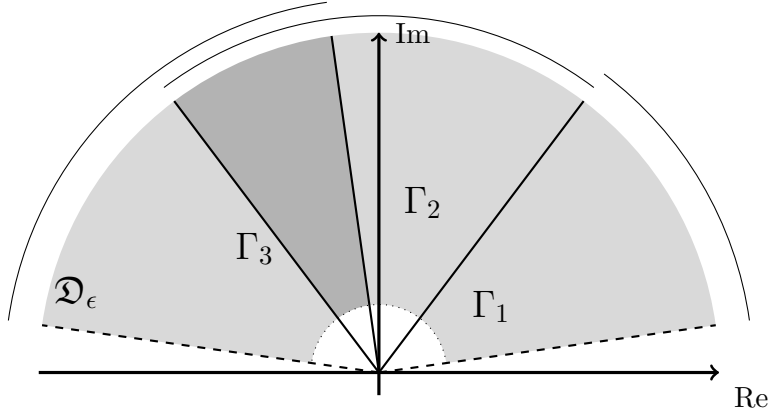
\begin{figure}
\centering
\begin{tikzpicture}[scale=1.5]

\pgfmathsetmacro{\aone}  {\tikzeps}
\pgfmathsetmacro{\atwo}  {45 + \tikzeps}
\pgfmathsetmacro{\athree}{90 + \tikzeps}
\pgfmathsetmacro{\afour} {135 - \tikzeps}
\pgfmathsetmacro{\afive} {180 - \tikzeps}

\pgfmathsetmacro{\Rone}   {\tikzR + 2*\tikzdelta}
\pgfmathsetmacro{\Rtwo}   {\tikzR + \tikzdelta}
\pgfmathsetmacro{\Rthree} {\tikzR + 2*\tikzdelta}

\fillwedge{\gSect}    {\aone}  {\atwo}          
\fillwedge{\gSect}    {\atwo}  {\athree}        
\fillwedge{\gOverlap} {\athree}{\afour}         
\fillwedge{\gSect}    {\afour} {\afive}         

\ifdim\tikzrmin pt>0pt
  \path[fill=white,draw=none] (0,0) circle (\tikzrmin);
  \draw[dotted] (\aone:\tikzrmin) arc (\aone:\afive:\tikzrmin);
\fi

\draw[\arcstyle] (\aone:\Rone)   arc (\aone:\atwo:\Rone);     
\draw[\arcstyle] (\atwo:\Rtwo)   arc (\atwo:\afour:\Rtwo);    
\draw[\arcstyle] (\athree:\Rthree) arc (\athree:\afive:\Rthree); 

\draw[very thick,->] (-\tikzR,0) -- (\tikzR,0)
      node[below right=2pt] {$\mathrm{Re}$};
\draw[very thick,->] (0,-0.2) -- (0,\tikzR)
      node[right=2pt] {$\mathrm{Im}$};

\draw[dashed, thick] (0,0) -- (\aone:\tikzR);
\draw[dashed, thick] (0,0) -- (\afive:\tikzR);

\draw[thick] (0,0) -- (\atwo:\tikzR);
\draw[thick] (0,0) -- (\athree:\tikzR);
\draw[thick] (0,0) -- (\afour:\tikzR);

\pgfmathsetmacro{\midone}  {0.5*(\aone+\atwo)}
\pgfmathsetmacro{\midtwo}  {0.5*(\atwo+\athree)}
\pgfmathsetmacro{\midthree}{0.5*(\athree+\afive)}

\node at (\midone:1.15)   {\Large $\Gamma_{1}$};
\node at (\midtwo:1.55)   {\Large $\Gamma_{2}$};
\node at (\midthree:1.55) {\Large $\Gamma_{3}$};

\pgfmathsetmacro{\dlabelangle}{0.96*\afive}
\node at (\dlabelangle:\tikzR*0.92) {\Large $\mathfrak{D}_{\epsilon}$};

\end{tikzpicture}
\caption{The sector $\mathfrak{D}_\epsilon$ is bounded by the dashed lines. The arcs outside the shaded region represent the angles covered by the sectors $\Gamma_1, \Gamma_2, \Gamma_3$. The darker shaded region shows the overlap of the sectors $\Gamma_2$ and $\Gamma_3$.}
\end{figure}

We subdivide the sector $\mathfrak{D}_\epsilon$ into three smaller sectors (truncated near 0):
\begin{align*}
&\Gamma_1(\epsilon, R) := \{ r e^{\rmi \varphi} \mid \varphi \in [\epsilon, \pi/4 + \epsilon], r > R\}\\
&\Gamma_2(\epsilon, R) := \{ r e^{\rmi \varphi} \mid \varphi \in [\pi/4 + \epsilon, 3\pi/4 - \epsilon], r > R\}\\
&\Gamma_3(\epsilon, R) := \{ r e^{\rmi  \varphi} \mid \varphi \in [\pi/2 + \epsilon, \pi - \epsilon], r > R\}.
\end{align*}
We will prove our desired estimate on the sets $\Gamma_2, \Gamma_3$ and use invariance of the estimate under the reflection $\lambda \mapsto -\overline{\lambda}$ to obtain the estimate on $\Gamma_1$.

The following auxiliary estimates allow us to assume a lower bound for one of $-\operatorname{Re}(\lambda^2)$, $-\operatorname{Im}(\lambda^2)$.

If $\lambda = r e^{\rmi \varphi}$ then $\operatorname{Re}(\lambda^2) = r^2 \cos(2\varphi)$.
On $\Gamma_2$ we have $\cos(2\varphi) \leq - \sin(2\epsilon)$, hence
$-\operatorname{Re}(\lambda^2) \geq R^2 \sin(2\epsilon)$.

On $\Gamma_3$ we have $-\cos \varphi \geq \sin \epsilon$, $\sin \varphi \geq \sin \epsilon$ and therefore
\begin{align*}
-\operatorname{Im}\left(\lambda^2\right)=-2(\operatorname{Re} \lambda)(\operatorname{Im} \lambda)=2 r^2(-\cos \varphi) \sin \varphi
\geq 2 r^2 \sin ^2 \epsilon>2 R^2 \sin ^2 \epsilon .
\end{align*}

To prove (\ref{eq:helmholtz-lambda-estimate}) on $\Gamma_2(\epsilon, R)$,
let $u \in H^1(\Omega)$ be a solution of $(-\Delta - \lambda^2) u = f$ in $\Omega$ with $\lambda \in \Gamma_2$.
In particular, $-\Re(\lambda^2) > R^2$.
Choosing $\delta := \min \{1, R^2 \sin(2\epsilon)\}$ ensures
\begin{align*}
\delta \norm{u}_{H^1(\Omega)}^2 \leq \operatorname{Re} \int_{\Omega} \left( |\nabla u|^2 - \lambda^2 |u|^2 \right) \mathrm{d} x.
\end{align*}
Applying the first Green identity to the right hand side gives
\begin{align*}
&\operatorname{Re} \int_{\Omega} \left( |\nabla u|^2 - \lambda^2 |u|^2 \right) \mathrm{d} x = 
\operatorname{Re} \int_{\Omega} \overline{u} f + \operatorname{Re} \int_{\partial\Omega} \overline{\gamma_0 u} \gamma_1 u\\
\leq\ 
 &C_1 \norm{u}_{H^1(\Omega)} \cdot \norm{f}_{H^{-1}(\Omega)} + C_2 \norm{\gamma_0 u}_{H^{1/2}(\partial \Omega)} \cdot \norm{\gamma_1 u}_{H^{-1/2}(\partial \Omega)}
\end{align*}
where we have simply used the continuity of the Sobolev space pairings.

The same proof (mutatis mutandis) on $\Gamma_3$ goes through, with $\delta := \min \{1, 2R^2 \sin^2\epsilon\}$ and replacing the real part with the imaginary part.

The map $\lambda\mapsto -\overline{\lambda}$ maps $\Gamma_1$ bijectively to a subset of $\Gamma_3$.
If $u$ solves the Helmholtz equation for $\lambda$ and source term $f$, then $\overline{u}$ solves the Helmholtz equation for $-\overline{\lambda}$ and source term $\overline{f}$, and so the same estimate applies for $-\overline{\lambda}$.

Thus we cover the entire truncated sector $\mathfrak{D}_\epsilon \setminus B_R(0)$.
\end{proof}

In addition we need a transmission type variant of the previous result.
\begin{theorem}
\label{thm:helmholtz-transmission-lambda-estimate}
Let $\Omega$ be a bounded Lipschitz domain, $\Omega^+$ the exterior domain.
Then for all $R > 0$ there exists a constant $C>0$ such that the following estimate holds for all $\lambda \in \mathfrak{D}_\epsilon\setminus B_R(0)$.

Let $u = u_+ + u_-$ where $u_\pm \in H^1(\Omega^\pm)$ solve the transmission problem
(\ref{eq:transmission-problem-def}) with $\kappa_\pm, \nu_0, \nu_1>0$.
Then $u$ satisfies the estimate
\begin{align}
  \label{eq:helmholtz-transmission-lambda-estimate}
\norm{u}_{H^1 \oplus H^1} \leq C\left(\norm{f}_{H^{-1} \oplus H^{-1}}
+ |\lambda|^2\norm{\phi_0}_{H^{1/2}} + \norm{\psi_0}_{H^{-1/2}}\right).
\end{align}
\end{theorem}
\begin{rem}
Using the Cauchy data space $\mathscr{C}$ we also have the following estimate, which is clearly slightly less precise, but just as good for our purposes:
\begin{align*}
\norm{u}_{H^1 \oplus H^1} \leq C\left(\norm{f}_{H^{-1} \oplus H^{-1}} + |\lambda|^2\norm{\Phi_0}_{\mathscr{C}} \right),\qquad \Phi_0 = (\phi_0, \psi_0).
\end{align*}
\end{rem}
\begin{proof}
Let $\Phi_0 = (\phi, \psi) \in \mathscr{C}$ be the transmission data.
We treat the cases $\phi = 0$ and $\psi = 0$ separately and then combine the results.
Just as in the proof of Theorem \ref{thm:helmholtz-lambda-estimate}, to achieve the estimate in $\mathfrak{D}_\epsilon \setminus B_R(0)$ we need to consider different sectors where one of $-\operatorname{Re}(\lambda^2)$, $-\operatorname{Im}(\lambda^2)$ is positive.
As the treatment of these sectors is identical to that in the proof of Theorem \ref{thm:helmholtz-lambda-estimate}, we will only give the computation in the sector $\Gamma_2(\epsilon, R)$, where $-\operatorname{Re}(\lambda^2) > R^2$.

Consider first the case $\phi = 0$.
Let $\delta := \min \{1, \kappa_-^2, \kappa_+^2, R^2\}$.
This choice of $\delta$ ensures
\begin{align*}
\delta \norm{u}_{H^1(\Omega^-) \oplus H^1(\Omega^+)}^2 \leq \operatorname{Re} \int_{\Omega_-} \left(\kappa_-^2 |\nabla u|^2 - \lambda^2 |u|^2 \right) \mathrm{d} x
+ \operatorname{Re} \int_{\Omega_+} \left( \kappa_+^2 |\nabla u|^2 - \lambda^2 |u|^2 \right) \mathrm{d} x.
\end{align*}

Now when we apply the first Green identity to the right hand side, the contributions from the interior and the exterior domains combine to give a term containing $\psi$:
\begin{align*}
  &\operatorname{Re} \int_{\Omega_-} \left(\kappa_-^2 |\nabla u|^2 - \lambda^2 |u|^2 \right) \mathrm{d} x
+ \operatorname{Re} \int_{\Omega_+} \left( \kappa_+^2 |\nabla u|^2 - \lambda^2 |u|^2 \right) \mathrm{d} x\\
=\ &\operatorname{Re} \int_{\Omega_-} \overline{u} f + \operatorname{Re} \int_{\Omega_+} \overline{u} f + \operatorname{Re} \int_{\partial\Omega} \overline{\gamma_0 u} \psi\\
\leq\ &\norm{u}_{H^1 \oplus H^1} \cdot \norm{f}_{H^{-1}\oplus H^{-1}}
+ \norm{\gamma_0 u}_{H^{1/2}(\partial \Omega)} \norm{\psi}_{H^{-1/2}(\partial \Omega)}.
\end{align*}
Since $\norm{\gamma_0 u} \lesssim \norm{u}_{H^1(\Omega^{\pm})}$, we certainly
have $\norm{\gamma_0 u} \lesssim \norm{u}_{H^1 \oplus H^1}$.
Dividing by $\norm{u}_{H^1 \oplus H^1}$ gives
\begin{align*}
\delta \norm{u}_{H^1 \oplus H^1} \lesssim  \norm{f}_{H^{-1}\oplus H^{-1}} + \norm{\psi}_{H^{-1/2}(\partial \Omega)}.
\end{align*}

If $\psi = 0$, then we obtain by an analogous computation
\begin{align*}
\delta \norm{u}_{H^1 \oplus H^1} \lesssim  \norm{f}_{H^{-1}\oplus H^{-1}} + |\lambda|^2\norm{\phi}_{H^{1/2}(\partial \Omega)}
\end{align*}
where the $|\lambda|^2$ comes from Lemma \ref{lem:continuity-neumann-trace}.

Now let $u$ be a solution of the transmission problem for source $f$ and boundary data $(\phi, \psi)$. Let $u_1$ be the solution of the transmission problem for $f=0$ with boundary data $(\phi, 0)$, and let $u_2$ be the solution of the transmission problem for $f=0$ with boundary data $(0, \psi)$.
Let $u_3$ be the solution of the transmission problem for a source $f$ and boundary data $(0,0)$.
Then $u_0 := u - u_1 - u_2 - u_3$ solves the homogeneous transmission problem.
But the unique solution of the homogeneous transmission problem is $u_0 = 0$,
so $u = u_1 + u_2 + u_3$.

This gives the estimate
\begin{align*}
\norm{u}_{H^1 \oplus H^1} \lesssim \norm{u_1} + \norm{u_2} + \norm{u_3}
\lesssim \norm{f}_{H^{-1} \oplus H^{-1}} + |\lambda|^2 \norm{\phi}_{H^{1/2}(\partial \Omega)} + \norm{\psi}_{H^{-1/2}(\partial \Omega)}.
\end{align*}
\end{proof}

\begin{proposition}
  \label{prop:boundary-op-bounds}
For any $R > 0$, the operators
\begin{align*}
&|\lambda|^{-2} S_\lambda^{-1}: H^{1/2} \to H^{-1/2}, \qquad
|\lambda|^{-2} Q_\lambda: H^{1/2} \to H^{-1/2},\\
&Q_\lambda^{-1}: H^{-1/2} \to H^{1/2}, \qquad
N_\lambda^{-1}: H^{-1/2} \to H^{1/2}
\end{align*}
are uniformly bounded for $\lambda \in \mathfrak{D}_\epsilon \setminus B_R(0)$.
\end{proposition}
\begin{proof}
Corollary \ref{cor:helmholtz-uniqueness-estimate} shows that if $(\Delta + \lambda^2) u = 0$, then
\begin{align*}
&\norm{\gamma_0 u} \lesssim \norm{u} \lesssim  \norm{\gamma_1 u}\\
&\norm{\gamma_1 u} \lesssim \norm{u} \lesssim |\lambda|^2 \norm{\gamma_0 u}
\end{align*}
which proves the boundedness of $Q_\lambda, Q_\lambda^{-1}$.

Since $S_\lambda^{-1} = Q_\lambda^- + Q_\lambda^+$ this implies boundedness for $S_\lambda^{-1}$.

Let us consider the case of $N_\lambda^{-1}$.
If we have $N_\lambda \phi = \psi$, we need to estimate $\norm{\phi}_{H^{1/2}}$ by $\norm{\psi}_{H^{-1/2}}$.
The idea is to put $u := \widetilde{D}_\lambda \phi$, and use the jump relations to recover $\phi$.
In $\mathfrak{D}_\epsilon \setminus B_R(0)$, the boundary layer potential $\widetilde{D}_\lambda$ decays exponentially, so that $u$ is certainly in $H^1(\Omega_+)\oplus H^1(\Omega_-)$.

We now have $\gamma_1 u = N_\lambda \phi$ and $[\gamma_0 u] = \phi$.
Thus by Corollary \ref{cor:helmholtz-uniqueness-estimate}:
\begin{align*}
\norm{\phi}_{H^{1/2}} = \norm{\gammaDir{+} u - \gammaDir{-}u} \leq \norm{\gammaDir{+} {u}}_{H^{1/2}} + \norm{\gammaDir{-} {u}}_{H^{1/2}}
\lesssim \norm{{u}_+}_{H^{1}} + \norm{{u}_-}_{H^1} \lesssim \norm{\gammaNeu{} u}_{H^{-1/2}} =\norm{N_\lambda \phi}.
\end{align*}
This proves boundedness of $N_\lambda^{-1}$.
\end{proof}

\begin{proposition}
  \label{prop:cald-op-inv-bound}
For all $R> 0$, $\epsilon > 0$, there exists a constant $C_{R, \epsilon} > 0$ such that
\begin{align*}
\norm{\Phi}_{\mathscr{C}} \leq C_{R, \epsilon} |\lambda|^2 \norm{H_\lambda \Phi}_{\mathscr{C}}
\end{align*}
for all $\lambda\in \mathfrak{D}_\epsilon\setminus B_R(0)$.
\end{proposition}
\begin{proof}
We estimate the norm of a solution $\Phi$ of
$H_\lambda \Phi = P_{\lambda/\kappa_+}^\pm \Phi_0$, treating the plus and minus cases separately (note that the spectral parameter is the \emph{exterior spectral parameter} $\lambda/\kappa_+$ in both cases!).
The Proposition is essentially a quantitative version of Theorem \ref{thm:transmission-problem-equivalence}, (ii); therefore our strategy will closely follow the proof of Proposition 4.2 in \cite{costabelDirectBoundaryIntegral1985} (the proof is mostly contained in the discussion following Lemma 4.2 loc.~cit.).

Specifically, we will construct an auxiliary problem which has $\Phi$ as the boundary data and $\Phi_0$ as the transmission data.
We can then apply Theorem \ref{thm:helmholtz-transmission-lambda-estimate} to estimate $\Phi$ in terms of $\Phi_0$.

Suppose first that $\Phi$ is a solution of
\(
H_\lambda \Phi = P_{\lambda/\kappa_+}^- \Phi_0.
\) 
If we put
\begin{align*}
\Phi_- := P_{\lambda/\kappa_-}^- \Phi,\qquad \Phi_+ := P_{\lambda/\kappa_+}^+ \left(\M \Phi + \Phi_0\right)
\end{align*}
then
\(
\Phi_+ - \M \Phi_- = \Phi_0.
\)
$\Phi_\pm$ are therefore the boundary data of a transmission problem with transmission data $\Phi_0$.

We now construct the data of a transmission problem with zero transmission data.
This involves exchanging the roles of interior and exterior domains.
Define
\begin{align*}
&\widetilde{\Phi}_- :=
\M \Phi + \Phi_0 - \Phi_+ = P_{\lambda/\kappa_+}^- \left(\M \Phi + \Phi_0 \right)\\
&\widetilde{\Phi}_+ :=
\Phi - \Phi_-
= P_{\lambda/\kappa_-}^+ \Phi.
\end{align*}
These data now satisfy
\begin{align*}
P_{\lambda/\kappa_+}^+ \widetilde{\Phi}_- = 0,\qquad
P_{\lambda/\kappa_-}^- \widetilde{\Phi}_+ = 0,\qquad
\M \widetilde{\Phi}_+ - \widetilde{\Phi}_- = 0.
\end{align*}

By Assumption $\widetilde{\mathrm{A}}$, this system has only the zero solution.
Hence
\begin{align*}
&\Phi_+ = \M \Phi + \Phi_0\\
&\Phi_- =\Phi.
\end{align*}

Now let $u_+$ be the exterior solution with boundary data $\Phi_+$ and $u_-$ the interior solution with boundary data $\Phi_-$.
Then using Theorem \ref{thm:helmholtz-transmission-lambda-estimate}, we have
\begin{align*}
\norm{ \Phi}_{\mathscr{C}}
= \norm{\gamma_- u}_{\mathscr{C}}
\lesssim \norm{u}_{H^1 \oplus H^1}
\lesssim |\lambda|^2 \norm{\Phi_0}_{\mathscr{C}}.
\end{align*}

Now suppose that $\Phi$ solves $H_\lambda \Phi = P_{\lambda/\kappa_+}^+ \Phi_0$.
We put $\Phi_+ := P_{\lambda/\kappa_-}^+ \Phi$ and $\Phi_- := P_{\lambda/\kappa_+}^- (\M\Phi - \Phi_0)$.

We then have
\begin{align*}
  \Phi_- = P_{\lambda/\kappa_+}^- \M \Phi - (1-P_{\lambda/\kappa_+}^+)\Phi_0
  = (P_{\lambda/\kappa_+}^- \M + H_\lambda)\Phi - \Phi_0
  = \M P_{\lambda/\kappa_-}^+ \Phi - \Phi_0 = \M \Phi_+ - \Phi_0.
\end{align*}

Now we solve an auxiliary homogeneous problem.
Put 
\begin{align*}
  &\widetilde{\Phi}_+ := P_{\lambda/\kappa_+}^+(\M\Phi - \Phi_0) = \M\Phi - \Phi_0 - \Phi_-\\
  &\widetilde{\Phi}_- := P_{\lambda/\kappa_-}^- \Phi = \Phi - \Phi_+.
\end{align*}
These data now solve $\M \widetilde{\Phi}_- - \widetilde{\Phi}_+ = 0$.
By Assumption A, we must have $\widetilde{\Phi}_- = 0$, $\widetilde{\Phi}_+ = 0$.

Now let $u$ be the unique solution with interior/exterior boundary data $\Phi_\pm$.
Again using Theorem \ref{thm:helmholtz-transmission-lambda-estimate}
we obtain
\begin{align*}
\norm{\Phi}_{\mathscr{C}} = \norm{\gammaCau{+} u}_{\mathscr{C}} \lesssim \norm{u}_{H^1 \oplus H^1}
\lesssim |\lambda|^2 \norm{\Phi_0}_{\mathscr{C}}.
\end{align*}

Putting everything together, suppose that $H_\lambda \Phi =\Phi_0$.
Put $\Phi_1 := P_{\lambda/\kappa_-}^+ \Phi$,
$\Phi_2 := P_{\lambda/\kappa_-}^{-} \Phi$.
By the intertwining relation (\ref{eq:calderon-intertwiner}), $\Phi_{1,2}$ solve $H_\lambda \Phi_1 = P_{\lambda/\kappa_+}^+ \Phi_0$ and $H_\lambda \Phi_2 = P_{\lambda/\kappa_+}^- \Phi_0$.

Thus we can apply estimates already established to $\Phi_1$ and $\Phi_2$:
\begin{align*}
\norm{\Phi}_{\mathscr{C}} = \norm{\Phi_1 + \Phi_2}_{\mathscr{C}}
\leq \norm{\Phi_1}_{\mathscr{C}} + \norm{\Phi_2}_{\mathscr{C}}
\lesssim |\lambda|^2 \norm{\Phi_0}_{\mathscr{C}} = |\lambda|^2 \norm{H_\lambda \Phi}_{\mathscr{C}}.
\end{align*}
\end{proof}

As a corollary, we obtain for $H_\lambda$ a similar estimate as for the other boundary layer operators in Proposition \ref{prop:boundary-op-bounds}.

\begin{corollary}
The operator $H_\lambda: \mathscr{C} \to \mathscr{C}$ is invertible, and for all $R> 0$, the operator
\begin{align*}
  &\lambda^{-2} H_\lambda^{-1}: \mathscr{C} \to \mathscr{C}
\end{align*}
is uniformly bounded for $\lambda\in \mathfrak{D}_\epsilon\setminus B_R(0)$.
\end{corollary}

\section{Relative trace formula.}

\subsection{Trace of the Neumann resolvent.}
\begin{proposition}[Krein type resolvent formula]
    Let $\Omega \subset \mathbb{R}^d$ be a Lipschitz domain (not necessarily connected).
    The resolvent with Neumann boundary conditions at $\partial\Omega$ is given by
    \begin{align}
        (-\Delta_N - \lambda^2)^{-1} - (-\Delta_0-\lambda^2)^{-1} = -\widetilde{D}_\lambda N_\lambda^{-1} \widetilde{D}_\lambda' : H^{-1}(\Omega^-) \oplus H^{-1}(\Omega^+) \to H^1(\Omega^-) \oplus H^1(\Omega^+).
    \end{align}
\end{proposition}
\begin{proof}
The 3rd Green identity (Lemma \ref{lem:third-green})
gives us that $u = (-\Delta_N - \lambda^2)^{-1} f$ is given by
\begin{align*}
    u = R_\lambda f + \widetilde{D}_\lambda [\gamma_D u] - \widetilde{S}_\lambda [\gamma_N u].
\end{align*}
Since $\gamma_N^\pm u = 0$, also the jump of the Neumann trace vanishes.
Hence we have
\begin{align*}
    0 = \gamma_N u = \widetilde{D}_\lambda' f + N_\lambda [\gamma_D u]
\end{align*}
which implies $[\gamma_D u] = - N_\lambda^{-1} \widetilde{D}_\lambda' f$
and therefore
\begin{align*}
    u = R_\lambda f - \widetilde{D}_\lambda N_\lambda^{-1} \widetilde{D}_\lambda' f.
\end{align*}
\end{proof}

Now consider the Neumann relative resolvent (we will drop the subscripts which only encumber the notation):
\begin{align*}
R_{\mathrm{rel}, \lambda}=\left(\left(-\Delta-\lambda^2\right)^{-1}-\left(-\Delta_0-\lambda^2\right)^{-1}\right)-\sum_{j=1}^M\left(\left(-\Delta_j-\lambda^2\right)^{-1}-\left(-\Delta_0-\lambda^2\right)^{-1}\right).
\end{align*}
Let us write
\begin{align*}
    N_\lambda = N_{\text{diag},\lambda} + T_\lambda,
\end{align*}
then by multiplying with $N_\lambda^{-1}$ from the left and $N_{\text{diag},\lambda}^{-1}$ from the right we obtain
\begin{align*}
    N_\lambda^{-1} - N_{\text{diag},\lambda}^{-1} = - N_{\text{diag},\lambda}^{-1} T_\lambda N_\lambda^{-1}.
\end{align*}
Thus we can estimate the trace norm of the left hand side by
\begin{align}
    \label{eq:hypersingular-trace-estimate}
\norm{T_\lambda}_{\Nuc} \leq C_1 e^{-C_2 \operatorname{Im}(\lambda)}
\end{align}
by the off diagonal estimate Lemma \ref{lem:trace-class-smoothing-ops} and the kernel estimate Corollary \ref{cor:kernel-pointwise-CN}.

In order to understand the pole structure at $\lambda = 0$, we need to investigate the action of $T_\lambda$ on constant boundary functions.
\begin{lemma}
\label{lem:pole-contribution-hypersingular}
Define $e_i: \partial\Omega_i \to \mathbb{R}$ by $e_i(x) = 1$.
Then
\begin{align*}
(T_\lambda e_j)(x) = \begin{cases} 0 & x \in \partial \Omega_j \\ -\lambda^2 \int_{\Omega_j} \partial_{\nu_x} G_\lambda(x,y) \mathrm{d} y & x \in \partial \Omega_k, k \neq j\end{cases}
\end{align*}
and
\begin{align*}
\langle e_i, T_\lambda e_j\rangle_{\partial\Omega} = \begin{cases} 0 & i=j\\ \lambda^4 \int_{\Omega_i \times \Omega_j} G_\lambda(x, y) \mathrm{d} x \mathrm{d} y & i \neq j.\end{cases}
\end{align*}
\end{lemma}
\begin{proof}
Let $u_j := \chi_{\Omega_j}$. Then $[\gamma_D u_j] = - e_j$ and $[\gamma_N u_j] = 0$.
Moreover $(-\Delta - \lambda^2) u_j = -\lambda^2 u_j$ on $\Omega_j$ and on $\Omega_j^+$.
By the 3rd Green identity we have
\begin{align*}
u_j = G_\lambda (-\lambda^2 u_j) - \widetilde{D}_\lambda e_j.
\end{align*}
If $x \in \Omega_i$ with $i \neq j$, then $u_j(x) = 0$ and therefore
\begin{align*}
\widetilde{D}_\lambda e_j(x) = -\lambda^2 \int_{\Omega_j} G_\lambda(x,y) \mathrm{d} y.
\end{align*}
Taking the Neumann trace we have $N_\lambda e_j(x) = -\lambda^2 \int_{\Omega_j} \partial_{\nu_x} G_\lambda(x,y) \mathrm{d} y$ for $x \in \partial \Omega_i$ with $i \neq j$.
Clearly $N_{\mathrm{diag}, \lambda} e_j = 0$ so the formula for $T_\lambda e_j$ follows.

We use this formula to compute
\begin{align*}
\langle e_i, T_\lambda e_j \rangle_{\partial \Omega} = -\lambda^2 \int_{\partial \Omega_i} \int_{\Omega_j} \partial_{\nu_x} G_\lambda(x,y) \mathrm{d} \sigma(x) \mathrm{d} y.
\end{align*}
Then by the divergence theorem we have
\begin{align*}
\langle e_i, T_\lambda e_j \rangle_{\partial \Omega} = -\lambda^2 \int_{\Omega_i \times \Omega_j} \Delta_x G_\lambda(x,y) \mathrm{d} x \mathrm{d} y = \lambda^4 \int_{\Omega_i \times \Omega_j} G_\lambda(x,y) \mathrm{d} x \mathrm{d} y.
\end{align*}
\end{proof}

We are now ready to show that $N_\lambda^{-1} - N_{\mathrm{diag}, \lambda}^{-1}$ is trace class even up to $\lambda = 0$.

\begin{proposition}
    \label{prop:hypersingular-inverse-trace-class}
The family of operators $\lambda \mapsto N_\lambda^{-1} - N_{\mathrm{diag}, \lambda}^{-1}$
is a holomorphic map $\mathfrak{D}_\epsilon \to \Nuc(H^{-1/2}(\partial \Omega), H^{1/2}(\partial \Omega))$
with trace-norm bound
\begin{align*}
\norm{N_\lambda^{-1} - N_{\mathrm{diag}, \lambda}^{-1}}_{\Nuc} \leq C_{\delta', \epsilon} e^{-\delta' \operatorname{Im}(\lambda)}.
\end{align*}
\end{proposition}
\begin{proof}
The formula
\begin{align*}
    N_{\text{diag},\lambda}^{-1} - N_\lambda^{-1} = N_{\text{diag},\lambda}^{-1} T_\lambda N_\lambda^{-1}.
\end{align*}
shows that for $\operatorname{Im}(\lambda) > 0$ (where $N_\lambda$ and $N_{\text{diag},\lambda}$ are invertible) the operator $N_\lambda^{-1} - N_{\text{diag},\lambda}^{-1}$ is trace class.

To expose the pole structure, we use the expansion of $(Q_\lambda^-)^{-1}$ from Proposition \ref{prop:expand-d-to-n-inv}.
We can write $N_\lambda$ both as $N_\lambda = -Q_\lambda^- (\frac 12 - D_\lambda)$ and as $N_\lambda = -(\frac12 - D_\lambda')Q_\lambda^-$.
Expanding $(Q_\lambda^-)^{-1} = \lambda^{-2} \Pi + M_\lambda$ (and remembering that $\Pi = \Pi_{\mathrm{diag}}$ since $\Pi$ is diagonal) therefore gives us
\begin{align*}
N_\lambda^{-1} &= -(\lambda^{-2} \Pi + M_\lambda)(\frac12 - D_\lambda')^{-1}\\
N_{\mathrm{diag}, \lambda}^{-1} &= - (\frac12 - D_{\mathrm{diag},\lambda})^{-1}(\lambda^{-2} \Pi + M_{\mathrm{diag},\lambda}).
\end{align*}

We therefore have
\begin{align*}
N_{\text{diag},\lambda}^{-1} - N_\lambda^{-1} = (\frac12 - D_{\mathrm{diag},\lambda})^{-1}\left[\lambda^{-4} \Pi T_\lambda \Pi + M_{\mathrm{diag}, \lambda} T_\lambda \Pi + \lambda^{-2} \Pi T_\lambda M_\lambda\right](\frac12 - D_\lambda')^{-1} + \mathrm{holomorphic}.
\end{align*}

We need to estimate the trace norm of the terms $\Pi T_\lambda \Pi$ and $T_\lambda \Pi$.
Since $T_\lambda^t = T_\lambda$ and $\Pi^t = \Pi$, we have $\norm{\Pi T_\lambda}_1 = \norm{T_\lambda \Pi}_1$,
so the estimate for $\Pi T_\lambda$ follows from the one for $T_\lambda \Pi$.

Let $P$ be the orthogonal projection onto the span of $\{e_1, \ldots, e_N\}$.
Then $P\Pi = \Pi P = \Pi$.
Then we can estimate 
\begin{align*}
    \norm{\Pi T_\lambda \Pi}_1 \leq \norm{\Pi}_1^2 \norm{P T_\lambda P}_{\mathrm{op}} \leq \norm{\Pi}_1^2 \sup_{i,j} |\langle e_i, T_\lambda e_j\rangle|.
\end{align*}

By Lemma \ref{lem:pole-contribution-hypersingular} and Corollary \ref{cor:kernel-pointwise-CN}, 
\begin{align*}
    \norm{\Pi T_\lambda \Pi}_1 \leq C \lambda^4 \sup_{i\neq j} \sup_{(x,y) \in \Omega_i\times\Omega_j} |G_\lambda(x,y)| \leq
    C \lambda^4 e^{-\delta' \operatorname{Im}(\lambda)}
\end{align*}
where $\delta' < \delta := \min_{i\neq j}\operatorname{dist}(\partial \Omega_i, \partial \Omega_j)$.

To estimate $\norm{T_\lambda \Pi}_1$ we use $\norm{T_\lambda \Pi}_1 \leq \norm{T_\lambda P}_{\mathrm{op}} \norm{\Pi}_1$,
together with
\begin{align*}
\norm{T_\lambda e_j}_{H^s(\partial \Omega_i)} \leq \lambda^2 \sup_{y \in \Omega_j} \norm{\gammaNeu{} G_\lambda(\cdot, y)}_{H^s(\partial \Omega_i)}
\end{align*}
We then use continuity of $\gammaNeu{}$ and the separated Sobolev estimates for $G_\lambda$ in Corollary \ref{cor:kernel-pointwise-sobolev}:
\begin{align*}
\norm{\gammaNeu{} G_\lambda(\cdot, y)}_{H^{-1/2}(\partial \Omega_i)} \leq \norm{G_\lambda(\cdot, y)}_{H^{1}(\Omega_i)}
+ \norm{\Delta G_\lambda(\cdot, y)}_{L^2(\Omega_i)}
\leq C_{\delta'} \lambda^2 e^{-\delta' \operatorname{Im}(\lambda)}, \quad i \neq j.
\end{align*}

Therefore every term in $N_{\mathrm{diag},\lambda}^{-1} - N_\lambda^{-1}$ has bounded trace norm near $\lambda = 0$ and hence $\lambda \mapsto N_\lambda^{-1} - N_{\mathrm{diag},\lambda}^{-1}$ extends to a holomorphic map $\mathfrak{D}_\epsilon \to \Nuc(H^{-1/2}(\partial \Omega), H^{1/2}(\partial \Omega))$.
\end{proof}

We can now rigorously define $\Xi_N$ and establish its basic properties.

\begin{lemma}
    \label{lem:xi-n-bounds}
The map $\mathfrak{D}_\epsilon \to \Nuc(H^{-1/2}(\partial \Omega))$, $\lambda \mapsto 1 - N_\lambda N_{\mathrm{diag},\lambda}^{-1}$ is holomorphic.

Thus the Fredholm determinant
\begin{align*}
\Xi_N(\lambda) = \log \det_{H^{-1/2}(\partial\Omega)} (N_\lambda N_{\mathrm{diag},\lambda}^{-1})
\end{align*}
is well defined, holomorphic on $\mathfrak{D}_\epsilon$ and
\begin{align*}
|\Xi_N(\lambda)| \leq C_{\delta', \epsilon} e^{-\delta' \operatorname{Im}(\lambda)} \qquad |\Xi_N'(\lambda)| \leq C_{\delta', \epsilon} e^{-\delta' \operatorname{Im}(\lambda)}.
\end{align*}
\end{lemma}
\begin{proof}
By Proposition \ref{prop:invert-hypersingular}, $N_\lambda N_{\mathrm{diag},\lambda}^{-1}$ is holomorphic as a map with values in bounded operators.
By Proposition \ref{prop:hypersingular-inverse-trace-class}, $1 - N_\lambda N_{\mathrm{diag},\lambda}^{-1} = N_\lambda (N_\lambda^{-1} -  N_{\mathrm{diag},\lambda}^{-1})$ is trace class with uniformly bounded trace norm on $\mathfrak{D}_\epsilon$.
By Lemma \ref{lem:trace-class-holo}, it therefore follows that $\lambda \mapsto 1 - N_\lambda N_{\mathrm{diag},\lambda}^{-1}$ is holomorphic as a trace-class valued map with
\begin{align*}
\norm{1 - N_\lambda N_{\mathrm{diag},\lambda}^{-1}}_{\Nuc{H^{-1/2}(\partial \Omega)}} \leq C_{\delta', \epsilon} e^{-\delta' \operatorname{Im}(\lambda)}.
\end{align*}

Thus the Fredholm determinant $\det(N_\lambda N_{\mathrm{diag},\lambda}^{-1})$ is well defined and holomorphic.
By invertibility of the operator in $\overline{\mathfrak{D}_\epsilon}$ the
determinant never vanishes 
\cite[Theorem 3.9]{simonNotesInfiniteDeterminants1977}
and therefore log det is analytic in the union of $\mathfrak{D}_\epsilon$ and a neighborhood of zero.

To estimate $\Xi_N$, note that since $\Xi_N$ is holomorphic we only need to prove the bound for $|\lambda| > R$ for some $R>0$.
Then $|1 - \det(N_\lambda N_{\mathrm{diag},\lambda}^{-1})| < \frac{1}{2}$ and therefore we can estimate $|\log(1+z)| \leq 2 |z|$:
\begin{align*}
|\Xi_N(\lambda)|=|\log (1+(\operatorname{det}(N_\lambda N_{\mathrm{diag},\lambda}^{-1})-1))| \leq 2|\operatorname{det}(N_\lambda N_{\mathrm{diag},\lambda}^{-1})-1| \leq 2 C e^{-\delta^{\prime} \operatorname{Im} \lambda} .
\end{align*}

To bound $\Xi_N'$, again we only need to consider $|\lambda| > R$.
That means that we can choose $r>0$ such that $B_r(\lambda) \subset \mathfrak{D}_{\epsilon/2}$ for all $\lambda \in \mathfrak{D}_\epsilon$ with $|\lambda| > R$.
By the maximum modulus principle,
\begin{align*}
|\Xi^{\prime}(\lambda)| \leq \frac{1}{r} \sup _{|z-\lambda|=r}|\Xi(z)| \leq \frac{C}{r} \sup _{|z-\lambda|=r} e^{-\delta^{\prime} \operatorname{Im} z}
\leq \frac{C e^{\delta^{\prime} r}}{r} e^{-\delta^{\prime} \operatorname{Im} \lambda}=C_{\delta^{\prime}, \varepsilon} e^{-\delta^{\prime} \operatorname{Im} \lambda}.
\end{align*}
\end{proof}

To conclude this section we establish the fundamental relation between $\Xi_N$ and the trace of the relative resolvent.
We first need to establish a trace bound on the relative resolvent.

\begin{proposition}
Let $\epsilon>0$, $R > 0$, let $\delta=\min_{j\neq k}\operatorname{dist}\left(\partial \Omega_j, \partial \Omega_k\right)$ and let $0 < \delta' < \delta$.
Then the operator $R_{N, \mathrm{rel}, \lambda}: L^2\left(\mathbb{R}^d\right) \rightarrow L^2\left(\mathbb{R}^d\right)$ is trace-class for all $\lambda \in \mathfrak{D}_\epsilon \setminus B_R(0)$ and the trace norm can be estimated by

\begin{align*}
\left\|R_{N, \mathrm{rel}, \lambda}\right\|_1\leq C_{\delta^{\prime}, \epsilon, R} e^{-\delta^{\prime} \operatorname{Im}(\lambda)}, \quad \lambda \in \mathfrak{D}_\epsilon.
\end{align*}
\end{proposition}
\begin{proof}
By the Krein relative resolvent formula we have
\begin{align}
    \label{eq:rel-neumann-resolvent}
    R_{\mathrm{rel}, \lambda} = \widetilde{D}_\lambda (N_{\text{diag}, \lambda}^{-1} - N_\lambda^{-1}) \widetilde{D}_\lambda'
\end{align}
and so we can estimate
\begin{align*}
\norm{R_{\mathrm{rel}, \lambda}}_1 \leq \norm{\widetilde{D}_\lambda}_{H^{1/2}(\partial \Omega) \to L^2(\mathbb{R}^d)}^2 \norm{N_{\text{diag}, \lambda}^{-1} - N_\lambda^{-1}}_1.
\end{align*}
The operator norm of $\widetilde{D}_\lambda$ is bounded by Lemma \ref{lem:bd-potential-bounds} and the trace of $N_{\text{diag}, \lambda}^{-1} - N_\lambda^{-1}$ is estimated by Proposition \ref{prop:hypersingular-inverse-trace-class}.
\end{proof}

\begin{lemma}
    The relative resolvent is trace class on $L^2(\mathbb{R}^d)$ and
    \begin{align}
        \tr_{L^2(\mathbb{R}^d)} (R_{\text{rel},\lambda}) = -\frac{1}{2\lambda}\Xi_N'(\lambda).
    \end{align}
\end{lemma}
\begin{proof}
By Lemma \ref{lem:fredholm-det-deriv}, $\Xi_N(\lambda)$ is differentiable for $\operatorname{Im}(\lambda) > 0$ and
\begin{align*}
\Xi_N'(\lambda) = \tr\left(\left(\frac{d}{d \lambda} N_\lambda\right)
     N_\lambda^{-1}-\left(\frac{d}{d \lambda} N_{\text{diag},\lambda}\right)
    N_{\text{diag},\lambda}^{-1}\right).
\end{align*}

Since $N_\lambda^{-1} - N_{\text{diag},\lambda}^{-1}$ is trace-class and $\widetilde{D}_\lambda$ is
bounded, we can take the trace of (\ref{eq:rel-neumann-resolvent}) and cyclically permute $\widetilde{D}_\lambda'$ to obtain
\begin{align*}
    \tr(R_{\text{rel},\lambda})
    =\tr \left(\widetilde{D}_\lambda' \widetilde{D}_\lambda  N_{\text{diag},\lambda}^{-1} - \widetilde{D}_\lambda'\widetilde{D}_\lambda N_{\lambda}^{-1} \right).
\end{align*}
Now $\widetilde{D}_\lambda^t\widetilde{D}_\lambda = \frac{1}{2\lambda}\frac{d}{d\lambda}N_\lambda$ (by Lemma \ref{lem:bd-op-derivative}), so we have:
\begin{align*}
    &\Tr(R_{\text{rel},\lambda})
    =\frac{1}{2 \lambda} \Tr\left(\left(\frac{d}{d \lambda} N_\lambda\right)
     N_{\text{diag},\lambda}^{-1}-\left(\frac{d}{d \lambda} N_{\lambda}\right)
    N_{\lambda}^{-1}\right)\\
    =\ &\frac{1}{2 \lambda} \Tr\left(\left(\frac{d}{d \lambda} N_{\mathrm{diag},\lambda}\right)
     N_{\text{diag},\lambda}^{-1}-\left(\frac{d}{d \lambda} N_{\lambda}\right)
    N_{\lambda}^{-1}\right)
    + \frac{1}{2 \lambda} \Tr\left(\left(\frac{d}{d\lambda}T_\lambda\right) N_{\mathrm{diag},\lambda}^{-1}\right)
\end{align*}
Since $T_\lambda$ is off-diagonal, i.e. $q_i T_\lambda q_i = 0$ for all $i$, we also have $q_i \dot{T}_\lambda q_i = 0$ where $\dot{T}_\lambda = \frac{d}{d\lambda} T_\lambda$.
Since $N_{\mathrm{diag},\lambda}^{-1}$ is diagonal, $\dot{T}_\lambda N_{\mathrm{diag},\lambda}^{-1}$ is off-diagonal and therefore the trace vanishes (cf. Remark \ref{rem:off-diag-trace-zero}).
We are left with
\begin{align*}
    \Tr(R_{\text{rel},\lambda})
    = -\frac{1}{2 \lambda} \Tr\left(\left(\frac{d}{d \lambda} N_\lambda\right)
     N_{\lambda}^{-1}-\left(\frac{d}{d \lambda} N_{\mathrm{diag}, \lambda}\right)
    N_{\text{diag},\lambda}^{-1}\right)
=-\frac{1}{2 \lambda} \Xi_N'(\lambda).
\end{align*}
\end{proof}


\section{Relative trace for the transmission problem.}

\begin{proposition}
    \label{prop:transmission-resolvent-krein-formula}
    The decomposition $L^2(\mathbb{R}^d) \simeq L^2(\Omega_-) \oplus L^2(\Omega_+)$ allows us to write the transmission resolvent as a matrix of operators
    \begin{align}
    (-\tm - \lambda^2)^{-1} 
    - \begin{pmatrix}
    \kappa_-^{-2}G_{\lambda/\kappa_-} & 0\\ 0 & \kappa_+^{-2} G_{\lambda/\kappa_+}
    \end{pmatrix}
    =
    \begin{pmatrix}
    B_{--} & B_{-+}\\ B_{+-} & B_{++}
    \end{pmatrix}
    \end{align}
    with
    \begin{align*}
    &B_{++} = -\kappa_+^{-2} \CLP_{\lambda/\kappa_+} \M H_\lambda^{-1} \CLP'_{\lambda/\kappa_+}\\
    &B_{+-} = \kappa_-^{-2} \CLP_{\lambda/\kappa_+} \M H_\lambda^{-1} \M\CLP'_{\lambda/\kappa_-}\\
    &B_{-+} = \kappa_+^{-2} \CLP_{\lambda/\kappa_-} H_\lambda^{-1} \CLP'_{\lambda/\kappa_+}\\
    &B_{- -} = -\kappa_-^{-2} \CLP_{\lambda/\kappa_-} H_{\lambda}^{-1} \M\CLP'_{\lambda/\kappa_-}.
    \end{align*}
    \end{proposition}
\begin{proof}
We solve (\ref{eq:transmission-problem-def}) for a source $f$ which is supported either in $\Omega^-$ or $\Omega^+$.
\\
\textbf{Exterior source.}\ 
First we solve the problem for an exterior source
\begin{align*}
    &\begin{cases}
        (-\Delta  - \kappa_+^{-2}\lambda^2) u_{+} = f/\kappa_+^2 \in \Omega^+\\
        \gamma^+ u_{+} = \M \Phi.
    \end{cases}
    &\begin{cases}
    (-\Delta - \kappa_-^{-2}\lambda^2)u_{-} = 0\qquad \text{in } \Omega^-\\
    \gamma^- u_{-} = \Phi
    \end{cases}
    \end{align*}

    We represent the solution in terms of the boundary data:
    \begin{align}
        \label{eq:trans-ext-1}
    &u_{+} = \kappa_+^{-2} G_{\lambda/\kappa_+} f + \CLP_{\lambda/\kappa_+} \M \Phi\\
    \label{eq:trans-ext-2}
    &u_- = -\CLP_{\lambda/\kappa_-} \Phi.
    \end{align}

By taking traces of Equations (\ref{eq:trans-ext-1}), (\ref{eq:trans-ext-2}), we obtain the boundary integral equations
\begin{align*}
    &\gammaCau{+}
    u_+
    =
    \kappa_+^{-2}\CLP'_{\lambda/\kappa_+} f
    + P_{\lambda/\kappa_+}^+ \M \Phi\\
    &\gammaCau{-}
    u_- = P_{\lambda/\kappa_-}^- \Phi
    \end{align*}
    The transmission boundary conditions 
    $\gammaCau{+}u_+=\M\gammaCau{-}u_-$
    give
\begin{align*}
\kappa_+^{-2} \CLP'_{\lambda/\kappa_+} f
= (\M P_{\lambda/\kappa_-}^-  - P_{\lambda/\kappa_+}^+ \M)\Phi
= - H_\lambda \Phi
\end{align*}
Therefore we have
\begin{align*}
\Phi = -\kappa_+^{-2} H_\lambda^{-1} \CLP'_{\lambda/\kappa_+} f.
\end{align*}

We conclude $u_+ = \kappa_+^{-2} G_{\lambda/\kappa_+} f + B_{++} f$, $u_- = B_{-+} f$ where
\begin{align*}
B_{++} = -\kappa_+^{-2}\CLP_{\lambda/\kappa_+} \M  H_\lambda^{-1} \CLP'_{\lambda/\kappa_+}, \qquad
B_{-+} = \kappa_+^{-2}\CLP_{\lambda/\kappa_-} H_\lambda^{-1} \CLP'_{\lambda/\kappa_+}.
\end{align*}\\
\textbf{Interior source.}\
Next we consider a source on the interior domain.
\begin{align*}
&\begin{cases}
    (-\Delta  - \kappa_+^{-2}\lambda^2) u_{+} = 0\qquad \text{in } \Omega^+\\
    \gamma^+ u_{+} = \M \Phi.
\end{cases}
&\begin{cases}
(-\Delta - \kappa_-^{-2}\lambda^2)u_{-} = \kappa_-^{-2}f\qquad \text{in } \Omega^-\\
\gamma^- u_- = \Phi.
\end{cases}
\end{align*}
We represent the solution as
\begin{align*}
&u_{+} =
\CLP_{\lambda/\kappa_+}\M\Phi\\
&u_- = \kappa_-^{-2} G_{\lambda/\kappa_-} f 
- \CLP_{\lambda/\kappa_-}\Phi.
\end{align*}

Now the boundary data satisfy
\begin{align*}
&\gammaCau{+}
u_+
=
P_{\lambda/\kappa_+}^+ \M \Phi\\
&\gammaCau{-}
u_- = \kappa_-^{-2} \CLP'_{\lambda/\kappa_-} f
+ P_{\lambda/\kappa_-}^- \Phi
\end{align*}
and therefore
\begin{align*}
    H_\lambda \Phi = P_{\lambda/\kappa_+}^+ \M \Phi - \M P_{\lambda/\kappa_-}^- \Phi
    = \kappa_-^{-2} \M \CLP'_{\lambda/\kappa_-} f.
\end{align*}
We conclude $u_+ = B_{+-} f$, $u_- = \kappa_-^{-2}G_{\lambda/\kappa_-} f + B_{--} f$ where
\begin{align*}
B_{- -} = -\kappa_-^{-2}\CLP_{\lambda/\kappa_-} H_\lambda^{-1} \M \CLP'_{\lambda/\kappa_-}, \qquad B_{+-} = \kappa_-^{-2} \CLP_{\lambda/\kappa_+} \M H_\lambda^{-1} \M \CLP'_{\lambda/\kappa_-}.
\end{align*}
\end{proof}

\begin{proposition}
    \label{prop:transmission-rel-resolvent-trace-estimate}
Let $\epsilon>0$, $R>0$, let $\delta=\min_{j\neq k}\operatorname{dist}\left(\partial \Omega_j, \partial \Omega_k\right)$ and let $0 < \delta' < \delta$.
Then the operator $R_{T, \mathrm{rel}, \lambda}: L^2\left(\mathbb{R}^d\right) \rightarrow L^2\left(\mathbb{R}^d\right)$ is trace-class for all $\lambda \in \mathfrak{D}_\epsilon \setminus B_R(0)$ and the trace norm can be estimated by

\begin{align*}
\left\|R_{T, \mathrm{rel}, \lambda}\right\|_1\leq C_{\delta^{\prime}, \epsilon, R} e^{-\delta^{\prime} \operatorname{Im}(\lambda)/\kappa}, \quad \lambda \in \mathfrak{D}_\epsilon \setminus B_R(0), \quad \kappa = \max(\kappa_-, \kappa_+).
\end{align*}
\end{proposition}
\begin{proof}
Write $G_\pm := \kappa_\pm^{-2} G_{\lambda/\kappa_\pm}$.
Recall from the introduction that we have the projections $p_i: H^s(\mathbb{R}^d) \to H^s(\Omega_i)$ and $q_i: H^s(\partial \Omega) \to H^s(\partial \Omega_i)$, and we define $p_+ := 1 - \sum_i p_i$.

Applying Proposition \ref{prop:transmission-resolvent-krein-formula} to the operator $-\Delta_{T, \Omega_i}$ yields
\begin{align*}
    (-\Delta_{T, \Omega_i} - \lambda^2)^{-1}
    = B_{\Omega_i} + p_iG_-p_i + p_+G_+p_+ + \sum_{k \neq i, l \neq i}p_k G_+p_l.
\end{align*}
Hence we can write the relative resolvent as
\begin{align*}
R_{T, \mathrm{rel}, \lambda} = R_\Omega - \sum_i R_i + B_\Omega - \sum_i B_{\Omega_i}.
\end{align*}
with
\begin{align*}
R_{\Omega_i} := (-\Delta_{T, \Omega_i} - \lambda^2)^{-1} - B_{\Omega_i} - G_+
= p_i (G_- - G_+) p_i
- p_+ G_+ p_i
- p_i G_+ p_+
\end{align*}
and
\begin{align*}
&R_\Omega :=(-\Delta_{T, \Omega} - \lambda^2)^{-1} - B_\Omega - G_+
= p_- (G_- - G_+) p_-
- p_+ G_+ p_- - p_- G_+ p_+
\end{align*}

From these formulas we get
\begin{align*}
R_\Omega - \sum_i R_{\Omega_i} =  \sum_{i\neq j}p_i(G_- - G_+)p_j
\end{align*}
By Lemma \ref{lem:trace-class-smoothing-ops}, $\norm{p_i G_\pm p_j}_1 \leq C_{\delta', \epsilon} e^{-\delta' \operatorname{Im} \lambda/ \kappa_\pm}$,
and therefore, defining $\kappa := \max(\kappa_-, \kappa_+)$, we have
\begin{align*}
    \norm{R_\Omega - \sum_i R_{\Omega_i}}_1 \leq C_{\delta', \epsilon} e^{-\delta' \operatorname{Im} \lambda / \kappa}.
\end{align*}
To estimate the trace norm of the relative resolvent, it remains to estimate $B_\Omega - \sum_i B_{\Omega_i}$.
Again Lemma \ref{lem:trace-class-smoothing-ops} gives us
\begin{align*}
\norm{p_i (B_\Omega - B_{\Omega_i}) p_j}_1 \leq C_{\delta', \epsilon} e^{-\delta' \operatorname{Im} \lambda / \kappa}
\end{align*}
so that we can concentrate on the diagonal elements.
We consider first the $++$-block:
\begin{align}
    \label{eq:pp-block-b-ops}
p_+ (B_\Omega - \sum_i B_{\Omega_i}) p_+ = -\kappa_+^{-2} p_+ \CLP_{\lambda/\kappa_+} \M (H_\lambda^{-1} - H_{\mathrm{diag},\lambda}^{-1}) \CLP'_{\lambda/\kappa_+} p_+.
\end{align}
By Lemma \ref{lem:bd-potential-bounds} $\norm{\CLP_{\lambda/\kappa_+}} \lesssim (1+ |\lambda/\kappa_+|^2)$.
As usual we can absorb this polynomial growth into the exponential decay.
By Lemma \ref{lem:trace-class-smoothing-ops},
\begin{align}
    \label{eq:hlambda-inv-diff-trace-estimate}
\norm{H_{\lambda}^{-1} - H_{\mathrm{diag},\lambda}^{-1}}_{\Nuc} \leq \norm{H_\lambda^{-1}}_{\mathrm{op}} \norm{H_{\mathrm{diag},\lambda} - H_\lambda}_{\Nuc} \norm{H_{\mathrm{diag},\lambda}^{-1}}_{\mathrm{op}}
\leq C_{\delta', \epsilon} e^{-\delta' \operatorname{Im} \lambda / \kappa}
\end{align}
and therefore
\begin{align*}
\norm{p_+ (B_\Omega - \sum_i B_{\Omega_i}) p_+}_1 \leq C_{\delta', \epsilon} e^{-\delta' \operatorname{Im} \lambda / \kappa}.
\end{align*}
Next consider the $k-k$-block:
\begin{align}
    \label{eq:kk-block-b-ops}
&p_k (B_\Omega - \sum_i B_{\Omega_i}) p_k = p_k (B_\Omega - B_{\Omega_k}) p_k - \sum_{i \neq k} p_k B_{\Omega_i} p_k\\
=\ &\kappa_-^{-2} p_k \CLP_{\lambda/\kappa_-} (H_\lambda^{-1} - (q_k H_\lambda q_k)^{-1}) \M \CLP'_{\lambda/\kappa_-} p_k
- \sum_{i \neq k} \kappa_+^{-2} p_k \CLP_{\lambda/\kappa_+} q_i \M (q_i H_\lambda q_i)^{-1} q_i\CLP'_{\lambda/\kappa_+} p_k.
\end{align}
Again by Lemma \ref{lem:trace-class-smoothing-ops},
\begin{align*}
\norm{p_k K_{\lambda/\kappa_+} q_i}_{\Nuc(\mathscr{C}, L^2(\mathbb{R}^d))} \leq C_{\delta', \epsilon} e^{-\delta' \operatorname{Im} \lambda / \kappa},
\end{align*}
which bounds the trace norm of the sum over $i\neq k$.

In the first term we can insert $\operatorname{id} = \sum_i q_i$:
\begin{align*}
\kappa_-^{-2} p_k \CLP_{\lambda/\kappa_-} (H_\lambda^{-1} - (q_k H_\lambda q_k)^{-1}) \M \CLP'_{\lambda/\kappa_-} p_k
= \kappa_-^{-2}  \sum_{i,j} p_k \CLP_{\lambda/\kappa_-} q_i (H_\lambda^{-1} - (q_k H_\lambda q_k)^{-1}) \M q_j \CLP'_{\lambda/\kappa_-} p_k
\end{align*}
The terms with $i \neq k$ or $j\neq k$ are again trace class because $p_k \CLP_{\lambda/\kappa_-} q_i$ and $q_j \CLP'_{\lambda/\kappa_-} p_k$ are trace class.

This leaves the $i=k$, $j=k$ term. Note that $\M$ commutes with $q_k$, so that we can write
\begin{align*}
\kappa_-^{-2} p_k \CLP_{\lambda/\kappa_-} q_k (H_\lambda^{-1} - (q_k H_\lambda q_k)^{-1}) q_k \M \CLP'_{\lambda/\kappa_-} p_k.
\end{align*}
The middle term is nothing other than $q_k (H_\lambda^{-1} - H_{\mathrm{diag},\lambda}^{-1}) q_k$ and by Equation (\ref{eq:hlambda-inv-diff-trace-estimate}) this finishes the estimate.
\end{proof}

\begin{proposition}
    \label{prop:transmission-rel-resolvent-trace}
For $\lambda \in \mathfrak{D}_\epsilon$, the transmission resolvent is trace class and
\begin{align}
\tr_{L^2(\mathbb{R}^d)} R_{\mathrm{rel},\lambda} = - \frac{1}{2\lambda}
\tr_{\mathcal{B}_{\lambda/\kappa_+}^-}  \left[\left(\frac{d}{d\lambda} H_\lambda\right) H_\lambda^{-1}-  \left(\frac{d}{d\lambda} H_{\mathrm{diag},\lambda}\right) H_{\mathrm{diag},\lambda}^{-1}\right].
\end{align}
where $\mathcal{B}_{\lambda/\kappa_+}^- = \operatorname{ran} P_{\lambda/\kappa_+}^-$ (cf. Lemma \ref{lem:calderon-projectors}).
\end{proposition}
\begin{proof}
The first step of the proof is to replace the relative resolvent by a sum of boundary correction terms $B_{\Omega_i}$.
This is not as straightforward as in the Dirichlet and Neumann cases, because the free resolvent terms do not all cancel.
It turns out that they cancel on the diagonal blocks, and the off-diagonal blocks do not contribute to the trace. 

Write $G_\pm := \kappa_\pm^{-2} G_{\lambda/\kappa_\pm}$.
The proof of Proposition \ref{prop:transmission-rel-resolvent-trace-estimate} shows that
\begin{align*}
(-\Delta_{T,\Omega} - \lambda^2)^{-1} - G_+ - \sum_i\left[(-\Delta_{T,\Omega_i} - \lambda^2)^{-1} - G_+ \right]
= B_{\Omega} - \sum_{i=1}^N B_{\Omega_i} + R_\Omega - \sum_i R_{\Omega_i}.
\end{align*}

To compute the trace, we will use an abbreviated notation.
The spectral parameter $\lambda$ will be suppressed entirely.
We will write $H_i = q_i H q_i$ for the transmission operator associated to $\Omega_i$.
Moreover we define
\begin{align*}
&K_{\pm} := K_{\lambda/\kappa_\pm}, \qquad K_{\pm, i} := K_{\pm} q_i\\
&L_+ := K_+^t p_+ K_+, \qquad
L_{+, d} := \bigoplus_i K_{+, i}^t (1 - p_i) K_{+, i}\\
&L_- := K_-^t p_- K_-, \qquad
L_{-, d} := \bigoplus_i K_{-, i}^t p_i K_{-, i}.
\end{align*}
Note that $L_{+, d}$ is \emph{not} the diagonal part of $L_+$.

It was shown in the proof of Proposition \ref{prop:transmission-rel-resolvent-trace-estimate} that $R_\Omega - \sum_i R_{\Omega_i}$ and $B_\Omega - \sum_i B_{\Omega_i}$ are trace class.
Since $R_\Omega - \sum_i R_{\Omega_i}$ consists of off-diagonal blocks, its trace vanishes.

More specifically, the proof also showed that the block terms in $B_\Omega - \sum_i B_{\Omega_i}$, Equations (\ref{eq:pp-block-b-ops}) and (\ref{eq:kk-block-b-ops}), are trace class.

The $++$-block contributes
\begin{align*}
&\tr_{L^2(\Omega^+)} \left[ p_+(B_\Omega - \sum_i B_{\Omega_i}) p_+\right] = -\kappa_+^{-2} \tr \left[p_+ K_+ \M (H^{-1} - H_{\diag}^{-1}) \CLP'_{\lambda/\kappa_+} p_+\right]\\
= &-\kappa_+^{-2} \tr_{L^2(\Omega^+)} \left[K_+^t \chi_{\Omega^+} K_+ \M (H^{-1} - H_{\diag}^{-1}) \right].
\end{align*}

For the $k-k$-block, we compute the individual terms
\begin{align*}
&p_k B_{\Omega} p_k=-\kappa_{-}^{-2} p_k K_{-} H^{-1} \M K_{-}^t p_k\\
&p_k B_{\Omega_k} p_k=-\kappa_{-}^{-2} p_k K_{-, k} H_k^{-1} \M K_{-, k}^t p_k\\
&p_k B_{\Omega_i} p_k=-\kappa_{+}^{-2} p_k K_{+, i} \M H_i^{-1} K_{+, i}^t p_k, \qquad i \neq k.
\end{align*}

Hence
\begin{align*}
&p_k\left(B_{\Omega}- B_{\Omega_k}\right) p_k
= -\kappa_{-}^{-2} p_k K_{-} H^{-1} \M K_{-}^t p_k+\kappa_{-}^{-2} p_k K_{-, k} H_k^{-1} \M K_{-, k}^t p_k\\
= &-\kappa_{-}^{-2} p_k K_{-} (H^{-1} - H_\diag^{-1}) \M K_{-}^t p_k
- \kappa_-^{-2} \sum_{i \neq k} p_k K_{-, i} H_i^{-1} \M K_{-, i}^t p_k.
\end{align*}
Summing over $k$ we get
\begin{align*}
\begin{aligned}
p_k\left(B_{\Omega}-\sum_i B_{\Omega_i}\right) p_k=-  \kappa_{-}^{-2} p_k K_{-}\left(H^{-1} - H_{\diag}^{-1}\right) \M K_{-}^t p_k &-\kappa_{-}^{-2} \sum_{i \neq k} p_k K_{-, i} H_i^{-1} \M K_{-, i}^t p_k \\
& +\kappa_{+}^{-2} \sum_{i \neq k} p_k K_{+, i} \M H_i^{-1} K_{+, i}^t p_k
\end{aligned}
\end{align*}

Both the first term on the right-hand side and the individual summands in the $i\neq k$ sums are trace class.
We can therefore compute the trace separately and permute cyclically under the trace to obtain
\begin{align*}
\begin{aligned}
 &\operatorname{Tr}_{L^2(\Omega_k)}\left[p_k\left(B_{\Omega}-\sum_i B_{\Omega_i}\right) p_k\right]
= -\kappa_{-}^{-2} \operatorname{Tr}_{\mathscr{C}}\left[K_{-}^t \chi_{\Omega_k} K_{-}\left(H^{-1} - H_{\diag}^{-1}\right) \M \right] \\
& -\kappa_{-}^{-2} \sum_{i \neq k} \operatorname{Tr}_{\mathscr{C}_i}\left[K_{-, i}^t \chi_{\Omega_k} K_{-, i} H_i^{-1} \M_i\right]
+\kappa_{+}^{-2} \sum_{i \neq k} \operatorname{Tr}_{\mathscr{C}_i}\left[K_{+, i}^t \chi_{\Omega_k} K_{+, i} \M H_i^{-1}\right] .
\end{aligned}
\end{align*}

Summing over $k$ gives
\begin{align*}
\begin{aligned}
& \sum_{k=1}^N \operatorname{Tr}_{L^2(\Omega_k)}\left[p_k\left(B_{\Omega}-\sum_i B_{\Omega_i}\right) p_k\right]
=  -\kappa_{-}^{-2} \operatorname{Tr}_{\mathscr{C}}\left[L_{-}\left(H^{-1} - H_{\diag}^{-1}\right) \M\right] \\
& -\kappa_{-}^{-2} \sum_i \operatorname{Tr}_{\mathscr{C}_i}\left[K_{-, i}^t (\chi_\Omega - \chi_{\Omega_i}) K_{-, i} H_i^{-1} \M\right]
 +\kappa_{+}^{-2} \sum_i \operatorname{Tr}_{\mathscr{C}_i}\left[K_{+, i}^t  (\chi_\Omega - \chi_{\Omega_i}) K_{+, i} \M H_i^{-1}\right] .
\end{aligned}
\end{align*}
We next represent the sums in terms of the operators $L_\pm, L_{\pm, d}$.

Because $H_\diag^{-1}$ is diagonal (and off-diagonal times diagonal is off-diagonal), only the diagonal boundary blocks of $L_{+}-L_{+, d}$ contribute to $\tr_{\mathscr{C}}\left(\left(L_{+}-L_{+,d}\right) H_\diag^{-1} \M\right)$.
The diagonal terms are precisely
\begin{align*}
q_i\left(L_{+}-L_{+, d}\right) q_i=-K_{+, i}^t (\chi_\Omega- \chi_{\Omega_i}) K_{+, i}.
\end{align*}
Therefore
\begin{align*}
    -\sum_i \tr_{\mathscr{C}_i}\left[K_{+, i}^t (\chi_\Omega- \chi_{\Omega_i}) K_{+, i} \M H_i^{-1}\right]
    = \tr_{\mathscr{C}}\left(\left(L_{+}-L_{+, d}\right) \M H_\diag^{-1}\right).
\end{align*}
Similarly, we find
\begin{align*}
\sum_i \tr_{\mathscr{C}_i}\left[K_{-, i}^t (\chi_\Omega - \chi_{\Omega_i}) K_{-, i} H_i^{-1} \M_i\right]
=\tr_{\mathscr{C}}\left(\left(L_{-}-L_{-, d}\right) H_\diag^{-1} \M\right).
\end{align*}
Summing up all contributions therefore gives
\begin{align*}
\begin{gathered}
\operatorname{Tr}_{\mathcal{H}}\left(B_{\Omega}-\sum_i B_{\Omega_i}\right)=-\kappa_{+}^{-2} \operatorname{Tr}_{\mathscr{C}}\left[L_{+} \M\left(H^{-1}-H_\diag^{-1}\right)+\left(L_{+}-L_{+, d}\right) \M H_\diag^{-1}\right] \\
-\kappa_{-}^{-2} \operatorname{Tr}_{\mathscr{C}}\left[L_{-}\left(H^{-1}-H_\diag^{-1}\right)\M+\left(L_{-}-L_{-, d}\right) H_\diag^{-1}\M\right] .
\end{gathered}
\end{align*}

By (\ref{eq:cald-potential-derivative}),
\begin{align*}
&\kappa_{+}^{-2} L_{+}=-\frac{1}{2 \lambda} \dot{A}_{+} P_{+}^{+}, \quad \kappa_{+}^{-2} L_{+, d}=-\frac{1}{2 \lambda} \dot{A}_{+, \diag} P_{+, \diag}^{+}\\
&\kappa_{-}^{-2} L_{-}=-\frac{1}{2 \lambda} \dot{A}_{-} P_{-}^{-}, \quad \kappa_{-}^{-2} L_{-, d}=-\frac{1}{2 \lambda} \dot{A}_{-, \diag} P_{-, \diag}^{-}.
\end{align*}
Therefore,
\begin{align*}
\operatorname{Tr}_{L^2\left(\mathbb{R}^d\right)}\left(B_{\Omega}-\sum_i B_{\Omega_i}\right)
=  \frac{1}{2 \lambda} \operatorname{Tr}_{\mathscr{C}}\big[\dot{A}_{+} P_{+}^{+} \M H^{-1}&-\dot{A}_{+, \diag} P_{+, \diag}^{+} \M H_\diag^{-1}\\
 &+ \dot{A}_{-} P_{-}^{-} H^{-1} \M-\dot{A}_{-, \diag} P_{-, \diag}^{-} H_\diag^{-1} \M\big] .
\end{align*}

By the definition of $H$ we have $P_+^+ \M = \M P_-^- + H$, and by
(\ref{eq:calderon-intertwiner}) we have $P_{+}^{\pm} H=H P_{-}^{\pm}$.
This allows us to move the Calderon projectors all the way to the right, giving
\begin{align*}
\tr_{\mathscr{C}}\left[\dot{A}_{+} \M H^{-1} P_{+}^{-}-\dot{A}_{+, \diag} \M H_\diag^{-1} P_{+, \diag}^{-}
+\M \dot{A}_{-} H^{-1} P_{+}^{-}-\M \dot{A}_{-, \diag} H_\diag^{-1} P_{-, \diag}^{-}\right]
+ \tr_{\mathscr{C}}\left(\dot{A}_{+}-\dot{A}_{+, \diag}\right).
\end{align*}
The last term is trace class by Lemma \ref{lem:trace-class-smoothing-ops} and the trace vanishes by Remark \ref{rem:off-diag-trace-zero}.
The remaining term is evidently
\begin{align*}
\tr_{L^2\left(\mathbb{R}^d\right)}\left(B_{\Omega}-\sum_i B_{\Omega_i}\right)=-\frac{1}{2 \lambda} \tr_{\mathscr{C}}\left[\dot{H} H^{-1} P_{+}^{-}-\dot{H}_\diag H_\diag^{-1} P_{+, \diag}^{-}\right]\\
=-\frac{1}{2 \lambda} \tr_{\mathscr{C}}\left[ \left(\dot{H} H^{-1}-\dot{H}_\diag H_\diag^{-1}\right) P_{+}^{-}\right]
= -\frac{1}{2 \lambda} \tr_{\mathcal{B}_{\lambda/\kappa_+}^-}\left[\dot{H} H^{-1}-\dot{H}_\diag H_\diag^{-1}\right] .
\end{align*}
Here we have used again that off-diagonal terms do not contribute to the trace, and therefore
\begin{align*}
\tr_{\mathscr{C} } \left[\dot{H}_\diag H_\diag^{-1} (P_{+, \diag}^- - P_{+}^{-})\right] = 0.
\end{align*}
\end{proof}



The Dirichlet and Neumann problems for $\Omega$ decouple into the problems for the individual $\Omega_i$.
Hence $P_{\lambda}^- = \sum_i P_{\lambda, i}^-$ and therefore the intertwining relation (\ref{eq:calderon-intertwiner}) holds also for $H_{\mathrm{diag},\lambda}$.
Therefore,
in the direct sum decomposition $\mathscr{C} = \mathcal{B}_{\lambda/\kappa_+}^- \oplus \mathcal{B}_{\lambda/\kappa_+}^+$, the operator $H_\lambda H_{\mathrm{diag},\lambda}^{-1}$ is block diagonal:
\begin{align*}
H_\lambda H_{\mathrm{diag},\lambda}^{-1} = \begin{pmatrix}
    H_\lambda^- (H_{\mathrm{diag},\lambda}^-)^{-1} & 0\\
    0 & H_\lambda^+ (H_{\mathrm{diag},\lambda}^+)^{-1}
\end{pmatrix}
\end{align*}
where $H_\lambda^- := P_{\lambda/\kappa_+}^- H_\lambda P_{\lambda/\kappa_-}^-$ and $H_{\mathrm{diag},\lambda}^- := P_{\lambda/\kappa_+}^- H_{\mathrm{diag},\lambda} P_{\lambda/\kappa_-}^-$.
The $+$ block is defined analogously.

\begin{proposition}
    \label{prop:transmission-xi-holo-decay}
\begin{enumerate}
\item The operator family $\lambda \mapsto \operatorname{id}_{\mathscr{C}} - H_{\lambda} H_{\mathrm{diag}, \lambda}^{-1}$ is a holomorphic map $\mathfrak{D}_\epsilon \to \Nuc(\mathscr{C})$.
\item The Fredholm determinant 
\begin{align*}
    \Xi_T(\lambda):=\log \det_{\mathcal{B}_{\lambda / \kappa_{+}}^{-}}\left(H_\lambda^{-}\left(H_{\mathrm {diag }, \lambda}^{-}\right)^{-1}\right)
\end{align*}
is holomorphic on $\mathfrak{D}_\epsilon$, and for $d\geq 3$ continuous on $\overline{\mathfrak{D}_\epsilon}$.
\item For $\lambda \in \mathfrak{D}_\epsilon$, $\Xi_T$ satisfies
\begin{align*}
\Xi_T'(\lambda) = -2\lambda \tr{} R_{\mathrm{rel},\lambda}.
\end{align*}
\item For $|\lambda| \geq R > 0$ we have the decay estimates
\begin{align}
|\Xi_T(\lambda)| \leq C_{\delta', \epsilon} e^{-\delta' \operatorname{Im}(\lambda)/\kappa}  \qquad |\Xi_T'(\lambda)| \leq C_{\delta', \epsilon} e^{-\delta' \operatorname{Im}(\lambda)/\kappa}.
\end{align}
If $d\geq 3$ then $\Xi_T, \Xi'_T$ are bounded near $\lambda=0$.
\end{enumerate}
\end{proposition}

\begin{proof}
\begin{enumerate}
\item $\operatorname{id} - H_\lambda H_{\mathrm{diag},\lambda}^{-1} = H_{\lambda} (H_{\lambda}^{-1} - H_{\mathrm{diag},\lambda}^{-1})$.
Since $\lambda \mapsto H_\lambda^{-1}$ is holomorphic $\mathfrak{D}_\epsilon \to \mathcal{L}(\mathscr{C})$, by (\ref{eq:hlambda-inv-diff-trace-estimate}) this is holomorphic with values in $\Nuc(\mathscr{C})$.
\item The operator on $\mathscr{C} = \mathcal{B}_{\lambda/\kappa_+}^- \oplus \mathcal{B}_{\lambda/\kappa_+}^+$ given by
\begin{align*}
\begin{pmatrix}
    H_\lambda^- (H_{\mathrm{diag},\lambda}^-)^{-1} & 0\\
    0 & \operatorname{id}
\end{pmatrix}
\end{align*}
is a trace-class perturbation of the identity:
\begin{align*}
\operatorname{id}_{\mathscr{C}} - \begin{pmatrix}
    H_\lambda^- (H_{\mathrm{diag},\lambda}^-)^{-1} & 0\\
    0 & \operatorname{id}
\end{pmatrix}
= 
\begin{pmatrix}
    \operatorname{id} - H_\lambda^- (H_{\mathrm{diag},\lambda}^-)^{-1} & 0\\
    0 & 0
\end{pmatrix}.
\end{align*}
Since this operator has trace $\tr(1 -  H_\lambda^- (H_{\mathrm{diag},\lambda}^-)^{-1})$ it follows that
$\det_{\mathcal{B}_{\lambda/\kappa_+}^-} (H_\lambda^- (H_{\mathrm{diag},\lambda}^-)^{-1}) = \det_{\mathscr{C}} (H_\lambda^- (H_{\mathrm{diag},\lambda}^-)^{-1} \oplus \operatorname{id})$.
Since the latter operator is holomorphic on $\mathfrak{D}_\epsilon$ and continuous on $\mathfrak{D}_\epsilon$ it follows by the same arguments as the Neumann case that $\Xi_T$ is holomorphic on $\mathfrak{D}_\epsilon$ and continuous on $\overline{\mathfrak{D}_\epsilon}$.
\item The determinant is now over the subspace $\mathcal{B}^-_{\lambda/\kappa_+}$, which depends on $\lambda$.
However, this does not impede the application of Jacobi's formula (\ref{eq:fredholm-jacobi-formula}).

We have holomorphic fibre isomorphisms
\begin{align*}
U(\lambda, \lambda_0): \mathcal{B}^-_{\lambda_0/\kappa_+} \to \mathcal{B}^-_{\lambda/\kappa_+},\quad
U(\lambda,\lambda_0) = P_{\lambda/\kappa_+}^+ P_{\lambda_0/\kappa_+}^+ + P_{\lambda/\kappa_+}^- P_{\lambda_0/\kappa_+}^-.
\end{align*}

Let $g(\lambda) = H_\lambda^- (H_{\mathrm{diag},\lambda}^-)^{-1}$.
Then
$\widetilde{g}(\lambda) := U(\lambda, \lambda_0)^{-1} g(\lambda) U(\lambda, \lambda_0)$ is a holomorphic function $\mathfrak{D}_\epsilon \to \mathcal{L}(\mathcal{B}^-_{\lambda_0/\kappa_+})$ and
\begin{align*}
\frac{\mathrm{d}}{\mathrm{d} \lambda} \det g(\lambda) = \frac{\mathrm{d}}{\mathrm{d} \lambda} \det \widetilde{g}(\lambda) = \tr(\widetilde{g}(\lambda)^{-1} \widetilde{g}'(\lambda)) = \tr (g(\lambda)^{-1} g'(\lambda)).
\end{align*}

Hence the formula follows from Proposition \ref{prop:transmission-rel-resolvent-trace}.
\item Using Proposition \ref{prop:transmission-rel-resolvent-trace-estimate} the argument is identical to the Neumann case.
\end{enumerate}
\end{proof}

\section{Proof of main theorems.}
\label{sec:proof-main-thm}

In this section we will put together all the results of the previous sections.
In order to obtain the trace formula, i.e. to prove Theorems \ref{thm:dirichlet-trace-formula}, \ref{thm:neumann-trace-formula}, \ref{thm:transmission-trace-formula}, all that is missing is a functional calculus argument which is the same for all boundary conditions;
it depends only on the formal properties of the relative resolvent and the function $\Xi$.
The proof is essentially the same as Theorem 1.3 \cite{strohmaierRelativeTraceFormula2021} and Theorem 1.4 \cite{HSW}, but we include it here for completeness.

\begin{proof}
In principle the proof is a simple calculation with the holomorphic functional calculus, integrating along the boundary of the sector $\mathfrak{D}_\epsilon$.
The only challenge is that the function $f(\lambda) = (\lambda^2 + m^2)^s$ does not decay at infinity, and therefore we cannot apply it directly to $-\Delta$.
Instead, we have to introduce a regulator, which we can remove as soon as we have made the subtractions that define the relative operator.

Let $A \in \{-\Delta_D, -\Delta_N, -\Delta_T\}$ and let $A_0$ be the corresponding comparison operator; for the Dirichlet and Neumann cases $A_0 = -\Delta_0$ and for the transmission case $A_0 = -\kappa_+^2 \Delta_0$.
Write $A_i$ for the operator with boundary conditions only at $\partial\Omega_i$ and let $R_{\mathrm{rel},\lambda}$ be the relative resolvent corresponding to $A$.

Let $s \in \mathbb{C}$ with $\operatorname{Re}(s) > 0$, $m\geq 0$, and set $f(z) = (z + m^2)^s - m^{2s}$.
Then $f(z) = O(|z|)$ as $z \to 0$ if $m > 0$ and $f(z) = O(|z|^s)$ as $z \to 0$ if $m=0$.
Define the relative operator
\begin{align*}
D_f := f(A) - f(A_0) - \sum_{i=1}^M \left[f(A_i) - f(A_0)\right].
\end{align*}
Note that the subtracted $m^{2s}$ cancels in the relative operator, so $D_f$ is indeed the correct operator to consider.

Let $\gamma_\epsilon$ be the boundary of the sector $\mathfrak{S}_\epsilon := \{z \in \mathbb{C} \mid |\arg z | < \epsilon\}$
and $\Gamma_\epsilon$ the boundary of the $\mathfrak{D}_\epsilon$ sector (both oriented counterclockwise).
Then $\gamma_\epsilon$ is the image of $\Gamma_\epsilon$ under the map $\lambda \mapsto \lambda^2$.
Hence for any operator $A$ with $\sigma(A) \subset [0, \infty)$ and any admissible function $f$, we have
\begin{align*}
f(A) = \frac{1}{2\pi i}\int_{\gamma_\epsilon} \frac{f(z)}{A - z} \mathrm{d} z = \frac{1}{i \pi} \int_{-\Gamma_\epsilon} \frac{\lambda f(\lambda^2)}{A - \lambda^2} \mathrm{d} \lambda = \frac{\rmi}{\pi} \int_{\Gamma_\epsilon} \frac{\lambda f(\lambda^2)}{A - \lambda^2} \mathrm{d} \lambda.
\end{align*}

The function $f_n(z) = f(z) e^{-\frac{1}{n} z}$ is holomorphic in the sector $\mathfrak{S}_\epsilon$ and decays exponentially.
It is therefore an admissible function for the Riesz-Dunford functional calculus.
Thus
$$
D_{f_n}=\frac{\mathrm{i}}{\pi} \int_{\tilde{\Gamma}_\epsilon} \lambda f_n(\lambda^2) R_{\mathrm{rel}, \lambda} \,\mathrm{d} \lambda.
$$
By the trace norm bound $\norm{R_{\mathrm{rel}, \lambda}}_1 \leq C_{\delta', \epsilon} e^{-\delta' \operatorname{Im}(\lambda)}$ the integral converges as a Bochner integral in trace class operators.
By the dominated convergence theorem for the Bochner integral, $\lim_{n \to \infty} D_{f_n} = D_f$ in trace norm, i.e.
$$
D_f=\frac{\mathrm{i}}{\pi} \int_{{\Gamma}_\epsilon} \lambda f(\lambda^2) R_{\mathrm{rel}, \lambda} \,\mathrm{d} \lambda
$$
Therefore $D_f$ is trace-class and the trace commutes with the integral.
Since $\tr R_{\mathrm{rel}, \lambda}=-\frac{1}{2\lambda}\Xi'(\lambda)$ we have
\begin{align*}
\tr\left(D_f\right)=\frac{\mathrm{i}}{\pi} \int_{{\Gamma}_\epsilon} \lambda f(\lambda^2) \tr\left(R_{\mathrm{rel}, \lambda}\right) \,\mathrm{d} \lambda=-\frac{\mathrm{i}}{2 \pi} \int_{{\Gamma}_\epsilon} f(\lambda^2) \Xi^{\prime}(\lambda) \,\mathrm{d} \lambda=\frac{\mathrm{i}}{\pi} \int_{{\Gamma}_\epsilon} \lambda f'(\lambda^2) \Xi(\lambda) \,\mathrm{d} \lambda.
\end{align*}

Now $f'(z) = s(z + m^2)^{s-1}$, so $\lambda f'(\lambda^2) = s\lambda(\lambda^2+m^2)^{s-1} =: g(\lambda)$.
Since we have fixed the principal branch of the logarithm the branch cut of $g$ is the ray $i[m, \infty)$.
At the branch cut, $g$ has left and right sided limits
\begin{align*}
    g_-(\rmi  t) &=  \rmi  st (t^2 - m^2)^{s-1} e^{\rmi \pi(s-1)}\\
    g_+(\rmi t) &=  \rmi  st (t^2 - m^2)^{s-1} e^{-\rmi \pi(s-1)}, \qquad t^2 > m^2.
\end{align*}
Hence the jump is $g_+(\rmi  t) - g_-(\rmi  t) = 2 s \sin(\pi s) t (t^2 - m^2)^{s-1}$.
Deforming the contour $\Gamma_\epsilon$ to hug the branch cut we arrive at the result
\begin{align*}
\tr D_f = \frac{2 s}{\pi} \sin (\pi s) \int_m^{\infty} t\left(t^2-m^2\right)^{s-1} \Xi(\rmi  t) \,\mathrm{d} t.
\end{align*}
\end{proof}

\section{The case \(d=2\)}
\label{sec:dimension-two}

The case of dimension $d=2$ requires special care because $S_\lambda$ is no longer Hahn holomorphic at $\lambda=0$.
Instead, there is a $\log \lambda$ singularity.
This singularity vanishes on the solution space $\mathcal{B}_\lambda^-$ and therefore the Fredholm determinant $\Xi_T$ is still continuous in a neighbourhood of $\lambda=0$.

Let us first investigate the structure of the logarithmic singularity.
Suppose that we have $M$ connected components, $\Omega = \bigsqcup_{i=1}^M \Omega_i$.
Write $e_i = 1_{\partial \Omega_i}$ for the characteristic function of the $i$-th boundary component.
We define a charge map
\begin{align*}
V_j: H^{-1 / 2}(\partial \Omega) \to \mathbb{C},\qquad
V_j(\varphi) =\int_{\partial \Omega_j} \varphi \,\mathrm{d} \sigma
\end{align*}
and an operator
\begin{align*}
R: H^{-1 / 2}(\partial \Omega) \to H^{1 / 2}(\partial \Omega),\qquad
R(\varphi) = 1_{\partial\Omega} \int_{\partial\Omega} \varphi\, \mathrm{d} \sigma = \left(\sum_{i=1}^M e_i\right)\left(\sum_{j=1}^M V_j(\varphi)\right) =: \sum_{i,j=1}^M R_{ij}(\varphi).
\end{align*}

Let us summarize the expansion of the Hankel functions in $d=2$, in the form that we need.

\begin{lemma}
Let $d=2$.
Define $G_0(x,y) := -\frac{1}{2\pi} \log |x-y|$.
For $\lambda \in \mathbb{C}^+$ near $0$, we have, as operators $H^{s}_{\comp}(\mathbb{R}^d) \to H^{s+2}_{\loc}(\mathbb{R}^d)$,
\begin{align*}
G_\lambda = G_0 - \frac{1}{2\pi} \log \lambda + O(\lambda^2 \log \lambda).
\end{align*}

Define $S_0 := \gammaDir{} G_0 \gammaDir{*}$.
As operators $H^{-1/2}(\partial \Omega) \to H^{1/2}(\partial \Omega)$ we have
\begin{align*}
S_\lambda = S_0 - \frac{1}{2\pi} \log \lambda R + O(\lambda^2 \log \lambda).
\end{align*}
\end{lemma}

\begin{lemma}
    \label{lem:holomorphic-charge}
In a neighbourhood of $\lambda=0$, there is a holomorphic operator family $A_j(\lambda): H^{1/2}(\partial\Omega) \to \mathbb{C}$, $j=1,\ldots, M$ such that 
\begin{align*}
V_j Q_\lambda^- = \lambda^2 A_j(\lambda).
\end{align*}
\end{lemma}
\begin{proof}
Let $f \in H^{1/2}(\partial\Omega)$ and let $u_j$ be the solution of the Dirichlet problem on $\Omega_j$, i.e. the unique solution of
\begin{align*}
\left(-\Delta-\lambda^2\right) u_j=0 \quad \text { in } \Omega_j, \quad \gammaDir{-} u_j=f.
\end{align*}
By the divergence theorem,
\begin{align*}
V_j(Q_\lambda^{-} f)=\int_{\partial \Omega_j} \gammaNeu{-} u_j \,\mathrm{d} \sigma=\int_{\Omega_j} \Delta u_j \,\mathrm{d} x=-\lambda^2 \int_{\Omega_j} u_j \,\mathrm{d} x.
\end{align*}
The solution $u_j$ still depends on $\lambda$.
To see that this dependence is holomorphic, let $P_0f$ be the harmonic extension of $f$ to $\Omega_j$.
Then we can write $u_j$ as
\begin{align*}
u_j = P_0 f + \lambda^2 (-\Delta_D - \lambda^2)^{-1} P_0 f
\end{align*}
which is holomorphic in $\lambda$.
\end{proof}

\begin{lemma}
    \label{lem:holomorphic-single-layer}
The operator families on $H^{1/2}(\partial \Omega)$,
\begin{align}
S_{\lambda/\kappa_+} Q_{\lambda/\kappa_-}^-, \qquad
\left( \frac{d}{d\lambda} S_{\lambda/\kappa_+}\right) Q_{\lambda/\kappa_-}^-, \qquad
S_{\lambda/\kappa_+} \frac{d}{d\lambda} Q_{\lambda/\kappa_-}^-,
\end{align}
are Hahn holomorphic in a neighbourhood of $\lambda=0$.
\end{lemma}
\begin{proof}
The only potential singular part in $S_{\lambda/\kappa_+} Q_{\lambda/\kappa_-}^-$ is $\log \lambda R Q_{\lambda/\kappa_-}^-$.
By Lemma \ref{lem:holomorphic-charge} this is actually $O(\lambda^2 \log \lambda)$ and therefore Hahn holomorphic.

Note that the higher order terms in the expansion of $S_\lambda$ are at least $O(\lambda^2 \log \lambda)$.
After differentiating, the leading order is $\lambda\log\lambda$, which is still Hahn holomorphic.
Hence the only possibly singular term is $\lambda^{-1} R Q_{\lambda/\kappa_-}^-$ which by the same lemma is $O(\lambda)$ and therefore Hahn holomorphic.
\end{proof}

Whenever $\lambda$ is not a Dirichlet eigenvalue, we have an isomorphism
\begin{align*}
J_\lambda^{-}: H^{1 / 2}(\partial \Omega) \to \mathcal{B}_\lambda^{-}, \quad J_\lambda^{-} f=\begin{pmatrix}f \\ Q_\lambda^{-} f\end{pmatrix}.
\end{align*}
The inverse is simply the projection onto the first component.
This allows us to transport $H_\lambda$ to $H^{1/2}(\partial \Omega)$.
Define
\begin{align*}
\widetilde{H}_\lambda := J_{\lambda/\kappa_+}^{-1} H_\lambda J_{\lambda/\kappa_-}, 
\qquad
\widetilde{H}_{\mathrm{diag}, \lambda} := J_{\lambda/\kappa_+}^{-1} H_{\mathrm{diag}, \lambda} J_{\lambda/\kappa_-}.
\end{align*}
We may write these operators explicitly as
\begin{align*}
\widetilde{H}_\lambda&=-\nu_0\left(\frac{1}{2}-D_{\lambda / \kappa_{+}}\right)-\nu_1 S_{\lambda / \kappa_{+}} Q_{\lambda / \kappa_{-}}^{-}, \\
\widetilde{H}_{\mathrm {diag }, \lambda}&=-\nu_0\left(\frac{1}{2}-D_{\mathrm {diag }, \lambda / \kappa_{+}}\right)-\nu_1 S_{\mathrm {diag }, \lambda / \kappa_{+}} Q_{\lambda / \kappa_{-}}^{-} .
\end{align*}
From Lemma \ref{lem:holomorphic-single-layer} it follows that these are Hahn holomorphic operators with Hahn holomorphic derivatives near $\lambda=0$.

\begin{lemma}
The operator $\widetilde{H}_\lambda^{-1}$ is holomorphic on $\mathbb{C}^+$ and Hahn holomorphic in a neighbourhood of $\lambda=0$.
\end{lemma}
\begin{proof}

Write $H_\lambda = H_{\lambda}^{\mathrm{reg}} + W_\lambda$ where $W_\lambda$ is Hahn meromorphic and finite rank.
Since $H_\lambda$ is Fredholm of index $0$ for $\lambda \in \mathbb{C}^+$, the same is true for $H_\lambda^{\mathrm{reg}}$.
Since $H_\lambda^{\mathrm{reg}}$ is Hahn holomorphic near $\lambda=0$, it follows that $H_\lambda^{\mathrm{reg}}$ is also Fredholm of index $0$ at $\lambda=0$.
The same is true for $\widetilde{H}_\lambda$.

Now, $H_0:= \lim_{\lambda\to 0} H_\lambda^{\mathrm{reg}}$ is the transmission operator at $\lambda = 0$ defined by taking $G_0(x,y) = \frac{1}{2\pi} \log|x-y|$ as the Green function.
This is the definition of the $\lambda=0$ operators in \cite{costabelDirectBoundaryIntegral1985}.
Therefore we conclude from Proposition 4.3 of \cite{costabelDirectBoundaryIntegral1985} that $H_0$ is injective.

Now observe that $\widetilde{H}_0 = H_0$.
Since $\widetilde{H}_\lambda$ is Hahn holomorphic, $\widetilde{H}_\lambda^{-1}$ is Hahn meromorphic of finite type.
Since $H_0$ is injective, comparison of coefficients shows that $\widetilde{H}_\lambda^{-1}$ is Hahn holomorphic at $\lambda=0$.
\end{proof}

\begin{proposition}
In $d=2$, the functions $\Xi_T, \Xi_T'$ are continuous in a neighbourhood of $\lambda = 0$.
\end{proposition}
\begin{proof}
For $\lambda \in \mathbb{C}^+$, $\Xi_T$ is the fixed space Fredholm determinant
\begin{align*}
\Xi_T(\lambda) = \log \det_{H^{1/2}(\partial \Omega)} \left(\widetilde{H}_\lambda \widetilde{H}_{\mathrm{diag}, \lambda}^{-1}\right).
\end{align*}
For the continuity of $\Xi_T$ it therefore suffices to observe that
\begin{align*}
1 - \widetilde{H}_\lambda \widetilde{H}_{\mathrm{diag}, \lambda}^{-1} = \left(\widetilde{H}_{\mathrm{diag}, \lambda} - \widetilde{H}_\lambda\right) \widetilde{H}_{\mathrm{diag}, \lambda}^{-1}
\end{align*}
is a Hahn holomorphic family of trace class operators.

For the derivative we have instead
\begin{align*}
\Xi_T'(\lambda) &=  \tr_{H^{1/2}(\partial\Omega)} \left[ \left(\frac{d}{d\lambda} \widetilde{H}_\lambda\right) \widetilde{H}_\lambda^{-1} - \left(\frac{d}{d\lambda} \widetilde{H}_{\mathrm{diag}, \lambda}\right) \widetilde{H}_{\mathrm{diag}, \lambda}^{-1} \right]\\
& = \tr \left[\left(\frac{d}{d \lambda}\left(\widetilde{H}_\lambda-\widetilde{H}_{\mathrm{diag}, \lambda}\right)\right) \widetilde{H}_\lambda^{-1}+\left(\frac{d}{d \lambda} \widetilde{H}_{\lambda, \mathrm{diag}}\right)\left(\widetilde{H}_\lambda^{-1}-\widetilde{H}_{\mathrm{diag}, \lambda}^{-1}\right) \right]
\end{align*}
and the argument of the trace is manifestly a Hahn holomorphic family of trace class operators.
\end{proof}

\begin{proposition}
In $d=2$, $\Xi_N(\lambda)$ is bounded near $\lambda=0$ and $|\Xi_N'(\lambda)| \leq C|\log \lambda|$.
\end{proposition}
\begin{proof}
The only thing that changes in $d=2$ is that $N_\lambda^{-1} - N_{\mathrm{diag}, \lambda}^{-1}$ is no longer Hahn holomorphic, but has a $\log \lambda$ singularity.
Recall that with $T_\lambda = N_\lambda - N_{\mathrm{diag}, \lambda}$ we compute
\begin{equation}
        \langle e_i,T_\lambda e_j\rangle_{\partial\Omega}
        =\lambda^4\int_{\Omega_i\times\Omega_j}G_\lambda(x,y)\,dx\,dy.
        \label{eq:d2-neumann-double-projection}
\end{equation}
The $\lambda^4$ factor remains, but the estimate of the integral picks up a $\log\lambda$ in $d=2$:
\begin{align*}
        \lambda^{-4}\langle e_i,T_\lambda e_j\rangle_{\partial\Omega}
        =-\frac{|\Omega_i||\Omega_j|}{2\pi}\log\lambda+O(1).
\end{align*}
Hence $\norm{N_\lambda^{-1} - N_{\mathrm{diag}, \lambda}^{-1}}_{\Nuc} \leq C |\log\lambda|$ near $\lambda =0$ and therefore the same applies to $|\Xi'_N(\lambda)|$.
\end{proof}

\section{Computations for parallel plates.}

The case of parallel plates occupies a distinguished place in the theoretical physics literature on the Casimir effect due to its computational accessibility.
Parallel plates are non-compact obstacles, and therefore our approach must be modified to handle this case.
It was shown in \cite{strohmaierDimensionalReductionFormulae2024} (for Dirichlet boundary conditions) that a partial Fourier transform along the plate directions affords a spectral decomposition which reduces the problem to a 1-dimensional one.
The reduced obstacles are points on the line, and therefore compact.
Here we sketch how to extend this result to transmission boundary conditions.

\subsection{Neumann boundary conditions.}
As an introduction to the method, we will consider the case of Neumann boundary conditions, which is exactly analogous to the Dirichlet case treated in \cite{strohmaierDimensionalReductionFormulae2024}.

Let $\Omega = \Omega_1 \sqcup \Omega_2 \subset \mathbb{R}$ be a disjoint union of closed intervals.
In $\mathbb{R}^3$ our obstacles have the form $\widetilde{\Omega}_i = \Omega_i \times \mathbb{R}^2$.
We denote coordinates on the line $\mathbb{R}$ by $x$ and coordinates on the plates $\mathbb{R}^2$ by $y = (y_1, y_2)$.
We will use tildes to denote objects on $\mathbb{R}^3$; operators without tilde refer to the corresponding 1-dimensional objects.

The partial Fourier transform is the unitary isomorphism
\begin{align*}
L^2(\mathbb{R}^3) \to L^2(\mathbb{R}, L^2(\mathbb{R}^2))\\
\widehat{f}(x)(\xi) = \frac{1}{2\pi} \int_{\mathbb{R}^2} f(x,y) e^{-i y \cdot \xi} \mathrm{d} y.
\end{align*}

Let $\widetilde{\Delta} = \partial_x^2 + \partial_{y_1}^2 + \partial_{y_2}^2$ be the Laplace operator on $\mathbb{R}^3$, and let $\widetilde{R}_{s,N}$ be the relative operator
\begin{align*}
\widetilde{R}_{s,N} = (-\widetilde{\Delta}_{N}+m^2)^{s} - (-\widetilde{\Delta}_{N,1}+m^2)^{s} - (-\widetilde{\Delta}_{N,2}+m^2)^{s} + (-\widetilde{\Delta}_{0}+m^2)^{s}.
\end{align*}
Conjugating with the partial Fourier transform maps $\widetilde{R}_{s,N}$ to
\begin{align*}
R_{s,N}(\xi)= \left(-\Delta_{N}+m^2+\xi^2\right)^{s} + \left(-\Delta_0+m^2+\xi^2\right)^{s} 
- \left(-\Delta_{N,1}+m^2+\xi^2\right)^{s}-\left(-\Delta_{N,2}+m^2+\xi^2\right)^{s}.
\end{align*}

This operator again has a representation in terms of the 1-dimensional relative resolvent:
\begin{align}
    \label{eq:1d-neumann-resolvent-rep}
R_{s,N}(\xi)=\frac{\mathrm{i}}{\pi} \int_{\tilde{\Gamma}} \lambda\left(\left(\lambda^2+m^2+\xi^2\right)^{s}-\left(m^2+\xi^2\right)^{s}\right) R_{\mathrm{rel},N}(\lambda) \mathrm{d} \lambda
\end{align}

The relative trace (in a von Neumann algebra sense; see \cite{strohmaierDimensionalReductionFormulae2024} for fuller explanations) is computed as
\begin{align*}
&\operatorname{tr}\left(\widetilde{R}_{s,N}\right)=\frac{1}{(2 \pi)^2} \int_{\mathbb{R}^2} \operatorname{tr}\left(R_{s,N}(\xi)\right) \mathrm{d} \xi\\
=\ &\frac{\mathrm{i}}{(2\pi)^2\pi} \int_{\tilde{\Gamma}} \lambda\left(\left(\lambda^2+m^2+\xi^2\right)^{s}-\left(m^2+\xi^2\right)^{s}\right) \operatorname{tr}\left(R_{\text {rel }, N}(\lambda)\right) \mathrm{d} \lambda \mathrm{d} \xi.
\end{align*}
Using that, for sufficiently negative $\operatorname{Re} s$, we have
\begin{align*}
\int_{\mathbb{R}^2}\left(\lambda^2+m^2+\xi^2\right)^{s} d \xi=\pi \frac{\Gamma\left(-s-1\right)}{\Gamma\left(-s\right)}\left(\lambda^2+m^2\right)^{1+s},
\end{align*}
one can do the $\xi$-integral to obtain
\begin{align*}
\operatorname{tr}\left(\widetilde{R}_{s,N}\right)=-\frac{1}{2\pi\Gamma\left(-s\right) \Gamma\left(1+s\right)} \int_m^{\infty} \lambda\left(\lambda^2-m^2\right)^{s} \Xi_N(i \lambda) \mathrm{d} \lambda.
\end{align*}
Here $\Xi_N$ is defined through the 1-dimensional relative resolvent.

The kernel of the free resolvent $(-\Delta_0 - \lambda^2)^{-1}$ in one dimension is (with $\operatorname{Im} \lambda > 0$)
\begin{align*}
G_{\lambda}(x,y) = \frac{\rmi}{2\lambda} e^{i\lambda|x-y|}.
\end{align*}
Let $a > 0$ and consider the obstacles $\Omega_1 = \{-a/2\}$ and $\Omega_2 = \{a/2\}$.
These correspond to infinitesimally thin plates.
One can also take the limit $\epsilon \to 0$ of extended plates
$\Omega_1 = [-a/2 - \epsilon, -a/2 + \epsilon]$ and $\Omega_2 = [a/2 - \epsilon, a/2 + \epsilon]$.
The resulting $\Xi$ function does not depend on $\epsilon$,
so we can safely take the limit $\epsilon \to 0$.
The computation is a straightforward generalization of the one we show here; we leave the details to the reader.

The boundary layer operators for $\partial\Omega = \{-a/2, a/2\}$ are given by
\begin{align*}
S_{\lambda} = \frac{\rmi }{2\lambda}
\begin{pmatrix}
1 & e^{\rmi \lambda a}\\ e^{\rmi \lambda a} & 1
\end{pmatrix}, \quad
D_{\lambda} = D_{\lambda}' = -\frac{1}{2} \begin{pmatrix}
0 & e^{\rmi \lambda a}\\ e^{\rmi \lambda a} & 0
\end{pmatrix}, \quad
N_{\lambda} = -\frac{\rmi \lambda}{2}
\begin{pmatrix}
-1 & e^{\rmi \lambda a}\\
e^{\rmi \lambda a} & -1
\end{pmatrix}.
\end{align*}
It is straightforward to verify the relation $S_\lambda N_\lambda = D_\lambda^2 - \frac14$.

We obtain 
\begin{align*}
N_\lambda N_{\mathrm{diag},\lambda}^{-1}
= \begin{pmatrix}
1 & -e^{\rmi \lambda a}\\
-e^{\rmi \lambda a} & 1
\end{pmatrix}
\end{align*}
and
\begin{align*}
\Xi_N(\lambda) = \log \left(1- e^{2i\lambda a}\right)
\end{align*}
which is the same result as for the Dirichlet case computed in \cite{strohmaierDimensionalReductionFormulae2024}.

\subsection{Transmission boundary conditions.}
The transmission case introduces several complications compared to the Dirichlet and Neumann cases.
\begin{enumerate}
\item The 1-dimensional boundary operators are matrices of twice the size, since we incorporate Dirichlet and Neumann data simultaneously.
\item We must work with extended rather than pointlike obstacles to impose transmission boundary conditions.
\item The $\xi$-integral cannot be completely eliminated, because the shift in the spectral parameter is $\kappa_-^2 \xi^2$ in the interior domain and $\kappa_+^2\xi^2$ in the exterior domain, and therefore cannot be absorbed into the mass term. Thus the analogue of $(\ref{eq:1d-neumann-resolvent-rep})$ is false for transmission conditions.
\end{enumerate}

The operator $\widetilde{\Delta}_T$ acts by
\begin{align*}
\widetilde{\Delta}_T u = \begin{cases}
\kappa_+^2 \Delta u & \supp(u) \subset \Omega^+\\
\kappa_-^2 \Delta u & \supp(u) \subset \Omega^-.
\end{cases}
\end{align*}
After conjugating by the partial Fourier transform, this becomes
\begin{align*}
\Delta_T(\xi) u = \begin{cases} \kappa_+^2 (u'' - \xi^2 u) \text{ in } \Omega^+,\\
\kappa_-^2 (u'' - \xi^2 u) \text{ in } \Omega^-.
\end{cases}
\end{align*}

The relative operator
\begin{align}
R_{s, T}(\xi) = (-\Delta_T(\xi) + m^2)^{s} + (-\kappa_+^2\Delta_0 + m^2)^{s} - (-\Delta_{T,1}(\xi) + m^2)^{s} - (-\Delta_{T,2}(\xi) + m^2)^{s}
\end{align}
has the representation
\begin{align*}
R_{s, T}(\xi) = \frac{\mathrm{i}}{\pi} \int_{\tilde{\Gamma}} \lambda\left(\left(\lambda^2+m^2\right)^{s}-\left(m^2\right)^{s}\right) R_{\mathrm{rel},T}(\lambda, \xi) \mathrm{d} \lambda
\end{align*}
where
\begin{align*}
R_{\mathrm{rel},T}(\lambda, \xi) = (-\Delta_T(\xi) + \lambda^2)^{-1} + (-\kappa_+^2\Delta_0 + \lambda^2)^{-1} - (-\Delta_{T,1}(\xi) + \lambda^2)^{-1} - (-\Delta_{T,2}(\xi) + \lambda^2)^{-1}.
\end{align*}
We can deform the contour to obtain
\begin{align*}
\operatorname{tr}(R_{s, T}(\xi))
= \frac{2}{\pi} \sin\left(\pi s\right) \int_m^\infty \lambda(\lambda^2 - m^2)^{s} \operatorname{tr} R_{\mathrm{rel}, T}(i \lambda, \xi) \mathrm{d} \lambda.
\end{align*}
Therefore we have
\begin{align*}
\operatorname{tr}(\widetilde{R}_{s, T})
= \frac{1}{(2 \pi)^2} \int_{\mathbb{R}^2} \operatorname{tr}\left(R_{s,T}(\xi)\right) \mathrm{d} \xi
 = \frac{1}{2\pi^3} \sin\left(\pi s\right)\int_m^\infty \lambda(\lambda^2 - m^2)^{s} \operatorname{tr} R_{\mathrm{rel},T}(i\lambda,\xi) \mathrm{d} \xi \mathrm{d} \lambda.
\end{align*}
Using the relation $\operatorname{tr} R_{\mathrm{rel},T}(\lambda,\xi) = -\frac{1}{2\lambda}\frac{\partial}{\partial \lambda} \Xi_{T}(\lambda, \xi)$ we obtain
\begin{align*}
\operatorname{tr}(\widetilde{R}_{s, T})
= \frac{s}{2\pi^3} \sin\left(\pi s\right)\int_m^\infty \lambda(\lambda^2 - m^2)^{s-1} \Xi_T(\rmi \lambda, \xi)\mathrm{d} \xi \mathrm{d} \lambda
\end{align*}
The $\Xi$ function is given by
\begin{align*}
&\Xi_{T}(\lambda, \xi) = \log \det_{\mathcal{B}_{\sigma_+}^-}(P_{\sigma_+}^-H_\lambda(\xi)P^-_{\sigma_-} (P^-_{\sigma_+}H_{\mathrm{diag},\lambda}(\xi)P^-_{\sigma_-})^{-1})\\
&H_{\lambda}(\xi) = P_{\sigma_{+}}^{+} \M-\M P_{\sigma_{-}}^{-}=\M P_{\sigma_{-}}^{+}-P_{\sigma_{+}}^{-} \M=-A_{\sigma_{+}} \M-\M A_{\sigma_{-}}\\
&\sigma_{ \pm}^2:=\frac{\lambda^2}{\kappa_{ \pm}^2}-\xi^2.
\end{align*}
The Casimir energy per unit area is the case $s=1/2$ (and $m=0$):
\begin{align*}
E_{\mathrm{Cas}} = \frac{1}{2} \operatorname{tr} \widetilde{R}_{1/2, T}
= \frac{1}{8\pi^3} \int_0^\infty \Xi_T(\rmi \lambda, \xi) \mathrm{d} \xi \mathrm{d} \lambda.
\end{align*}
Note that 
\begin{align*}
    \sigma_{\pm}(\rmi t, \xi) = \sqrt{-(t/\kappa_{\pm})^2 - \xi^2} = \rmi  \sqrt{(t/\kappa_{\pm})^2 + \xi^2}
    =: \rmi  \eta_\pm(t,\xi),\qquad t > 0.
\end{align*}

As stated above we must work with extended obstacles $\Omega_1 = [-d/2 - a, - d/2]$, $\Omega_2 = [d/2, d/2 + a]$; thus $d$ represents the distance between the plates and $a$ the thickness of the plates.
We therefore have 4 boundary components; the combination of Dirichlet and Neumann data gives us $8 \times 8$ matrices as boundary operators.

Notebooks for WolframLanguage and Jupyter verifying our computations are available at \href{https://github.com/qft-hofmann/rel-trace-notebook}{Github}.
We define the reflection coefficient
\begin{align*}
R(\lambda, \xi) = \frac{\kappa_-^2 \sigma_- - \kappa_+^2 \sigma_+}{\kappa_-^2 \sigma_- + \kappa_+^2 \sigma_+}.
\end{align*}
Then we find for $\Xi_T$:
\begin{align}
    \label{eq:finite-slabs-trans-xi}
\Xi_T(\rmi t, \xi) = \log \left(\frac{\left(1-R^2 e^{-2 a \eta_{-}}\right)^2-R^2 e^{-2 d \eta_{+}}\left(1-e^{-2 a \eta_{-}}\right)^2}{\left(1-R^2 e^{-2a \eta_{-}}\right)^2}\right)
\end{align}
To the best of our knowledge, this formula has not appeared in the literature before.

In the limit $a\to\infty$ we have $e^{-2a\eta_{\pm}} \to 0$ and $\Xi_T(\lambda, \xi)$ reduces to
\begin{align*}
\Xi_T(\rmi t,\xi) = \log \left(1 - R^2 e^{-2d\eta_+}\right).
\end{align*}
This result agrees with the computation of Milton et al. in \cite{miltonCasimirEnergyDispersion2010}.
In that article, the authors consider the electromagnetic field and therefore obtain a sum of such contributions from the TE mode and the TM mode, which can be treated as independent scalar fields with transmission boundary conditions.

Note that in the limit $a \to 0$ we obtain $\Xi_T = 2 \log(1+ R) - 2\log(1-R)$.
The Casimir energy is therefore independent of the distance $d$ between the plates, and hence there is no Casimir force.
Thus the effect of transmission boundary conditions vanishes in the limit
of infinitesimally thin plates, in contrast to Dirichlet and Neumann boundary conditions.

\appendix

\section{Kernel estimates.}

In this Appendix we collect estimates on the free Green function $G_\lambda$ which are needed in the main text.
We need to estimate both high frequency $|\lambda| \to \infty$ and low frequency $|\lambda| \to 0$ behaviour, and we need to track growth in different norms.
But the main distinction is between the near-field regime with $(x,y) \in B \times B$ ($B$ bounded), and the separated regime, where $(x,y) \in U \times V$ with $\operatorname{dist}(U,V) > 0$.
An important special case of the separated regime is the far-field regime, in which one of $U, V$ is unbounded.
The most convenient way to treat the different cases of the separated regime is by estimating $G_\lambda(x_0, \cdot)$ for a fixed $x_0$.

\subsection{Pointwise estimate.}
The following pointwise statement is the basis of all subsequent estimates.
It is a slight reformulation of Lemma A.2 in \cite{HSW}.
\begin{lemma}
    \label{lem:green-function-deriv-bounds}
    Let $G_\lambda$ be the kernel of $(-\Delta - \lambda^2)^{-1}$ for $\operatorname{Im}(\lambda) > 0$.

    By symmetry we can write $G_\lambda(x,y) = g_\lambda(x-y)$, where $g_\lambda(\xi)$ actually depends only on $|\xi|$.
    
    Fix $\epsilon>0$.
    Then for every multi-index $\alpha\in \mathbb{N}_0^d$ there exist constants $C_{\alpha,\epsilon}>0$ and $c_\epsilon>0$ such that the following estimates hold for all $\xi\in \mathbb{R}^d\setminus{0}$ and all $\lambda \in \mathfrak{D}_\epsilon$:
    \begin{enumerate}
    \item If $|\lambda||\xi|\le 1$, then
    \begin{align}
    |D^\alpha g_\lambda(\xi)|
    \le
    C_{\alpha,\epsilon}
    \begin{cases}
    |\xi|^{-(d-2+|\alpha|)}, & d\ge 3 \text{ or } |\alpha|>0,\\
    1+\bigl|\log\bigl(|\lambda| |\xi|\bigr)\bigr|, & d=2,\ \alpha=0.
    \end{cases}
\end{align}
    
    \item If $|\lambda||\xi|\ge 1$, then
    \begin{align}
    |D^\alpha g_\lambda(\xi)|
    \le
    C_{\alpha,\epsilon} |\lambda|^{d-2+|\alpha|}e^{-\operatorname{Im}\lambda |\xi|}.
    \end{align}
    
    \item Consequently, if $|\xi|>1$, then
    \begin{align}
    |D^\alpha g_\lambda(\xi)|
    \le
    C_{\alpha,\epsilon}e^{-c_\epsilon \operatorname{Im}\lambda |\xi|}
    \begin{cases}
    |\xi|^{-(d-2+|\alpha|)}, & d\ge 3 \text{ or } |\alpha|>0,\\
    1+\bigl|\log(|\lambda|)\bigr|, & d=2,\ \alpha=0.
    \end{cases}
\end{align}
    \end{enumerate}
    \end{lemma}


\begin{corollary}\label{lem:far-field-pt-bounds}
        Assume $d\ge 3$. Let $K\subset B_R(0)$ be compact, let $\epsilon\in(0,\pi/2)$, and let
        $\alpha,\beta$ be multiindices. Then there exist constants
        $C_{\alpha,\beta,\epsilon}>0$ and $c_\epsilon>0$ such that
        \[
        |D_x^\alpha D_z^\beta G_\lambda(x,z)|
        \le C_{\alpha,\beta,\epsilon}\,
        |x|^{2-d-|\alpha|-|\beta|} e^{-c_\epsilon (\Im\lambda)|x|}
        \]
        for all $\lambda\in\mathfrak D_\epsilon$, all $z\in K$, and all $|x|>2R$.
\end{corollary}

\subsection{Separated regime estimates.}

The following are useful consequences of the pointwise estimates \ref{lem:green-function-deriv-bounds}.
These are not the most precise possible estimates; we have absorbed some polynomial or log growth into the exponential term.
This declutters the estimates and is sufficient for our purposes.

\begin{corollary}
    \label{cor:kernel-pointwise-CN}
For any compact $K\subset \mathbb{R}^d$, $N\in \mathbb{N}$, $\epsilon > 0$ there exist constants $C_1, C_2 > 0$ such that for all $\lambda \in \mathfrak{D}_\epsilon$ and all $x_0 \in \mathbb{R}^d$ with $\delta := \operatorname{dist}(x_0, K) > 0$, we have
\begin{align}
\norm{D^\alpha_x G_\lambda(x_0, \cdot)}_{C^N(K)} \leq C_1 \delta^{2-d} e^{-C_2 \delta \operatorname{Im}\lambda}.
\end{align}
\end{corollary}

\begin{corollary}
    \label{cor:kernel-pointwise-sobolev}
    Let $U \subset \mathbb{R}^d$ be an open domain.
    For all $\epsilon > 0$, $\delta' > 0$, $m\in \mathbb{R}$ there exists $C > 0$ such that for all $\lambda \in \mathfrak{D}_\epsilon$ and all $x_0 \in \mathbb{R}^d$ with $\delta := \operatorname{dist}(x_0, U) > \delta'$, we have
    \begin{align}
    \norm{D^\alpha_x G_\lambda(x_0, \cdot)}_{H^m(U)} \leq C \rho(\operatorname{Im} \lambda)^{1/2} e^{-\delta' \operatorname{Im} \lambda}
    \end{align}
\end{corollary}

\subsection{Near field estimate.}
We turn now to the only near field estimate which we need.

The following is a step in the proof of Proposition 7.1 in \cite{strohmaierRelativeTraceFormula2021},
where however it is given only for $d=3$.
\begin{lemma}
    There is a constant $C_{\epsilon,\chi}$ such that, for all $\lambda \in \mathfrak{D}_\epsilon$,
    \begin{align}
    \norm{\chi G_\lambda \chi u}_{H^{1}(\mathbb{R}^d)} &\leq C_{\epsilon,\chi} \norm{u}_{H^{-1}(\mathbb{R}^d)}\qquad d\geq 3\\
    \norm{\chi G_\lambda \chi u}_{H^{1}(\mathbb{R}^2)} &\leq C_{\epsilon,\chi}(1 + |\log|\lambda||) \norm{u}_{H^{-1}(\mathbb{R}^2)} \qquad d=2.
    \end{align}
\end{lemma}
    \begin{proof}
The proof is essentially the one in \cite{strohmaierRelativeTraceFormula2021}.
We repeat it here to leave the reader in no doubt that it goes through for $d\geq 3$.

    As a first step we replace the cutoff functions by a cutoff of the form $\eta(|x-y|)$.
    Let $\eta \in C_0^{\infty}(\mathbb{R})$ be a function that is one near $\left[-R_1, R_1\right]$, where $R_1$ is the diameter
    of the support of $\chi$.
    This choice implies that $\chi\left(-\Delta-\lambda^2\right)^{-1} \chi=\chi R_{\eta, \lambda} \chi$, where $R_{\eta, \lambda}$ is the convolution operator $R_{\eta,\lambda} f = k_\lambda \ast f$ with
    \begin{align*}
    k_\lambda(\xi) = \eta(|\xi|) g_\lambda(\xi).
    \end{align*}

    Because $(-\Delta + 1)g_\lambda(\xi) = \delta_0(\xi) + (1+\lambda^2) g_\lambda(\xi)$, the kernel $k_\lambda$ satisfies $(1-\Delta) k_\lambda \ast f = f + h_\lambda \ast f$ where
    \begin{align*}
    h_\lambda:=\left(1+\lambda^2\right) \eta g_\lambda-2 \nabla \eta \cdot \nabla g_\lambda-(\Delta \eta) g_\lambda.
    \end{align*}
    We will show that $\sup _{\lambda \in \mathfrak{D}_\epsilon}\left\|h_\lambda\right\|_{L^1\left(\mathbb{R}^d\right)}<\infty$ and therefore
    \begin{align*}
    \norm{R_{\eta, \lambda} f}_{H^{s+2}(\mathbb{R}^d)} \leq C \norm{(1-\Delta) R_{\eta, \lambda} f}_{H^s(\mathbb{R}^d)} \leq C \norm{f + h_\lambda \ast f}_{H^s(\mathbb{R}^d)} \leq C_s \norm{f}_{H^s(\mathbb{R}^d)}.
    \end{align*}
    The last estimate is by Young's convolution inequality and the fact that convolution operators commute with $1-\Delta$.
    
    Since $\eta \equiv 1$ in a neighbourhood of zero, and $\eta$ is compactly supported, the set
    \begin{align*}
        A:=\operatorname{supp} \nabla \eta \cup \operatorname{supp} \Delta \eta
    \end{align*}
    is contained in an annulus
    \begin{align*}
    A \subseteq \left\{\xi \in \mathbb{R}^d: a \leq|\xi| \leq b\right\}
    \end{align*}
    with $0 < a < b < \infty$.
    By Lemma \ref{lem:green-function-deriv-bounds} with $m=0,1$, both $g_\lambda$ and $\nabla g_\lambda$ are uniformly bounded on $A$:
    \begin{align*}
    \sup_{\lambda \in \mathfrak{D}_\epsilon} \sup_{\xi \in A} (|(\Delta \eta)(\xi) g_\lambda(\xi)| + 2 |(\nabla \eta)(\xi) \cdot (\nabla g_\lambda)(\xi)|) < \infty,
    \end{align*}
    and since $A$ has finite measure,
    \begin{align*}
        \sup _{\lambda \in \mathfrak{D}_\epsilon}\left(\left\|(\Delta \eta) g_\lambda\right\|_{L^1}+\left\|\nabla \eta \cdot \nabla g_\lambda\right\|_{L^1}\right)<\infty.
    \end{align*}
    
    It remains to estimate $\norm{(1+\lambda^2)\eta g_\lambda}_{L^1}$.
    Choose $R_2> 0$ such that $\operatorname{supp}\eta \subset B_{R_2}(0)$.
    
    Split the ball into the two regions
    
    $$
    E_{<}:=\{|z| \leq R_2,|\lambda||z| \leq 1\}, \quad E_{>}:=\{|z| \leq R_2,|\lambda||z| \geq 1\} .
    $$

    On $E_{<}$ we have $\left|g_\lambda(z)\right| \leq C|z|^{2-d}$,
    hence, in polar coordinates,
\begin{align*}
    \int_{E_{<}}\left(1+|\lambda|^2\right)\left|g_\lambda(z)\right| \mathrm{d} z \leq C\left(1+|\lambda|^2\right) \int_0^{\min \left(R_2,|\lambda|^{-1}\right)} r^{d-1} r^{2-d} \mathrm{d} r  \leq C
\end{align*}

    On $E_{>}$ we have $\left|g_\lambda(z)\right| \leq C|\lambda|^{d-2} e^{-c_\epsilon|\lambda||z|}$.
    Therefore
\begin{align*}
    \int_{E_{>}}\left(1+|\lambda|^2\right)\left|g_\lambda(z)\right| d z \leq C\left(1+|\lambda|^2\right)|\lambda|^{d-2} \int_{|\lambda|^{-1}}^{R_2} e^{-c_\epsilon|\lambda| r} r^{d-1} \mathrm{d}r \\
    =C\left(1+|\lambda|^2\right)|\lambda|^{-2} \int_1^{R_2|\lambda|} e^{-c_\epsilon s} s^{d-1} \mathrm{d}s \leq C
\end{align*}
    
    Hence we have
    \begin{align*}
        \sup _{\lambda \in \mathfrak{D}_\epsilon}\left\|\left(1+\lambda^2\right) \eta g_\lambda\right\|_{L^1}<\infty
    \end{align*}
    which yields the claim.

    For $d=2$ the proof is the same, except that the integrals over $E_<$ and $E_>$ reproduce the log term of the $g_\lambda$ estimate.
    \end{proof}

\section{Layer potential bounds}
\label{sec:layer-potential-bounds}
In addition to the estimates of the boundary operators established in Section \ref{sec:boundary-op-poly-growth}, we need to estimate the growth in $\lambda$ of the boundary layer potentials $S_\lambda$ and $D_\lambda$.
We will use the following spectral weights to express the estimates.

For $t>0$ define
\begin{align}
\label{eq:rho-d-def}
\rho_d(t):=
\begin{cases}
1+t^{-2}, & d=2,\\
1+t^{-1}, & d=3,\\
1+\log(1+t^{-1}), & d=4,\\
1, & d\geq 5,
\end{cases}
\end{align}
and
\begin{align}
\label{eq:ell-d-def}
\ell_d(t):=
\begin{cases}
1+\log(1+t^{-1}), & d=2,\\
1, & d\geq 3.
\end{cases}
\end{align}
Thus $\ell_d(t)^{1/2}\leq C\rho_d(t)^{1/2}$ and
$\ell_d(t)\leq C\rho_d(t)^{1/2}$ for all $t>0$.

\begin{lemma}
\label{lem:bd-potential-bounds}
For $\lambda\in\mathfrak D_\epsilon$ one has
\begin{align}
\label{eq:S-global-bound}
\norm{\widetilde S_\lambda u}_{H^1(\mathbb R^d)}
&\leq C_\epsilon \rho_d(|\lambda|)^{1/2}
\norm{u}_{H^{-1/2}(\partial\Omega)},\\
\label{eq:D-global-bound}
\norm{\widetilde D_\lambda v}_{H^1(\Omega^-)\oplus H^1(\Omega^+)}
&\leq C_\epsilon\bigl(1+|\lambda|^2+\ell_d(|\lambda|)^{1/2}\bigr)
\norm{v}_{H^{1/2}(\partial\Omega)},\\
\label{eq:K-global-bound}
\norm{K_\lambda\Phi}_{H^1(\Omega^-)\oplus H^1(\Omega^+)}
&\leq C_\epsilon\bigl(|\lambda|^2+\rho_d(|\lambda|)^{1/2}\bigr)
\norm{\Phi}_{\mathscr C}.
\end{align}
In particular, for $d\geq 3$ the double layer estimate reduces to
\[
\norm{\widetilde D_\lambda v}_{H^1(\Omega^-)\oplus H^1(\Omega^+)}
\leq C_\epsilon(1+|\lambda|^2)\norm{v}_{H^{1/2}(\partial\Omega)}.
\]
\end{lemma}

The estimates in this Lemma neatly split into a near field (Lemma \ref{lem:near-field-layer-estimate}) and a far field estimate (Lemma \ref{lem:far-field-layer-estimate}).
In the case of the single layer potential, the far field contains a monopole contribution which grows like $\rho_d(|\lambda|)^{1/2}$ (Lemma \ref{lem:monopole}).

To split off the monopole term, we introduce the total charge
\[
q(u):=\int_{\partial\Omega}u\,d\sigma, \qquad u \in H^{-1/2}(\partial\Omega).
\]
Let
\[
\psi_0:=|\partial\Omega|^{-1}\mathbf 1_{\partial\Omega},
\qquad q(\psi_0)=1,
\]
and set
\[
Pu:=u-q(u)\psi_0 .
\]
Since $q(Pu)=0$, the operator $P$ removes the monopole contribution of $u$.

\begin{lemma}[Far field estimate]
\label{lem:far-field-layer-estimate}
Let $\chi\in C_c^\infty(\mathbb R^d)$ satisfy $\chi\equiv 1$ on a
neighbourhood of $\partial\Omega$.  Then, for all
$\lambda\in\mathfrak D_\epsilon$,
\begin{align}
\label{eq:far-S-zero-charge}
\norm{(1-\chi)\widetilde S_\lambda Pu}_{H^1(\mathbb R^d)}
&\leq C_{\epsilon,\chi}\ell_d(|\lambda|)^{1/2}
\norm{Pu}_{H^{-1/2}(\partial\Omega)},\\
\label{eq:far-D}
\norm{(1-\chi)\widetilde D_\lambda v}_{H^1(\Omega^-)\oplus H^1(\Omega^+)}
&\leq C_{\epsilon,\chi}\ell_d(|\lambda|)^{1/2}
\norm{v}_{H^{1/2}(\partial\Omega)}.
\end{align}
For $d\geq 3$ the factor $\ell_d(|\lambda|)^{1/2}$ in
\eqref{eq:far-S-zero-charge} and \eqref{eq:far-D} can be replaced by
$e^{-c_\epsilon\delta'\operatorname{Im}\lambda}$ for every
$0<\delta'<\operatorname{dist}(\operatorname{supp}(1-\chi),\partial\Omega)$.
\end{lemma}

\begin{proof}
Choose a bounded open set $U$ with
\[
\partial\Omega\Subset U\Subset\{\chi=1\},
\]
and choose $R>0$ so that $U\subset B_R(0)$.  For
$x\in\mathbb R^d\setminus U$ and $y\in U$ define the zero-charge kernel
\[
Z_\lambda(x,y):=
G_\lambda(x,y)-\int_{\partial\Omega}G_\lambda(x,z)\psi_0(z)\,d\sigma(z).
\]
Since $q(Pu)=0$,
\[
\widetilde S_\lambda Pu(x)
=
\bigl\langle \gamma_D^*Pu,Z_\lambda(x,\cdot)\bigr\rangle_{H^{-1}(U),H^1(U)} .
\]

We first estimate $Z_\lambda$.  If
$x\in \operatorname{supp}(1-\chi)\cap B_{2R}(0)$, then $x$ has positive
distance from $U$.  The pointwise estimates of Appendix~A, Lemma
\ref{lem:green-function-deriv-bounds}, applied to derivatives in the $y$ and
mixed $x,y$ variables, imply
\[
\norm{Z_\lambda(x,\cdot)}_{H^1(U)}
+\norm{\nabla_x Z_\lambda(x,\cdot)}_{H^1(U)}
\leq C_{\epsilon,\chi}.
\]
If $d\geq3$, the same argument gives the sharper factor
$e^{-c_\epsilon\delta'\operatorname{Im}\lambda}$.

Now assume $|x|>2R$.  For $y,z\in U$, the mean-value theorem gives
\[
G_\lambda(x,y)-G_\lambda(x,z)
=
\int_0^1 (y-z)\cdot
\nabla_2G_\lambda\bigl(x,z+t(y-z)\bigr)\,dt .
\]
The derivative estimates in Appendix~A, Lemma
\ref{lem:green-function-deriv-bounds}, give, after possibly decreasing
$c_\epsilon>0$,
\begin{align}
\label{eq:Z-tail-1}
\norm{Z_\lambda(x,\cdot)}_{H^1(U)}
&\leq C_{\epsilon,\chi}|x|^{1-d}
e^{-c_\epsilon(\operatorname{Im}\lambda)|x|},\\
\label{eq:Z-tail-2}
\norm{\nabla_xZ_\lambda(x,\cdot)}_{H^1(U)}
&\leq C_{\epsilon,\chi}|x|^{-d}
e^{-c_\epsilon(\operatorname{Im}\lambda)|x|}.
\end{align}
Indeed, when $|\lambda||x|>1$ the polynomial powers of $|\lambda|$ in the
large-argument estimate are absorbed into the exponential by
$s^N e^{-cs}\leq C_{N,c}e^{-cs/2}$.

The estimates above imply
\[
|\widetilde S_\lambda Pu(x)|
\leq C\norm{Pu}_{H^{-1/2}(\partial\Omega)}
\norm{Z_\lambda(x,\cdot)}_{H^1(U)}
\]
and the same bound with $\nabla_xZ_\lambda$ for
$|\nabla_x\widetilde S_\lambda Pu(x)|$.  Hence the tail contribution is
controlled by
\[
\int_{2R}^\infty r^{2(1-d)}e^{-2c_\epsilon(\operatorname{Im}\lambda)r}
r^{d-1}\,dr
=
\int_{2R}^\infty r^{1-d}e^{-2c_\epsilon(\operatorname{Im}\lambda)r}\,dr .
\]
This integral is uniformly bounded for $d\geq3$.  For $d=2$ it is bounded by
\[
C\bigl(1+\log(1+|\lambda|^{-1})\bigr),
\]
because $C_\epsilon^{-1} |\lambda|\operatorname{Im}\lambda \leq C_\epsilon |\lambda|$ on
$\mathfrak D_\epsilon$ for some $C_\epsilon > 0$.  The gradient tail is uniformly bounded in every
dimension $d\geq2$.  This proves \eqref{eq:far-S-zero-charge}.

The proof of \eqref{eq:far-D} is analogous.  Since
\[
\widetilde D_\lambda v(x)
=
\int_{\partial\Omega}\partial_{\nu_y}G_\lambda(x,y)v(y)\,d\sigma(y),
\]
only derivatives in the boundary variable occur.  The estimates just used for
$\nabla_yG_\lambda$ and $\nabla_x\nabla_yG_\lambda$, together with the
continuous embedding
$H^{1/2}(\partial\Omega)\hookrightarrow L^2(\partial\Omega)$, give
\[
|\widetilde D_\lambda v(x)|
\leq C\norm{v}_{H^{1/2}(\partial\Omega)}
|x|^{1-d}e^{-c_\epsilon(\operatorname{Im}\lambda)|x|}
\]
and
\[
|\nabla_x\widetilde D_\lambda v(x)|
\leq C\norm{v}_{H^{1/2}(\partial\Omega)}
|x|^{-d}e^{-c_\epsilon(\operatorname{Im}\lambda)|x|}.
\]
Integration gives the same dimensional factor as above.  This proves
\eqref{eq:far-D}.
\end{proof}

\begin{lemma}[Near field estimate]
\label{lem:near-field-layer-estimate}
Let $\chi\in C_c^\infty(\mathbb R^d)$.  Then
\begin{align}
\label{eq:near-S}
\norm{\chi\widetilde S_\lambda u}_{H^1(\mathbb R^d)}
&\leq C_{\epsilon,\chi}\ell_d(|\lambda|)
\norm{u}_{H^{-1/2}(\partial\Omega)},\\
\label{eq:near-D}
\norm{\chi\widetilde D_\lambda v}_{H^1(\Omega^-)\oplus H^1(\Omega^+)}
&\leq C_{\epsilon,\chi}(1+|\lambda|^2)
\norm{v}_{H^{1/2}(\partial\Omega)}.
\end{align}
\end{lemma}

\begin{proof}
Choose $\chi_1\in C_c^\infty(\mathbb R^d)$ with $\chi_1\equiv1$ on a
neighbourhood of $\partial\Omega$ and on a neighbourhood of
$\operatorname{supp}\chi$.  Since $\gamma_D^*u$ is supported on
$\partial\Omega$,
\[
\chi\widetilde S_\lambda u
=
\chi G_\lambda\chi_1\gamma_D^*u .
\]
The localized resolvent estimate obtained from Appendix~A, Lemma
\ref{lem:green-function-deriv-bounds}, gives
\[
\norm{\chi G_\lambda\chi_1 f}_{H^1(\mathbb R^d)}
\leq C_{\epsilon,\chi}\ell_d(|\lambda|)
\norm{f}_{H^{-1}(\mathbb R^d)} .
\]
Here the factor $\ell_d$ is needed only for the two-dimensional constant
logarithmic term of the Green function.  Since
\[
\gamma_D^*:H^{-1/2}(\partial\Omega)\longrightarrow H^{-1}(\mathbb R^d)
\]
is continuous, \eqref{eq:near-S} follows.

The double layer potential is more involved.
We cannot simply replace $\gammaDir{*}$ by $\gammaNeu{*}$, because the latter maps into the dual of $H^1_\Delta$, rather than into $H^{-1}_\comp$.
Instead we follow the strategy of Theorem 6.11 of McLean \cite{mcleanStronglyEllipticSystems2000}.

On $\mathfrak{D}_\epsilon$, $S_\lambda$ is invertible and hence we have the Dirichlet solution operator $\mathcal{U} = \widetilde{S}_\lambda S_{\lambda}^{-1}$.
Now for $\psi \in H^{1/2}(\partial\Omega)$, put
\begin{align*}
\mathcal{U}\psi = \begin{cases}
\widetilde{S}_\lambda S_\lambda^{-1} \psi & \text{on $\Omega^-$}\\
0 & \text{on $\Omega^+$}.
\end{cases}
\end{align*}
Then $(-\Delta - \lambda^2) \mathcal{U}\psi = 0$ on $\Omega^\pm$.
The third Green identity therefore gives
\begin{align*}
\mathcal{U}\psi = \widetilde{D}_\lambda [\gamma_0 \mathcal{U}\psi] - \widetilde{S}_\lambda [\gamma_1 \mathcal{U}\psi]
\end{align*}
and since $\gamma_0^+ \mathcal{U}\psi = \gamma_1^+ \mathcal{U}\psi = 0$, this is
\begin{align*}
\mathcal{U}\psi = - \widetilde{D}_\lambda \psi + \widetilde{S}_\lambda Q_\lambda^{-} \psi
\end{align*}
where $Q_\lambda^-$ is the interior Dirichlet-to-Neumann operator.
Proposition \ref{prop:boundary-op-bounds} gives the sector bounds
\begin{align*}
\norm{Q_\lambda^{-}\psi }_{H^{-1/2}} \leq C_\epsilon(1+|\lambda|^2) \norm{\psi}_{H^{1/2}},
\qquad
\norm{S_\lambda^{-1}\psi}_{H^{-1/2}} \leq C_\epsilon' (1+ |\lambda|^2) \norm{\psi}_{H^{1/2}}.
\end{align*}

Hence
\begin{align*}
&\norm{\mathcal{U}\psi}_{H^1(\Omega^-)} = \norm{\widetilde{S}_\lambda S_\lambda^{-1} \psi}_{H^1(\Omega^-)}
\leq \norm{\chi \widetilde{S}_\lambda S_\lambda^{-1} \psi}_{H^1(\mathbb{R}^d)}
\lesssim \norm{S_\lambda^{-1}\psi}_{H^{-1/2}}
\lesssim (1+ |\lambda|^2) \norm{\psi}_{H^{1/2}},\\
&\norm{\chi\widetilde{S}_\lambda Q_\lambda^{-}\psi}_{H^1(\mathbb{R}^d)}
\lesssim \norm{Q_\lambda^{-}\psi}_{H^{-1/2}}
\lesssim (1 + |\lambda|^2) \norm{\psi}_{H^{1/2}}.
\end{align*}
Therefore
\begin{align*}
\norm{\chi\widetilde{D}_\lambda \psi}_{H^1(\Omega^+)\oplus H^1(\Omega^-)}
\leq C_{\epsilon,\chi} (1+ |\lambda|^2) \norm{\psi}_{H^{1/2}(\partial\Omega)}.
\end{align*}
\end{proof}

\begin{lemma}[Monopole contribution]
\label{lem:monopole}
For $\lambda\in\mathfrak D_\epsilon$,
\begin{align}
\label{eq:monopole-bound}
\norm{(1-\chi)\widetilde S_\lambda\psi_0}_{H^1(\mathbb R^d)}
\leq C_{\epsilon,\chi}\rho_d(|\lambda|)^{1/2}.
\end{align}
\end{lemma}

\begin{proof}
On the compact set
$\operatorname{supp}(1-\chi)\cap B_{2R}(0)$ the kernel is separated from the
boundary.  Appendix~A, Lemma \ref{lem:green-function-deriv-bounds}, gives a
uniform bound in dimensions $d\geq3$ and a bound
$C(1+\log(1+|\lambda|^{-1}))$ in dimension $d=2$.  Both are bounded by
$C\rho_d(|\lambda|)^{1/2}$.

It remains to estimate the exterior tail $|x|>2R$.  For $d\geq3$, Appendix~A,
Corollary \ref{lem:far-field-pt-bounds}, gives
\[
|\widetilde S_\lambda\psi_0(x)|
\leq C_{\epsilon,\chi}|x|^{2-d}
e^{-c_\epsilon(\operatorname{Im}\lambda)|x|},
\qquad
|\nabla_x\widetilde S_\lambda\psi_0(x)|
\leq C_{\epsilon,\chi}|x|^{1-d}
e^{-c_\epsilon(\operatorname{Im}\lambda)|x|}.
\]
The gradient term is square-integrable uniformly for all $d\geq3$.  The
$L^2$ term gives
\[
\int_{2R}^{\infty}
r^{2(2-d)}e^{-2c_\epsilon(\operatorname{Im}\lambda)r}r^{d-1}\,dr
=
\int_{2R}^{\infty}
r^{3-d}e^{-2c_\epsilon(\operatorname{Im}\lambda)r}\,dr .
\]
This is bounded by $C|\lambda|^{-1}$ for $d=3$, by
$C(1+\log(1+|\lambda|^{-1}))$ for $d=4$, and by $C$ for $d\geq5$.

For $d=2$, Lemma \ref{lem:green-function-deriv-bounds} gives, after enlarging
$R$ if necessary,
\[
|\widetilde S_\lambda\psi_0(x)|
\leq C
\begin{cases}
1+|\log(|\lambda||x|)|, & |\lambda||x|\leq 1,\\
e^{-c_\epsilon(\operatorname{Im}\lambda)|x|}, & |\lambda||x|\geq 1,
\end{cases}
\]
and
\[
|\nabla_x\widetilde S_\lambda\psi_0(x)|
\leq C|x|^{-1}e^{-c_\epsilon(\operatorname{Im}\lambda)|x|}.
\]
For the $L^2$ term, the change of variables $s=|\lambda|r$ gives
\[
\int_{2R}^{\infty}|\widetilde S_\lambda\psi_0(r)|^2 r\,dr
\leq C|\lambda|^{-2}
\left[
\int_0^1(1+|\log s|)^2s\,ds+
\int_1^\infty e^{-c_\epsilon s}s\,ds
\right]
\leq C|\lambda|^{-2}.
\]
The gradient term is smaller, up to a logarithmic factor, and is therefore
also bounded by $C|\lambda|^{-2}$ after squaring.  Taking square roots proves
\eqref{eq:monopole-bound} in dimension two.  Combining the preceding estimates
proves the lemma.
\end{proof}

\begin{proof}[Proof of Lemma \ref{lem:bd-potential-bounds}]
Let $u\in H^{-1/2}(\partial\Omega)$.  Decompose
\[
u=Pu+q(u)\psi_0.
\]
Then
\[
\widetilde S_\lambda u
=
\chi\widetilde S_\lambda u
+
(1-\chi)\widetilde S_\lambda Pu
+
q(u)(1-\chi)\widetilde S_\lambda\psi_0 .
\]
Using Lemma \ref{lem:near-field-layer-estimate}, Lemma
\ref{lem:far-field-layer-estimate}, and Lemma \ref{lem:monopole}, together
with
\[
\ell_d(|\lambda|)\leq C\rho_d(|\lambda|)^{1/2},
\qquad
\ell_d(|\lambda|)^{1/2}\leq C\rho_d(|\lambda|)^{1/2},
\]
gives \eqref{eq:S-global-bound}.

For the double layer potential,
\[
\widetilde D_\lambda v
=
\chi\widetilde D_\lambda v
+
(1-\chi)\widetilde D_\lambda v .
\]
The near-field estimate \eqref{eq:near-D} and the far-field estimate
\eqref{eq:far-D} give \eqref{eq:D-global-bound}.

Finally, if $\Phi=(\phi,\psi)$, then
\[
K_\lambda\Phi=\widetilde D_\lambda\phi-\widetilde S_\lambda\psi .
\]
Combining \eqref{eq:S-global-bound} and \eqref{eq:D-global-bound}, and using
$\ell_d^{1/2}\leq C\rho_d^{1/2}$ and $\rho_d\geq 1$, gives
\[
\norm{K_\lambda\Phi}_{H^1(\Omega^-)\oplus H^1(\Omega^+)}
\leq
C_\epsilon\bigl(|\lambda|^2+\rho_d(|\lambda|)^{1/2}\bigr)
\norm{\Phi}_{\mathscr C}.
\]
This proves \eqref{eq:K-global-bound}.
\end{proof}

\section{Trace-class and nuclear operators.}
If $H_1, H_2$ are Hilbert spaces, we write $\Nuc(H_1, H_2)$ for the space of nuclear operators $H_1 \to H_2$.
The space $\Nuc(H, H)$ is usually called the \emph{trace class}; sometimes that term is also applied to $\Nuc(H_1, H_2)$ for $H_1 \neq H_2$.
If $T: H_1 \to H_2$ is a bounded linear operator, $T^*T : H_1 \to H_1$ is a positive operator, and therefore we can define $|T| = (T^*T)^{1/2}$,
called the \emph{absolute value} of $T$.
If $T \in \Nuc(H_1, H_2)$ then $A T \in \Nuc(H_1, H_3)$ and $T B \in \Nuc(H_0, H_2)$ for any bounded operators $A: H_2 \to H_3$ and $B: H_0 \to H_1$.

\begin{proposition}[Theorem 48.2, \cite{trevesTopologicalVectorSpaces1967}]
Let $T: H_1 \to H_2$ be a bounded linear operator between Hilbert spaces.
The following are equivalent:
\begin{enumerate}
\item $T \in \Nuc(H_1, H_2)$.
\item $|T| \in \Nuc(H_1, H_1)$
\item $|T|$ is compact and $\Tr |T| < \infty$.
\end{enumerate}
If any of these equivalent conditions is satisfied, $\norm{T}_{\Nuc(H_1, H_2)} = \Tr |T|$.
\end{proposition}

We give the following useful characterizations of $\Nuc(H_1, H_2)$.
\begin{proposition}
A bounded linear map $T: H_1 \to H_2$ is nuclear if and only if either of the following conditions is satisfied.
\begin{enumerate}
\item There exist orthonormal sequences $(e_n)_{n\in \mathbb{N}}$ in $H_1$ and $(f_n)_{n\in \mathbb{N}}$ in $H_2$ and a sequence $(\lambda_n)_{n\in \mathbb{N}} \in \ell^1(\mathbb{N})$ such that $T x = \sum_{n=1}^\infty \lambda_n \langle x, e_n\rangle f_n$ for all $x \in H_1$.
In this case $\norm{T}_{\Nuc} \leq \norm{\lambda}_{\ell^1}$.
\item $T = RS$ for some Hilbert-Schmidt operators $R: H \to H_2$ and $S: H_1 \to H$ for some Hilbert space $H$.
In this case $\norm{T}_{\Nuc} \leq \norm{R}_{\HS} \norm{S}_{\HS}$.
\end{enumerate}
\end{proposition}

On the boundary of a Lipschitz domain $\Omega$, the notion of a ``smooth'' kernel loses some of its meaning, since the smoothness of the kernel is limited by the regularity of the boundary.
However, the restrictions to $\partial\Omega$ of smooth kernels still possess useful properties, such as being trace class.
\begin{lemma}
    \label{lem:trace-class-smoothing-ops}
Let $\Omega_1, \ldots, \Omega_k\subset \mathbb{R}^d$ be bounded Lipschitz domains with $\overline{\Omega}_i \cap \overline{\Omega}_j = \emptyset$ for $i \neq j$.
Choose bounded open neighbourhoods $\overline{\Omega}_i \subset U_i \subset \mathbb{R}^d$ with $\operatorname{dist}(U_i, U_j) > 0$ for $i \neq j$.
Write $\Omega = \bigcup_{i=1}^k \Omega_i$ and $U = \bigcup_{i=1}^k U_i$.
Let $A$ be a bounded operator on $L^2(\mathbb{R}^d)$ with kernel $K \in \mathscr{D}'(\mathbb{R}^d \times \mathbb{R}^d)$.
Assume that $K$ is smooth away from the diagonal $\{(x,x)\mid x \in \mathbb{R}^d\}$.

For every $i\neq j$ the following off-diagonal block operators are nuclear
\begin{align}
p_j A p_i: H^s(\Omega) &\to H^t(\Omega), \qquad s > -\frac{1}{2}, t \in \mathbb{R} \label{eq:off-block-nuclear-list-1}\\
q_j \gammaDir{} A \gammaDir{*} q_i: H^s(\partial \Omega) &\to H^t(\partial \Omega), \qquad s, t \in (-1,1) \label{eq:off-block-nuclear-list-2}\\
q_j \gammaDir{} A \gammaNeu{*} q_i: H^{s}(\partial \Omega) &\to H^t(\partial \Omega), \qquad s \in (0,1), t \in(-1,1) \label{eq:off-block-nuclear-list-3}\\
q_j\gammaNeu{} A \gammaDir{*} q_i: H^s(\partial \Omega) &\to H^{t}(\partial \Omega), \qquad s\in (-1, 1), t \in (-1, 0) \label{eq:off-block-nuclear-list-4}\\
q_j \gammaNeu{} A \gammaNeu{*} q_i: H^{s}(\partial \Omega) &\to H^{t}(\partial \Omega) \qquad s \in (0,1), t\in (-1,0).\label{eq:off-block-nuclear-list-last}
\end{align}

In each case, there exists a compact set $E_{ji} \subset U_j \times U_i$, $N \in \mathbb{N}$ and $C_{s,t,N} > 0$ such that
\begin{align}
\norm{T_{ji}}_{\Nuc}\leq C_{s,t,N} \norm{K}_{C^N(E_{ji})}
\end{align}
where $T_{ji}$ refers generically to the operator on the left-hand side of Equations \eqref{eq:off-block-nuclear-list-1}--\eqref{eq:off-block-nuclear-list-last}.
\end{lemma}
\begin{rem}
    \label{rem:off-diag-trace-zero}
If $T_{ji}$ is a trace-class off-diagonal block operator on a Hilbert space $\mathcal{H} = \bigoplus_i \mathcal{H}_i$, then automatically $\tr_{\mathcal{H}}(T_{ij}) = 0$.
\end{rem}
\begin{proof}[Proof of Lemma \ref{lem:trace-class-smoothing-ops}.]
Fix $i, j \in \{1, \ldots, k\}$ with $i \neq j$.
By modifying $K$ by a smooth cutoff function which is $\equiv 1$ near $U$, we can assume that $K \in C_c^\infty(U_j \times U_i)$.
Note that the bounded sets $U_i, U_j\subset \mathbb{R}^d$ can be considered as a subset of the torus $\mathbb{T}^d$, and therefore the kernel $K$ can be expanded in a Fourier series $K(x,y) = \sum_{m,n \in \mathbb{Z}^d} a_{m,n} e_m(x) e_n(y)$
with Fourier basis functions $e_m(x) = e^{2\pi \rmi  m \cdot x / L}$ for some $L>0$ such that $U_i, U_j \subset [-L/2,L/2]^d$.
Since $K$ is smooth, the Fourier coefficients $a_{m,n}$ are rapidly decaying: for all $M \in \mathbb{N}$, integration
by parts gives
\[
        |a_{mn}|
        \le C_M \|K\|_{C^{2M}(E_{ji})}
              \langle m\rangle^{-M}\langle n\rangle^{-M},
        \qquad
        \langle m\rangle=(1+|m|^2)^{1/2}.
\]

It follows that to estimate the trace norm of $K$, considered as a map between certain Sobolev spaces $\mathcal{H}_1 \to \mathcal{H}_2$,
it suffices to find bounds on $\norm{e_m}_{\mathcal{H}_2}, \norm{e_n}_{\mathcal{H}_1^*}$ which have at most polynomial growth in $m, n$.

We now proceed to give the required polynomial estimates.

For \ref{eq:off-block-nuclear-list-1}, the spaces are $H^t(\Omega_j)$ and $H^s(\Omega_i)^* = H^{-s}_{\overline{\Omega}_i}$.
In the former case the norm on $\Omega_i$ is estimated by the norm of any extension to $\mathbb{R}^d$:
\begin{align*}
\norm{e_m}_{H^t(\Omega_j)} \leq \norm{e_m \phi}_{H^t(\mathbb{R}^d)}
\lesssim \langle m\rangle^t
\end{align*}
for any cutoff function $\phi$ with $\phi \equiv 1$ near $\Omega_j$.
The space $H^{-s}_{\overline{\Omega}_i}$ is the space of distributions in $H^{-s}(\mathbb{R}^d)$ with support in $\overline{\Omega}_i$.
Hence
\begin{align*}
\norm{e_n}_{H^s(\Omega_i)^*} = \norm{e_n \chi_{\Omega_i}}_{H^{-s}(\mathbb{R}^d)} \lesssim \langle n\rangle^{-s}, s > -\frac{1}{2}.
\end{align*}
Here $\chi_{\overline{\Omega}_i}$ is the characteristic function of $\overline{\Omega}_i$, i.e. a sharp cutoff function.
Hence the regularity of $\chi_{\overline{\Omega}_i} e_n$ is less than $H^{1/2}$.

For (\ref{eq:off-block-nuclear-list-2}) -- (\ref{eq:off-block-nuclear-list-last}), the ranges of the Sobolev orders are determined by the continuity of the trace maps.
For any $s \in (0,1)$, we have
\begin{align*}
&\norm{\gammaDir{} e_n}_{H^{s}(\partial U)} \lesssim \norm{e_n}_{H^{s+1/2}(\widetilde{U})} \lesssim \langle n\rangle^{s+1/2},\\
&\norm{\gammaNeu{} e_n}_{H^{s-1}(\partial U)} \lesssim \norm{e_n}_{H^{s+1/2}(\widetilde{U})} + \norm{\Delta e_n}_{L^2(\widetilde{U})} \lesssim (\langle n\rangle^{s+1/2} + \langle n \rangle^2).
\end{align*}
The nuclearity ranges specified in the statement of the lemma then follow.
\end{proof}

\subsection{Fredholm determinant.}
The Fredholm determinant of an operator $1 + A$ on a Hilbert space $\mathcal{H}$ satisfies the standard bound
(\cite{simonNotesInfiniteDeterminants1977}, Equation (3.7)):
\begin{align}
    \label{eq:fredholm-det-bound}
|\operatorname{det}(1+A) - 1| \leq \norm{A}_{\Nuc} \exp \left(\norm{A}_{\Nuc(\mathcal{H})}+ 1 \right)
.
\end{align}

We also need the following Hilbert space generalization of \emph{Jacobi's formula}
\begin{lemma}
    \label{lem:fredholm-det-deriv}
If $A(\lambda)$ is a holomorphic trace-class perturbation of the identity,
then the Fredholm determinant $\det A(\lambda)$ is holomorphic and
\begin{align}
    \label{eq:fredholm-jacobi-formula}
    \frac{d}{d\lambda} \log \det A(\lambda) = \tr \left( A(\lambda)^{-1} \frac{d}{d\lambda} A(\lambda)\right).
\end{align}
\end{lemma}
\begin{proof}
The statement follows from Theorem 3.3 and Theorem 6.5 in \cite{simonNotesInfiniteDeterminants1977}.
Without loss of generality assume that we differentiate at $\lambda_0 = 0$.

We will write $A(\lambda) = I + B(\lambda)$ where $B(\lambda)$ is holomorphic and trace class.
By Theorem 3.3 loc.~cit., $\det(1 + B(\lambda))$ is holomorphic.

Now Taylor expand $B(\lambda) = B(0) + B'(0)\lambda + O(\lambda^2)$.
The Lipschitz estimate in Theorem 6.5 loc.~cit.~ yields, with $C>0$ independent of $\lambda$, in a neighborhood of $\lambda = 0$:
\begin{align*}
&|\det (1+ B(\lambda)) - \det(1 + B(0) + \lambda B'(0))|\\
\leq\ &\norm{B(\lambda) - B(0) - \lambda B'(0)}_n \exp(\Gamma_n (\norm{B(\lambda)}_n + \norm{B(0) + \lambda B'(0)}_n+1)^n)\\
\leq\ & C \lambda^2 
\end{align*}
This implies
\begin{align*}
&\det(1+ B(\lambda)) = \det(1 + B(0) + \lambda B'(0)) + O(\lambda^2)\\
=\ &\det(1+B(0)) \det(1 + \lambda (1+B(0))^{-1}B'(0)) + O(\lambda^2)\\
=\ &\det(1+B(0))(1 + \lambda \tr((1+B(0))^{-1}B'(0))) + O(\lambda^2).
\end{align*}
hence
\begin{align*}
\left.\frac{\mathrm{d}}{\mathrm{d}\lambda}\right|_{\lambda = 0} \det (1+ B(\lambda)) = \det(1+B(0)) \tr((1+B(0))^{-1}B'(0)).
\end{align*}
Then the logarithmic derivative is given by
\begin{align*}
    &\left.\frac{\mathrm{d}}{\mathrm{d}\lambda}\right|_{\lambda = 0} \log \det (A(\lambda))
= \left.\frac{\mathrm{d}}{\mathrm{d}\lambda}\right|_{\lambda = 0} \log \det (1 + B(\lambda))
= \operatorname{det}(1+B(0))^{-1} \left.\frac{\mathrm{d}}{\mathrm{d}\lambda}\right|_{\lambda = 0} \det(1 + B(\lambda))\\
=\ &\tr((1+B(0))^{-1} B'(0))
= \tr (A(0)^{-1} A'(0)).
\end{align*}
\end{proof}


To check that an operator family is holomorphic in the trace class, we use the following lemma.
\begin{lemma}
\label{lem:trace-class-holo}
Let $H$ be a Hilbert space, $D \subset \mathbb{C}$ an open domain and $D \to \mathcal{L}(H)$, $z \mapsto A(z)$ a holomorphic family of bounded operators.
Suppose that $\sup_{z \in K} \norm{A(z)}_{\Nuc(H)} < \infty$ for every compact $K \subset D$.
Then $A(z)$ is holomorphic in the $\Nuc(H)$-topology.
\end{lemma}
\begin{proof}
Since $A(z)$ is holomorphic in $\mathcal{L}(H)$, near any $z_0 \in D$ it possesses a power series expansion $A(z) = \sum_{n=0}^\infty A_n (z-z_0)^n$ with $A_n \in \mathcal{L}(H)$.
The coefficients $A_n$ are given by the Cauchy integral formula
\begin{align*}
A_n = \frac{1}{2\pi \rmi} \int_{|z-z_0|=r} \frac{A(z)}{(z-z_0)^{n+1}} \mathrm{d} z
\end{align*}
where we choose $r > 0$ such that the closed disk $\overline{B(z_0, r)}$ is contained in $D$.
Since $\norm{A(z)}_{\Nuc(H)}$ is uniformly bounded on the integration contour, the integral converges in the trace norm topology.
Hence $\norm{A_n}_{\Nuc(H)} \leq C r^{-n}$ for some $C > 0$.
Therefore the power series of $A(z)$ converges in the trace norm topology for $|z-z_0| < r$.
\end{proof}

\end{document}